\newcount\Comments
\def \confversion {1}
\def \twocols {1}

\ifx \confversion\undefined
\documentclass[11pt]{article}
\usepackage{fullpage,paper}
\title{Differentially Private Ordinary Least Squares}
\author{Or Sheffet\\ \texttt{osheffet@ualberta.ca}\\ University of Alberta\\Edmonton, AB Canada}
\usepackage[boxruled,linesnumbered]{algorithm2e}
\newcommand{\mycite}[1]{\cite{#1}}
\newcommand{\mycitet}[1]{\cite{#1}}
\newcommand{\myparagraph}[1]{\paragraph{#1}}
\else

\documentclass{article}

\usepackage{times}
\usepackage{graphicx} 
\usepackage{subcaption} 

\usepackage{natbib}

\usepackage{algorithm}
\usepackage{algorithmic}

\usepackage{hyperref}

\usepackage[accepted]{icml2017}

\icmltitlerunning{Differentially Private Ordinary Least Squares}

\usepackage{paper}
\usepackage{hyperref}       
\usepackage{url}            
\usepackage{amsfonts}       
\usepackage{nicefrac}       

\newcommand{\mycite}[1]{\cite{#1}}
\newcommand{\mycitet}[1]{\yrcite{#1}}
\newcommand{\myparagraph}[1]{\vspace{2mm}\noindent\textbf{#1}}
\usepackage[final]{pdfpages}
\fi

\usepackage{graphicx,color,hyperref,url,amsmath,amssymb}
\newcommand{\cut}[1]{}
\newcommand{\PDF}{\ensuremath{\mathsf{PDF}}}
\newcommand{\CDF}{\ensuremath{\mathsf{CDF}}}

\newcommand{\gaussian}{{\mathcal{N}}}
\newcommand{\T}{{\mathsf{T}}}
\newcommand{\mpi}{{+}}
\renewcommand{\vec}[1]{\pmb{#1}}
\newcommand{\hvec}[1]{\hat{\pmb{#1}}}
\newcommand{\tvec}[1]{\tilde{\pmb{#1}}}
\newcommand{\svmin}{\sigma_{\min}}
\newcommand{\svmax}{\sigma_{\max}}
\newcommand{\trace}{\mathrm{trace}}

\ifx \fullversion \undefined

\ifx \confversion \undefined
\newcommand{\omitproof}[1]{The proof is omitted and appears in the full version.}
\newcommand{\omitfromversion}[1]{}
\newcommand{\myvspace}[1]{\vspace{#1}}
\renewcommand{\paragraph}[1]{\smallskip\noindent\textbf{#1}}
\else
\newcommand{\omitproof}[1]{The proof is deferred to the supplementary material.}
\newcommand{\omitfromversion}[1]{}
\newcommand{\myvspace}[1]{}
\fi 
\else
\newcommand{\omitproof}[1]{{#1}}
\newcommand{\omitfromversion}[1]{{#1}}

\newcommand{\myvspace}[1]{}
\fi

\ifx \fullversion\undefined
\fi

\begin{document}
\ifdefined\confversion
	\twocolumn[
	\icmltitle{Differentially Private Ordinary Least Squares}
	
	
	
	\icmlsetsymbol{equal}{*}
	
	\begin{icmlauthorlist}
		\icmlauthor{Or Sheffet}{uofa}
	\end{icmlauthorlist}
	
	\icmlaffiliation{uofa}{Computing Science Dept., University of Alberta, Edmonton AB, Canada. This work was done when the author was at Harvard University, supported by NSF grant CNS-123723}
	
	\icmlcorrespondingauthor{Or Sheffet}{osheffet@ualberta.ca}
	
	\icmlkeywords{Differential Privacy, Ordinary Least Squares, t-Value, p-Value}
	
	\vskip 0.3in
	]
	
	
	
	\printAffiliationsAndNotice{}  

\fi

\begin{abstract}
Linear regression is one of the most prevalent techniques in machine learning;
however, it is also common to use linear regression for its \emph{explanatory} capabilities rather than label prediction. 
Ordinary Least Squares (OLS) is often used in statistics to establish a correlation between an attribute (e.g. gender) and a label (e.g. income)  in the presence of other (potentially correlated) features. 
OLS assumes a particular model that randomly generates the data, and derives \emph{$t$-values}~--- representing the likelihood of each real value to be the true correlation. 
Using $t$-values, OLS can release a \emph{confidence interval}, which is an interval on the reals that is likely to contain the true correlation; and when this interval does not intersect the origin, we can \emph{reject the null hypothesis} as it is likely that the true correlation is non-zero.
Our work aims at achieving similar guarantees on data under differentially private estimators.
First, we show that for well-spread data, the  Gaussian Johnson-Lindenstrauss Transform (JLT) gives a very good approximation of $t$-values; secondly, when JLT approximates Ridge regression (linear regression with $l_2$-regularization) we derive, under certain conditions, confidence intervals using the projected data; lastly, we derive, under different conditions, confidence intervals for the ``Analyze Gauss'' algorithm~\mycite{DworkTTZ14}.


\end{abstract}

\ifx\forappendix\undefined
\setcounter{page}{1}
\else
\bibliographystyle{alpha}
\bibliography{paper}
\appendix
\newpage
\setcounter{page}{14}
\fi

\section{Introduction}
\label{sec:intro}

Since the early days of differential privacy, its main goal was to design privacy preserving versions of existing techniques for data analysis. It is therefore no surprise that several of the first differentially private algorithms were machine learning algorithms, with a special emphasis on the ubiquitous problem of linear regression~\mycite{KasiviswanathanLNRS08, ChaudhuriMS11, KiferST12, BassilyST14}.
However, 
\emph{all} existing body of work on differentially private linear regression measures utility by bounding the distance between the linear regressor  found by the standard non-private algorithm and the regressor found by the  privacy-preserving algorithm. 
This is motivated from a machine-learning perspective, since bounds on the difference in the estimators translate to error bounds on prediction (or on the loss function). Such bounds are (highly) interesting and non-trivial, yet they are of little use in situations where one uses linear regression to establish correlations rather than predict labels. 

In the statistics literature, Ordinary Least Squares (OLS) is a technique that uses linear regression in order to infer the correlation between a variable and an outcome, especially in the presence of other factors. 
{And so, in this paper, we draw a distinction between ``linear regression,'' by which we refer to the machine learning technique of finding a specific estimator for a specific loss function; and ``Ordinary Least Squares,'' by which we refer to the statistical inference done assuming a specific model for generating the data and that uses linear regression.} 
Many argue that OLS is the most prevalent technique in social sciences~\cite{AgrestiF09}. 
Such works make no claim as to the labels of a new unlabeled batch of samples. Rather they aim to establish the existence of a strong correlation between the label and some feature. Needless to say, in such works, the privacy of individuals' data is a concern. 

In order to determine that a certain variable $x_j$ is positively (resp. negatively) correlated with an outcome $y$, OLS assumes a model where the outcome $y$ is a noisy version of a linear mapping of all variables: $y= \vec{\beta} \cdot \vec x + e$ (with $e$ denoting random Gaussian noise) for some predetermined and unknown $\vec \beta$. Then, given many samples $(\vec x_i, y_i)$ OLS establishes two things: (i) when fitting  a linear function to best predict $y$ from $\vec x$ \emph{over the sample} (via computing $\hat{\vec\beta} = \left( \sum_i \vec x_i \vec x_i^\T \right)^{-1}\left(\sum_i y_i \vec x_i\right)$) the coefficient $\hat \beta_j$ is positive (resp. negative); and (ii) \emph{inferring}, based on $\hat{\beta_j}$, that the true $\beta_j$ is likely to reside in $\mathbb{R}_{>0}$ (resp. $\mathbb{R}_{<0}$). In fact, the crux in OLS is by describing $\beta_j$ using a probability distribution over the reals, indicating where $\beta_j$ is likely to fall, derived by computing \emph{$t$-values}. These values take into account both the variance in the data as well as the variance of the noise $e$.\footnote{For example, imagine we run linear regression on a certain $(X,\vec y)$ which results in a vector $\hvec\beta$ with coordinates $\hat\beta_1 = \hat\beta_2 = 0.1$. Yet while the column $X_1$ contains many $1$s and $(-1)$s, the column $X_2$ is mostly populated with zeros. In such a setting, OLS gives that it is likely to have $\beta_1\approx 0.1$, whereas no such guarantees can be given for $\beta_2$.}
Based on this probability distribution one can define the \emph{$\alpha$-confidence interval}~--- an interval $I$ centered at $\hat\beta_j$ whose likelihood to contain $\beta_j$ is $1-\alpha$. Of particular importance is the notion of \emph{rejecting the null-hypothesis}, where the interval $I$ does not contain the origin, and so one is able to say with high confidence that $\beta_j$ is positive (resp. negative).  
Further details regarding OLS appear in Section~\ref{sec:preliminaries}.

In this work we give the \emph{first} analysis of statistical inference for OLS using differentially private estimators. We emphasize that the novelty of our work does not lie in the differentially-private algorithms, which are, as we discuss next, based on the Johnson-Lindenstrauss Transform (JLT) and on additive Gaussian noise and are already known to be differentially private~\mycite{BlockiBDS12, DworkTTZ14}. Instead, the novelty of our work lies in the analyses of the algorithms and in proving that the output of the algorithms is useful for statistical inference.


\ifx\fullversion\undefined
\myparagraph{The Algorithms.}
\else
\subsection{The Algorithms}
\label{subsec:DP_algorithm}
\fi
Our first algorithm (Algorithm~\ref{alg:private_JL}) is an adaptation of Gaussian JLT. Proving that this adaptation remains $(\epsilon,\delta)$-differentially private is straightforward (the proof appears in Appendix~\ref{apx_subsec:thmprivacy}). As described, the algorithm takes as input a parameter $r$ (in addition to the other parameters of the problem) that indicates the number of rows in the JL-matrix. Later, we analyze what should one set as the value of $r$. 
\ifx\confversion\undefined 
\begin{algorithm}[t]
\SetAlgoVlined
\SetKwInOut{Input}{input}\SetKwInOut{Output}{output}
\KwIn{A matrix $A \in \mathbb{R}^{n\times d}$ and a bound $B>0$ on the $l_2$-norm of any row in $A$.\\Privacy parameters: $\epsilon, \delta > 0$.\\Parameter $r$ indicating the number of rows in the resulting matrix.}
\BlankLine
Set $w$ s.t. $w^2 = \frac {8B^2}{\epsilon}\left(\sqrt {2r\ln(8/\delta)} + 2\ln(8/\delta)\right)$. 

Sample $Z \sim Lap(4B^2/\epsilon)$ and let $\svmin(A)$ denote the smallest singular value of $A$.

\eIf {$\svmin(A)^2 > w^2 + Z + \frac {4B^2\ln(1/\delta)}\epsilon$} 
{Sample a $(r\times n)$-matrix $R$ whose entries are i.i.d samples from a normal Gaussian. 
	
\KwRet{$RA$ and ``{\tt matrix unaltered}''.}}
{Let $A'$ denote the result of appending $A$ with the $d\times d$-matrix $wI_{d\times d}$.
	
Sample a $(r\times (n+d))$-matrix $R$ whose entries are i.i.d samples from a normal Gaussian.

\KwRet{$RA'$ and ``{\tt matrix altered}''.}}
\caption{Outputting a private Johnson-Lindenstrauss projection of a matrix.\label{alg:private_JL}}
\end{algorithm}
Our second algorithm is taken verbatim from the work of Dwork et al~\mycitet{DworkTTZ14}. 
\begin{algorithm}[b]
	\SetAlgoVlined
	\SetKwInOut{Input}{input}\SetKwInOut{Output}{output}
	\KwIn{A matrix $A \in \mathbb{R}^{n\times d}$ and a bound $B>0$ on the $l_2$-norm of any row in $A$.\\Privacy parameters: $\epsilon, \delta > 0$.}
	\BlankLine
	$N \leftarrow$ symmetric $(d\times d)$-matrix with upper triangle entries sampled i.i.d from $\gaussian\left(0, \tfrac {2B^4\ln(2/\delta)}{\epsilon^2}\right)$.\\
	\KwRet{$A^\T A + N$.}
	\caption{The ``Analyze Gauss'' Algorithm of Dwork et al~\cite{DworkTTZ14}.\label{alg:AG}}
\end{algorithm}
\else
\begin{algorithm}[ht]
	\caption{Outputting a private Johnson-Lindenstrauss projection of a matrix.\label{alg:private_JL}}
	\begin{algorithmic}
	\STATE{\bfseries Input:} {A matrix $A \in \mathbb{R}^{n\times d}$ and a bound $B>0$ on the $l_2$-norm of any row in $A$.\\Privacy parameters: $\epsilon, \delta > 0$.\\Parameter $r$ indicating the number of rows in the resulting matrix.}\vspace{1mm}
	\STATE Set $w$ s.t. $w^2 = \frac {8B^2}{\epsilon}\left(\sqrt {2r\ln(8/\delta)} + 2\ln(8/\delta)\right)$. 
	
	\STATE Sample $Z \sim Lap(4B^2/\epsilon)$ and let $\svmin(A)$ denote the smallest singular value of $A$.
	
	\IF {$\svmin(A)^2 > w^2 + Z + \frac {4B^2\ln(1/\delta)}\epsilon$} 
	\STATE{Sample a $(r\times n)$-matrix $R$ whose entries are i.i.d samples from a normal Gaussian.}
	\STATE{\bfseries return }{$RA$ and ``{\tt matrix unaltered}''.}
	\ELSE \STATE Let $A'$ denote the result of appending $A$ with the $d\times d$-matrix $wI_{d\times d}$.
		
		\STATE Sample a $(r\times (n+d))$-matrix $R$ whose entries are i.i.d samples from a normal Gaussian.
		
		\STATE {\bfseries return}{$RA'$ and ``{\tt matrix altered}''.}
	\ENDIF
	\end{algorithmic}
\end{algorithm}
Our second algorithm is taken verbatim from the work of Dwork et al~\mycitet{DworkTTZ14}. 
\begin{algorithm}[hb]
		\caption{``Analyze Gauss'' Algorithm of Dwork et al~\mycitet{DworkTTZ14}.\label{alg:AG}}
	\begin{algorithmic}
	\STATE{\bfseries Input:} {A matrix $A \in \mathbb{R}^{n\times d}$ and a bound $B>0$ on the $l_2$-norm of any row in $A$.\\Privacy parameters: $\epsilon, \delta > 0$.}\vspace{1mm}
	\STATE $N \leftarrow$ symmetric $(d\times d)$-matrix with upper triangle entries sampled i.i.d from $\gaussian\left(0, \tfrac {2B^4\ln(2/\delta)}{\epsilon^2}\right)$.
	\STATE{\bfseries return} {$A^\T A + N$.}
	\end{algorithmic}
\end{algorithm}
\fi
We deliberately focus on algorithms that approximate the $2^{\rm nd}$-moment matrix of the data and then run hypothesis-testing by post-processing the output, for two reasons. First, they enable sharing of data\footnote{Researcher $A$ collects the data and uses the approximation of the $2^{\rm nd}$-moment matrix to test some OLS hypothesis; but once the approximation is published researcher $B$ can use it to test for a completely different hypothesis.} and running unboundedly many hypothesis-tests. Since, we do not deal with OLS based on the private single-regression ERM algorithms~\mycite{ChaudhuriMS11, BassilyST14} as such inference requires us to use the Fisher-information matrix of the loss function~--- but these algorithms do not minimize a private loss-function but rather prove that outputting the minimizer of the perturbed loss-function is private. This means that differentially-private OLS based on these ERM algorithms requires us to devise new versions of these algorithms, making this a second step in this line of work... (After first understanding what we can do using existing algorithms.) We leave this approach~--- as well as performing private hypothesis testing using a PTR-type algorithm~\mycite{DworkL09} (output merely reject / don't-reject decision without justification), or releasing only relevant tests judging by their $p$-values~\mycite{DworkSZ15}~--- for future work.

\DeclareRobustCommand{\thmalgprivacystatement}{Algorithm~\ref{alg:private_JL} is $(\epsilon,\delta)$-differentially private.}

\ifx\fullversion\undefined
\myparagraph{Our Contribution and Organization.}
\else
\subsection{Our Contribution and Organization.} 
\label{subsec:our_work_and_related_work}
\fi
We analyze the performances of our algorithms on a matrix $A$ of the form $A= [X;\vec y]$, where each coordinate $y_i$ is generated according to the \emph{homoscedastic model} with Gaussian noise, which is a classical model in statistics. We assume the existence of a vector $\vec\beta$ s.t. for every $i$ we have $y_i = \vec\beta^\T \vec x_i + e_i$ and $e_i$ is sampled i.i.d from $\gaussian(0,\sigma^2)$.\footnote{This model may seem objectionable. Assumptions like the noise independence, $0$-meaned or sampled from a Gaussian distribution have all been called into question in the past. Yet due to the prevalence of this model we see fit to initiate the line of work on differentially private Least Squares with this Ordinary model.}

We study the result of running Algorithm~\ref{alg:private_JL} on such data in the two cases: where $A$ wasn't altered by the algorithm and when $A$ was appended by the algorithm. In the former case, Algorithm~\ref{alg:private_JL} boils down to projecting the data under a Gaussian JLT. Sarlos~\mycitet{Sarlos06} has already shown that the JLT is useful for linear regression, yet his work bounds the $l_2$-norm of the difference between the estimated regression before and after the projection. Following Sarlos' work, other works in statistics have analyzed compressed linear regression~\mycite{ZhouLW07, PilanciW14a, PilanciW14b}.
However, none of these works give confidence intervals based on the projected data, presumably for three reasons. Firstly, these works are motivated by computational speedups, and so they use fast JLT as opposed to our analysis which leverages on the fact that our JL-matrix is composed of i.i.d Gaussians. Secondly, the focus of these works is not on OLS but rather on newer versions of linear regression, such as Lasso or when $\vec\beta$ lies in some convex set. Lastly, it is evident that the smallest confidence interval is derived from the data itself. Since these works do not consider privacy applications, (actually, \mycite{ZhouLW07, PilanciW14a} do consider privacy applications of the JLT, but quite different than differential privacy) they assume the analyst has access to the data itself, and so there was no need to give confidence intervals for the projected data. Our analysis is therefore the first, to the best of our knowledge, to derive $t$-values~--- and therefore achieve all of the rich expressivity one infers from $t$-values, such as confidence bounds and null-hypotheses rejection~--- for OLS estimations \emph{without having access to $X$ itself}. We also show that, under certain conditions, the sample complexity for correctly rejecting the null-hypothesis increases from a certain bound $N_0$ (without privacy) to a bound of $N_0 + \tilde O(\sqrt{N_0}\cdot \kappa(\tfrac 1 n A^\T A)/\epsilon)$ with privacy (where $\kappa(M)$ denotes the condition number of the matrix $M$.) This appears in Section~\ref{sec:OLS_for_JL_data}.

In Section~\ref{sec:ridge_regression} we analyze the case Algorithm~\ref{alg:private_JL} does append the data and the JLT is applied to $A'$. In this case, solving the linear regression problem on the projected $A'$ approximates the solution for \emph{Ridge Regression}~\mycite{Tikhonov63, HoerlK70}. In Ridge Regression we aim to solve $\min_{\vec z}\left(\sum_i (y_i - \vec z^\T \vec x_i)^2 + w^2\|\vec z\|^2 \right)$, which means we penalize vectors whose $l_2$-norm is large. 
In general, it is not known how to derive $t$-values from Ridge regression, and the literature on deriving confidence intervals solely from Ridge regression is virtually non-existent. Indeed, prior to our work there was no need for such calculations, as access to the data was (in general) freely given, and so deriving confidence intervals could be done by appealing back to OLS. We too are unable to derive approximated $t$-values in the general case, but under additional assumptions about the data~--- which admittedly depend in part on $\|\vec{\beta}\|$ and so cannot be verified solely from the data~--- we show that solving the linear regression problem on $RA'$ allows us to give confidence intervals for $\beta_j$, thus correctly determining the correlation's sign. 

In Section~\ref{sec:analyze_gauss_short} we discuss the ``Analyze Gauss'' algorithm~\mycite{DworkTTZ14} that outputs a noisy version of a covariance of a given matrix using additive noise rather than multiplicative noise. Empirical work~\cite{XiKI11} shows that Analyze Gauss's output might be non-PSD if  the input has small singular values, and this results in truly bad regressors. Nonetheless, under additional conditions (that imply that the output is PSD), we derive confidence bounds for Dwork et al's ``Analyze Gauss'' algorithm. Finally, in Section~\ref{sec:experiment} we experiment with the heuristic of computing the $t$-values directly from the outputs of Algorithms~\ref{alg:private_JL} and~\ref{alg:AG}. We show that Algorithm~\ref{alg:private_JL} is more ``conservative'' than Algorithm~\ref{alg:AG} in the sense that it tends to not reject the null-hypothesis until the number of examples is large enough to give a very strong indication of rejection. In contrast, Algorithm~\ref{alg:AG} may wrongly rejects the null-hypothesis even when it is true.

\myparagraph{Discussion.} Some works have already looked at the intersection of differentially privacy and statistics~\cite{DworkL09, Smith11, ChaudhuriH12, DuchiJW13, DworkSZ15} (especially focusing on robust statistics and rate of convergence). But only a handful of works studied the significance and power of hypotheses testing under differential privacy, without arguing that the noise introduced by differential privacy vanishes asymptotically~\cite{VuS09, UhlerSF13, WangLK15, RogersVLG16}. These works are experimentally promising, yet they (i) focus on different statistical tests (mostly Goodness-of-Fit and Independence testing), (ii) are only able to prove results for the case of simple hypothesis-testing (a single hypothesis) with an efficient data-generation procedure through repeated simulations~--- a cumbersome and time consuming approach. In contrast, we deal with a composite hypothesis (we simultaneously reject all $\vec{\beta}$s with $sign(\beta_j)\neq sign(\hat\beta_j)$) by altering the confidence interval (or the critical region).

One potential reason for avoiding confidence-interval analysis for differentially private hypotheses testing is that it does involve re-visiting existing results. Typically, in statistical inference the sole source of randomness lies in the underlying model of data generation, whereas the estimators themselves are a deterministic function of the dataset. In contrast, differentially private estimators are inherently random \emph{in their computation}. Statistical inference that considers \emph{both} the randomness in the data and the randomness in the computation is highly uncommon, and this work, to the best of our knowledge, is the first to deal with randomness in OLS hypothesis testing. We therefore strive in our analysis to separate the two sources of randomness~--- as in classic hypothesis testing, we use $\alpha$ to denote the bound on any bad event that depends solely on the homoscedastic model, and use $\nu$ to bound any bad event that depends on the randomized algorithm.\footnote{Or any randomness in generating the feature matrix $X$ which standard OLS theory assumes to be fixed, see Theorems~\ref{thm:baseline_rejecting_nh} and~\ref{thm:JL_simple_case_rejecting_null_hypothesis}.} (Thus, any result which is originally of the form ``$\alpha$-reject the null-hypothesis'' is now converted into a result ``($\alpha+\nu$)-reject the null hypothesis''.)

\myvspace{-8pt}
\section{Preliminaries and OLS Background}
\label{sec:preliminaries}

\myparagraph{Notation.} Throughout this paper, we use $lower$-case letters to denote scalars (e.g., $y_i$ or $e_i$); $\pmb{bold}$ characters to denote vectors; and UPPER-case letters to denote matrices. The $l$-dimensional all zero vector is denoted $\vec 0_l$, and the $l\times m$-matrix of all zeros is denoted $0_{l\times m}$. We use $\vec e$ to denote the specific vector $\vec y-X\vec\beta$ in our model; and though the reader may find it a bit confusing but hopefully clear from the context~--- we also use $\vec e_j$ and $\vec e_k$ to denote elements of the natural basis (unit length vector in the direction of coordinate $j$ or $k$). We use $\epsilon,\delta$ to denote the privacy parameters of Algorithms~\ref{alg:private_JL} and~\ref{alg:AG}, and use $\alpha$ and $\nu$ to denote confidence parameters (referring to bad events that hold w.p. $\leq \alpha$ and $\leq \nu$ resp.) based on the homoscedastic model or the randomized algorithm resp. We also stick to the notation from Algorithm~\ref{alg:private_JL} and use $w$ to denote the positive scalar for which $w^2 =  \frac {8B^2}\epsilon\left(\sqrt{2r\ln(8/\delta)} + \ln(8/\delta)\right) $ throughout this paper. We use standard notation~for SVD composition of a matrix ($M= U\Sigma V^\T$), its singular values and its Moore-Penrose inverse~($M^\mpi$). 
\DeclareRobustCommand{\linearAlgebraPreliminaries}{
\myparagraph{Linear Algebra and Pseudo-Inverses.} Given a matrix $M$ we denote its SVD as $M = U S V^\T$ with $U$ and $V$ being orthonormal matrices and $S$ being a non-negative diagonal matrix whose entries are the singular values of $M$. We use $\svmax(M)$ and $\svmin(M)$ to denote the largest and smallest singular value resp. Despite the risk of confusion, we stick to the standard notation of using $\sigma^2$ to denote the variance of a Gaussian, and use $\sigma_{j}(M)$ to denote the $j$-th singular value of $M$. We use $M^\mpi$ to denote the Moore-Penrose inverse of $M$, defined as $M^\mpi = V S^{-1} U^T$ where $S^{-1}$ is a matrix with $S^{-1}_{j,j} = 1/S_{j,j}$ for any $j$ s.t. $S_{j,j}>0$. \omitfromversion{It is known that when $M\in \mathbb{R}^{a\times b}$ with $a\geq b$ and $b =rank(M)$, then $M^\mpi = (M^\T M)^{-1} M^\T$ (and when $a=b$ then $M^\mpi = M^{-1}$). In such a case it holds that $M^\mpi (M^{\mpi})^\T = (M^\T M)^{-1}$, and that $M^\mpi M = I_{b\times b}$. The matrix $P_U \stackrel{\rm def} = MM^\mpi$ is a projection matrix that fixes any vector $\vec u \in colspan(U)$ and nullifies any vector in $(colspan(U))^\perp$.  A $m\times m$-matrix $M$ is said to be positive semi-definite (PSD) if $\vec x^\T M \vec x \geq 0$ for any $\vec x\in\mathbb{R}^m$, and positive definite if $\vec x^\T M \vec x > 0$ for any $\vec x\in\mathbb{R}^m$. For two PSD matrices $M$ and $N$ we use the notation $M\preceq N$ to denote the fact that $\vec x^\T M \vec x \leq \vec x^\T N \vec x$ for any $\vec x$. For a given matrix, $\|M\|$ denotes the spectral norm ($=\svmax(M)$) and $\|M\|_F$ denotes the Frobenious norm $(\sum_{j,k} M_{j,k}^2)^{1/2}$. It is known that $\|M\|_F^2= \trace(M^\T M) = \sum_{j}\sigma_j^2(M)$.}
}

\myparagraph{The Gaussian distribution.} 
\DeclareRobustCommand{\gaussianfull}{
	A univariate Gaussian $\gaussian(\mu, \sigma^2)$ denotes the Gaussian distribution whose mean is $\mu$ and variance $\sigma^2$, with $\PDF(x) = (\sqrt{2\pi\sigma^2})^{-1}\exp(-\tfrac {x-\mu}{2\sigma^2})$. Standard concentration bounds on Gaussians give that $\Pr[ x> \mu+2\sigma\sqrt{\ln(1/\nu)}] < \nu$ for any $\nu\in(0,\tfrac 1 e)$. A multivariate Gaussian $\gaussian(\vec\mu,\Sigma)$ for some positive semi-definite $\Sigma$ denotes the multivariate Gaussian distribution where the mean  of the $j$-th coordinate is the $\mu_j$ and the co-variance between coordinates $j$ and $k$ is $\Sigma_{j,k}$. The $\PDF$ of such Gaussian is defined only on the subspace $colspan(\Sigma)$, where for every $x\in colspan(\Sigma)$ we have $\PDF(\vec x) = \left((2\pi)^{rank(\Sigma)}\cdot \tilde\det(\Sigma)\right)^{-1/2}\exp\left(-\tfrac 1 2 (\vec x-\vec \mu)^\T \Sigma^\mpi (\vec x-\vec\mu)\right)$ and $\tilde\det(\Sigma)$ is the multiplication of all non-zero singular values of $\Sigma$. A matrix Gaussian distribution denoted $\gaussian(M_{a\times b}, U, V)$ has mean $M$, variance $U$ on its rows and variance $V$ on its columns. For full rank $U$ and $V$ it holds that $\PDF_{\gaussian(M,U,V)}(X) = (2\pi)^{-ab/2}(\det(U))^{-b/2}(\det(V))^{-a/2}\cdot\exp(-\tfrac 1 2 \trace\left(V^{-1} (X-M)^\T U^{-1} (X-M) \right) )$. In our case, we will only use matrix Gaussian distributions with $\gaussian(M_{a\times b}, I_{a\times a}, V)$ and so each row in this matrix is an i.i.d sample from a $b$-dimensional multivariate Gaussian $\gaussian((M)_{j\to}, V)$.

We will repeatedly use the rules regarding linear operations on Gaussians. That in, for any $c$, it holds that $c\gaussian(\mu,\sigma^2) = \gaussian(c\cdot\mu,c^2\sigma^2)$. For any $C$ it holds that $C\cdot \gaussian(\vec \mu, \Sigma) = \gaussian(C\vec\mu, C \Sigma C^\T)$. And for any $C$ is holds that $\gaussian(M,U,V) \cdot C = \gaussian (MC, U, C^\T V C)$. In particular, for any $\vec c$ (which can be viewed as a $b\times 1$-matrix) it holds that $\gaussian(M,U,V) \cdot\vec c = \gaussian(M\vec c, U, \vec c^T V \vec c) = \gaussian(M\vec c, \vec c^T V \vec c\cdot U)$. 

We will also require the following proposition. 
\DeclareRobustCommand{\propositionstatement}{Given $\sigma^2,\lambda^2$ s.t. $1 \leq \frac{\sigma^2}{\lambda^2} \leq c^2$ for some constant $c$, let $X$ and $Y$ be two random Gaussians s.t. $X\sim\gaussian(0,\sigma^2)$ and $Y\sim\gaussian(0,\lambda^2)$. It follows that $\tfrac 1 c\PDF_{ Y}(x) \leq \PDF_X(x) \leq c\PDF_{cY}(x)$ for any $x$.}
\begin{proposition}
	\label{pro:ratio_close_gaussians} \propositionstatement
\end{proposition}
\cut{
	\omitproof{\begin{proof}
			\begin{align*}
			& \frac {\PDF_X(x)}{\PDF_{cY}(x)} = \sqrt{ \frac{c^2\lambda^2}{\sigma^2}}\cdot \frac{\exp(-\tfrac {x^2}{2\sigma^2})}{\exp(-\tfrac {x^2}{2c^2\lambda^2})} \leq c\cdot \exp(\frac {x^2}2 (\frac {1} {c^2\lambda^2} - \frac 1 {\sigma^2})) \leq c\cdot \exp(0)=c\cr
			& \frac {\PDF_X(x)}{\PDF_{Y}(x)} = \sqrt{ \frac{\lambda^2}{\sigma^2}}\cdot \frac{\exp(-\tfrac {x^2}{2\sigma^2})}{\exp(-\tfrac {x^2}{2\lambda^2})} \geq c^{-1} \exp( \tfrac {x^2}2(\tfrac 1 {\lambda^2}-\tfrac 1 {\sigma^2}) ) \geq c^{-1} \exp(0) = c^{-1}
			\end{align*}
		\end{proof}}
	}
\begin{corollary}
		\label{cor:probabilities_ratio_close_gaussians}
		Under the same notation as in Proposition~\ref{pro:ratio_close_gaussians}, for any set $S\subset \mathbb{R}$ it holds that
		\ifx\twocols\undefined 
		\[ \tfrac 1 c \Pr_{x\leftarrow Y}[x\in S] \leq \Pr_{x\leftarrow X}[x\in S] \leq c\Pr_{x\leftarrow cY}[x\in S] = c\Pr_{x\leftarrow Y}[x\in S/c]\]
		\else 
		$\tfrac 1 c \Pr_{x\leftarrow Y}[x\in S] \leq \Pr_{x\leftarrow X}[x\in S] \leq c\Pr_{x\leftarrow cY}[x\in S] = c\Pr_{x\leftarrow Y}[x\in S/c]$
		\fi 
\end{corollary}
}
\ifx\fullversion\undefined
A univariate Gaussian $\gaussian(\mu, \sigma^2)$ denotes the Gaussian distribution whose mean is $\mu$ and variance $\sigma^2$. Standard concentration bounds on Gaussians give that $\Pr[ x> \mu+2\sigma\sqrt{\ln(2/\nu)}] < \nu$ for any $\nu\in(0,\tfrac 1 e)$. A multivariate Gaussian $\gaussian(\vec\mu,\Sigma)$ for some positive semi-definite $\Sigma$ denotes the multivariate Gaussian distribution where the mean  of the $j$-th coordinate is the $\mu_j$ and the covariance between coordinates $j$ and $k$ is $\Sigma_{j,k}$. The $\PDF$ of such Gaussian is defined only on the subspace $colspan(\Sigma)$. A matrix Gaussian distribution, denoted $\gaussian(M_{a\times b}, I_{a\times a}, V)$ has mean $M$, independence among its rows and variance $V$ for each of its columns. We also require the following property of Gaussian random variables: Let $X$ and $Y$ be two random Gaussians s.t. $X\sim\gaussian(0,\sigma^2)$ and $Y\sim\gaussian(0,\lambda^2)$ where $1 \leq \frac{\sigma^2}{\lambda^2} \leq c^2$ for some $c$, then for any $S\subset\R$ we have $\tfrac 1 c \Pr_{x\leftarrow Y}[x\in S] \leq \Pr_{x\leftarrow X}[x\in S] \leq c\Pr_{x\leftarrow Y}[x\in S/c]$ (see Proposition~\ref{pro:ratio_close_gaussians}).
\else
\gaussianfull
\fi

\myparagraph{Additional Distributions.} We denote by $Lap(\sigma)$ the \emph{Laplace distribution} whose mean is $0$ and variance is $2\sigma^2$.
The \emph{$\chi^2_k$-distribution}, where $k$ is referred to as the degrees of freedom of the distribution, is the distribution over the $l_2$-norm squared of the sum of $k$ independent normal Gaussians. That is, given i.i.d $X_1, \ldots, X_k \sim \gaussian(0,1)$ it holds that $\vec\zeta \stackrel{\rm def}= (X_1, X_2, \ldots, X_k) \sim\gaussian(\vec 0_k, I_{k\times k})$, and  $\|\vec\zeta\|^2\sim \chi^2_k$. Existing tail bounds on the $\chi^2_k$ distribution~\mycite{LaurentM00} give that
\ifx\fullversion\undefined
$\Pr\left[ \|\vec\zeta\|^2 \in (\sqrt{k} \pm \sqrt{2\ln(2/\nu)})^2 \right] \geq 1-\nu$.
\else
\begin{eqnarray*}
\Pr\left[ \|\vec\zeta\|^2 \in (\sqrt{k} \pm \sqrt{2\ln(2/\nu)})^2 \right] \geq 1-\nu
\end{eqnarray*}
\fi
\DeclareRobustCommand{\tdistributionfull}{
The \emph{$T_k$-distribution}, where $k$ is referred to as the degrees of freedom of the distribution, denotes the distribution over the reals created by \emph{independently} sampling $Z \sim \gaussian(0,1)$ and $\|\zeta\|^2 \sim\chi^2_k$, and taking the quantity $\frac Z {\sqrt{\|\zeta\|^2/k}}$. Its $\PDF$ is given by $\PDF_{T_k}(x) \propto \left(1 + \tfrac {x^2} k \right)^{-\tfrac {k+1} 2}$. It is a known fact that as $k$ increases, $T_k$ becomes closer and closer to a normal Gaussian. The $T$-distribution is often used to determine suitable bounds on the rate of converges, as we illustrate in Section~\ref{apx_subsec:OLS_background}. As the $T$-distribution is heavy-tailed, existing tail bounds on the $T$-distribution (which are of the form: if $\tau_\nu = C\sqrt{k((1/\nu)^{2/k}-1)}$ for some constant $C$ then $\int_{\tau_\nu}^\infty \PDF_{T_k}(x)dx < \nu$)  are often cumbersome to work with. Therefore, in many cases in practice, it common to assume $\nu = \Theta(1)$ (most commonly, $\nu = 0.05$) and use existing tail-bounds on normal Gaussians.}
\ifx\fullversion\undefined
The $T_k$-distribution, where $k$ is referred to as the degrees of freedom of the distribution, denotes the distribution over the reals created by \emph{independently} sampling $Z \sim \gaussian(0,1)$ and $\|\zeta\|^2 \sim\chi^2_k$, and taking the quantity $Z/{\sqrt{\|\zeta\|^2/k}}$. It is a known fact that $T_k \stackrel{k\to\infty}\rightarrow \gaussian(0,1)$, thus it is a common practice to apply Gaussian tail bounds to the $T_k$-distribution when $k$ is sufficiently large.
\else
\tdistributionfull
\fi

\myparagraph{Differential Privacy.} In this work, we deal with input in the form of a $n\times d$-matrix with each row bounded by a $l_2$-norm of $B$. Two inputs $A$ and $A'$ are called \emph{neighbors} if they differ on a single row. 
\begin{definition}[\mycite{DworkKMMN06}]
\label{def:privacy}
An algorithm $\textsf{ALG}$ which maps $(n\times d)$-matrices into some range $\mathcal{R}$ is \emph{$(\epsilon,\delta)$-differential privacy} it holds that
\ifx\fullversion\undefined
$\Pr[\mathsf{ALG}(A) \in \mathcal{S}] \leq e^\epsilon \Pr[\mathsf{ALG}(A') \in \mathcal{S}] + \delta$
\else
\[\Pr[\mathsf{ALG}(A) \in \mathcal{S}] \leq e^\epsilon \Pr[\mathsf{ALG}(A') \in \mathcal{S}] + \delta\]\fi 
for all neighboring inputs $A$ and $A'$ and all subsets $\mathcal{S}\subset\mathcal{R}$. 
\end{definition}
\DeclareRobustCommand{\DPfull}{
It is known~\mycite{DworkMNS06} that if $\mathsf{ALG}$ outputs a vector in $\R^d$ such that for any $A$ and $A'$ it holds that $\|\mathsf{ALG}(A)-\mathsf{ALG}(A')\|_1 \leq B$, then adding Laplace noise $Lap(1/\epsilon)$ to each coordinate of the output of $\mathsf{ALG}(A)$ satisfies $\epsilon$-differential privacy. Similarly, \mycitet{DworkMNS06} showed that if for any neighboring $A$ and $A'$ it holds that $\|\mathsf{ALG}(A)-\mathsf{ALG}(A')\|_2^2 \leq \Delta^2$ then adding Gaussian noise $\gaussian(0, \Delta^2\cdot\tfrac{2\ln(2/\delta)}{\epsilon^2})$ to each coordinate of the output of $\mathsf{ALG}(A)$ satisfies $(\epsilon,\delta)$-differential privacy.

Another standard result~\mycite{DworkKMMN06} gives that the composition of the output of a $(\epsilon_1,\delta_1)$-differentially private algorithm with the output of a $(\epsilon_2,\delta_2)$-differentially private algorithm results in a $(\epsilon_1+\epsilon_2,\delta_1+\delta_2)$-differentially private algorithm.}
\ifdefined\fullversion \DPfull \fi

%
\myparagraph{Background on OLS.}
\label{subsec:background_OLS}
For the unfamiliar reader, we give here a \emph{very} brief overview of the main points in OLS. Further details, explanations and proofs appear in Section~\ref{apx_subsec:OLS_background}.

We are given $n$ observations $\{(\vec x_i, y_i)\}_{i=1}^n$ where  $\forall i, \vec x_i\in\mathbb{R}^p$ and $y_i \in \mathbb{R}$. We assume the existence of $\vec\beta \in \mathbb{R}^p$ s.t. the label $y_i$ was derived by $y_i = \vec\beta^\T\vec x_i + e_i$ where $e_i \sim \gaussian(0,\sigma^2)$ independently (also known as the  homoscedastic Gaussian model). We use the matrix notation where $X$ denotes the $(n\times p)$- feature matrix and $\vec y$ denotes the labels. We assume $X$ has full rank.

The parameters of the model are therefore $\vec \beta$ and $\sigma^2$, which we set to discover. To that end, we minimize $\min_{\vec z} \| \vec y - X\vec z\|^2$ and have
\begin{eqnarray}
\hvec\beta  = \resizebox{0.3\textwidth}{!}{$(X^\T X)^{-1} X^\T \vec y = (X^\T X)^{-1} X^\T (X\vec\beta + \vec e) $}= \vec\beta + X^\mpi\vec e\label{eq:hat_beta}\\
\vec\zeta = \resizebox{0.25\textwidth}{!}{$\vec y-X\hat{\vec\beta} = (X\vec\beta + \vec e) - X(\vec\beta + X^\mpi\vec e)$} = (I - XX^\mpi)\vec e \label{eq:zeta}
\end{eqnarray}
And then for any coordinate $j$ the \emph{$t$-value}, which is the quantity
\ifx\fullversion\undefined
$  t(\beta_j) \stackrel{\rm def}{=} \frac {\hat\beta_j -\beta_j} {\sqrt{(X^\T X)^{-1}_{j,j}}\cdot  \frac{\|\vec\zeta\|}{\sqrt{n-p}}}$,
\else
\[ t_{\hat\beta_j}(\beta_j) \stackrel{\rm def}{=} \frac {\hat\beta_j -\beta_j} {\sqrt{(X^\T X)^{-1}_{j,j}}\cdot  \frac{\|\vec\zeta\|}{\sqrt{n-p}}} = \frac {\hat\beta_j -\beta_j} {\sigma\sqrt{(X^\T X)^{-1}_{j,j}}} \Big/  \frac{\|\vec\zeta\|}{\sigma\sqrt{n-p}}
\]
\fi
is distributed according to $T_{n-p}$-distribution. I.e., 
\ifx\fullversion\undefined
$\Pr\left[\hvec\beta\textrm{ and }\vec\zeta\textrm{ satisfying } 
t(\beta_j)
\in S \right] = \int\limits_S \PDF_{T_{n-p}}(x)dx$
\else
\[ \Pr\left[\hvec\beta\textrm{ and }\vec\zeta\textrm{ satisfying } \frac {\hat\beta_j -\beta_j} {\sqrt{(X^\T X)^{-1}_{j,j}}\cdot  \frac{\|\vec\zeta\|}{\sqrt{n-p}}} \in S \right] = \int_S \PDF_{T_{n-p}}(x)dx\]
\fi
for any measurable $S\subset \mathbb{R}$. Thus $t(\beta_j)$ describes the likelihood of any $\beta_j$~--- for any $z\in\mathbb{R}$ we can now give an estimation of how likely it is to have $\beta_j=z$ (which is $\PDF_{T_{n-p}}(t(z))$), and this is known as \emph{$t$-test} for the value $z$.
In particular, given $0<\alpha < 1$, we denote $c_\alpha$ as the number for which the interval $(- c_\alpha, c_\alpha)$ contains a probability mass of $1-\alpha$ from the $T_{n-p}$-distribution. And so we derive a corresponding \emph{confidence interval} $I_\alpha$ centered at $\hat\beta_j$ where $\beta_j \in I_\alpha$ with confidence of level of $1-\alpha$. \omitfromversion{Using tail bounds on the $T_{n-p}$-distribution~\mycite{Soms76}, we have that the length of the interval is $|I_\alpha| = O\left(\sqrt{(X^\T X)^{-1}_{j,j} \cdot \tfrac {\|\vec\zeta\|^2}{n-p}} \cdot \sqrt{(n-p)((\tfrac 1 \alpha)^{\tfrac 2 {n-p-1}}-1)} \right)$. Furthermore, since it is known that as the number of degrees of freedom of a $T$-distribution tends to infinity then the $T$-distribution becomes close to a normal Gaussian, it is common to use the $\PDF$ of a normal Gaussian instead. I.e., denote $\tau_\alpha$ as the number of which $\int_{\tau_\alpha}^\infty \PDF_{\gaussian(0,1)}(x)dx = \tfrac \alpha 2$, then $I_\alpha = \beta_j \pm \tau_\alpha \sqrt{(X^\T X)^{-1}_{j,j} \cdot \tfrac {\|\vec\zeta\|^2}{n-p}}$.}

Of particular importance is the quantity $t_0 \stackrel{\rm def} = t(0) = \frac {\hat\beta_j\sqrt{n-p}} {\|\vec\zeta\|\sqrt{(X^\T X)^{-1}_{j,j}}}$,since if there is no correlation between $x_j$ and $y$ then the likelihood of seeing $\hat{\beta_j}$ depends on the ratio of its magnitude to its standard deviation. As mentioned earlier, since $T_k \stackrel{k\to\infty}{\rightarrow} \gaussian(0,1)$, then rather than viewing this $t_0$ as sampled from a $T_{n-p}$-distribution, it is common to think of $t_0$ as a sample from a normal Gaussian $\gaussian(0,1)$. This allows us to associate $t_0$ with a $p$-value, estimating the event ``$\beta_j$ and $\hat\beta_j$ have different signs.''
\ifx\fullversion\undefined
\else
Formally, we define $p_0 = \int_{|t_0|}^{\infty} \frac 1 {\sqrt{2\pi}} e^{-x^2/2}dx$. It is common to reject the null hypothesis when $p_0$ is sufficiently small (typically, below $0.05$).\footnote{Indeed, it is more accurate to associate with $t_0$ the value $\int_{|t_0|}^\infty \PDF_{T_{n-p}}(x)dx$ and check that this value is $<\alpha$. However, as most uses take $\alpha$ to be a constant (often $\alpha = 0.05$), asymptotically the threshold we get for rejecting the null hypothesis are the same.}
\fi
Specifically, given $\alpha \in (0,1/2)$, we \emph{$\alpha$-reject the null hypothesis} if $p_0 <\alpha$. Let $\tau_\alpha$ be the number s.t. $\Phi(\tau_\alpha) = \int_{\tau_\alpha}^{\infty} \tfrac 1 {\sqrt{2\pi}} e^{-x^2/2}dx = \alpha$. This means we $\alpha$-reject the null hypothesis when  $|t_0| > \tau_\alpha$.
We now lower bound the number of i.i.d sample points needed in order to $\alpha$-reject the null hypothesis. This bound is our basis for comparison between standard OLS and the differentially private version.\footnote{Theorem~\ref{thm:baseline_rejecting_nh} also illustrates how we ``separate'' the two sources of privacy. In this case, $\nu$ bounds the probability of bad events that depend to sampling the rows of $X$, and $\alpha$ bounds the probability of a bad event that depends on the sampling of the $\vec y$ coordinates.}
\begin{theorem}
\label{thm:baseline_rejecting_nh}
Fix any positive definite matrix $\Sigma\in\mathbb{R}^{p\times p}$ and any $\nu \in (0,\tfrac 1 2)$. Fix parameters $\vec\beta \in \mathbb{R}^p$ and $\sigma^2$ and a coordinate $j$ s.t. $\beta_j \neq 0$. Let $X$ be a matrix whose $n$ rows are i.i.d samples from $\gaussian(\vec 0, \Sigma)$, and $\vec y$ be a vector where $y_i - (X\vec\beta)_i$ is sampled i.i.d from $\gaussian(0,\sigma^2)$. Fix $\alpha \in (0,1)$. 
Then w.p. $\geq 1-\alpha-\nu$ we have that OLS's $(1-\alpha)$-confidence interval has length $O(c_\alpha\sqrt{\sigma^2 / (n\svmin(\Sigma)}))$ provided $n \geq C_1 (p+\ln(1/\nu))$ for some sufficiently large constant $C_1$. Furthermore, there exists a constant $C_2$ such that w.p. $\geq 1-\alpha- \nu$ OLS (correctly) rejects the null hypothesis provided
\ifx\fullversion\undefined
$n \geq \max\left\{C_1 (p+\ln(1/\nu)), ~~p + C_2 \frac {\sigma^2}{\beta_j^2} \cdot \frac {c_\alpha^2+\tau_\alpha^2} {\svmin(\Sigma)} \right\}$,
\else
\[ n \geq \max\left\{C_1 (p+\ln(1/\nu)), ~~ C_2 \frac {\sigma^2}{\beta_j^2} \cdot \frac {c_\alpha^2+\tau_\alpha^2} {\svmin(\Sigma)} \right\}\] 
\fi
where $c_\alpha$ is the number for which $\int_{-c_\alpha}^{c_\alpha} \PDF_{T_{n-p}}(x)dx = 1-\alpha$. 
\end{theorem}
\myvspace{-8pt}
\ifdefined\fullversion
\section{Ordinary Least Squares over Projected Data}
\else
\section{OLS over Projected Data}
\fi
\label{sec:OLS_for_JL_data}

In this section we deal with the output of Algorithm~\ref{alg:private_JL} in the special case where  Algorithm~\ref{alg:private_JL} outputs \texttt{matrix unaltered} and so we work with $RA$.

To clarify, the setting is as follows. We denote $A = [X;\vec y]$ the column-wise concatenation of the $(n\times (d-1))$-matrix $X$ with the $n$-length vector $\vec y$. (Clearly, we can denote any column of $A$ as $\vec y$ and any subset of the remaining columns as the matrix $X$.) 
We therefore denote the output $RA = [RX; R\vec y]$ and for simplicity we denote $M=RX$ and $p=d-1$. 
We denote the SVD decomposition of $X = U\Sigma V^\T$. So $U$ is an orthonormal basis for the column-span of $X$ and as $X$ is full-rank $V$ is an orthonormal basis for $\mathbb{R}^p$. Finally, in our work we examine the linear regression problem derived from the projected data. 
That is, we denote
\begin{eqnarray}
{\tvec\beta} = (X^\T R^\T R X)^{-1} (RX)^\T (R\vec y) = \vec\beta + (RX)^\mpi R\vec e \label{eq:JL_beta} \end{eqnarray}\begin{eqnarray}
\tilde\sigma^2 =\frac r{r-p}  \|\tvec\zeta\|^2 \textrm{~~, with ~~~} {\tvec\zeta} = \tfrac 1 {\sqrt r} R\vec y - \tfrac 1 {\sqrt r} (RX)\tvec\beta \label{eq:JL_sigma}
\end{eqnarray}
We now give our main theorem, for estimating the $t$-values based on $\tvec\beta$ and $\tilde\sigma$.
\begin{theorem}
\label{thm:main_theorem_matrix_unaltered}
Let $X$ be a $(n\times p)$-matrix, and parameters $\vec\beta \in \mathbb{R}^p$ and $\sigma^2$ are such that we generate the vector $\vec y= X\vec\beta + \vec e$ with each coordinate of $\vec e$ sampled independently from $\gaussian (0,\sigma^2)$. 
Assume Algorithm~\ref{alg:private_JL} projects the matrix $A=[X;\vec y]$ without altering it.
Fix $\nu\in(0,1/2)$ and $r = p+\Omega(\ln(1/\nu))$. Fix coordinate $j$. Then we have that w.p. $\geq 1-\nu$  deriving $\tvec\beta$ and $\tilde\sigma^2$ as in Equations~\eqref{eq:JL_beta} and~\eqref{eq:JL_sigma}, the pivot quantity 
\ifx\fullversion\undefined
$\tilde t(\beta_j) = \frac {\tilde\beta_j-\beta_j} {\tilde\sigma \sqrt{(X^\T R^\T R X)^{-1}_{j,j}}}$
\else
\[ \tilde t(\beta_j) = \frac {\tilde\beta_j-\beta_j} {\tilde\sigma \sqrt{(X^\T R^\T R X)^{-1}_{j,j}}}\]
\fi
has a distribution $\mathcal{D}$ satisfying $e^{-a}\PDF_{T_{r-p}}(x) \leq \PDF_{\mathcal{D}}(x) \leq e^{a} \PDF_{T_{r-p}}(e^{-a}x)$ for any $x\in \mathbb{R}$, where we denote $a= \tfrac {r-p}{n-p}$.\end{theorem}

The implications of Theorem~\ref{thm:main_theorem_matrix_unaltered} are immediate: all estimations one can do based on the $t$-values from the true data $X, \vec y$,
we can now do based on $\tilde t$ modulo an approximation factor of $\exp(\tfrac {r-p} {n-p})$. In particular, Theorem~\ref{thm:main_theorem_matrix_unaltered} enables us to deduce a corresponding confidence interval based on $\tvec\beta$. 

\begin{corollary}
\label{cor:JL_confidence_interval}
In the same setting as in Theorem~\ref{thm:main_theorem_matrix_unaltered}, w.p. $\geq 1-\nu$ we have the following. Fix any $\alpha\in (0,\tfrac 1 2)$. Let $\tilde c_\alpha$ denote the number s.t. the interval $(\tilde c_\alpha,\infty)$ contains 
$\tfrac \alpha 2 e^{-a}$
probability mass of the $T_{r-p}$-distribution. Then 
\ifx\fullversion\undefined $\Pr[ \beta_j \in \left( \tilde \beta_j \pm e^{a}\cdot \tilde c_\alpha\cdot \tilde\sigma \sqrt{(X^\T R^\T R X)^{-1}_{j,j}}\right)] \geq 1-\alpha$.
\else \[\Pr[ \beta_j \in \left( \tilde \beta_j \pm e^{\tfrac{r-p}{n-p}}\tilde c_\alpha\cdot \tilde\sigma \sqrt{(X^\T R^\T R X)^{-1}_{j,j}}\right)] \geq 1-\alpha \]\fi
\footnote{Moreover, this interval is essentially optimal: denote $\tilde d_\alpha$ s.t the interval $(\tilde d_\alpha,\infty)$ contains $\tfrac\alpha 2 e^{\tfrac {r-p}{n-p}}$ probability mass of the $T_{r-p}$-distribution. Then 
\ifx\fullversion\undefined $\Pr[ \beta_j \in \left( \tilde \beta_j \pm \tilde d_\alpha\cdot \tilde\sigma \sqrt{(X^\T R^\T R X)^{-1}_{j,j}}\right)] \leq 1-\alpha$.
\else \[\Pr[ \beta_j \in \left( \tilde \beta_j \pm \tilde d_\alpha\cdot \tilde\sigma \sqrt{(X^\T R^\T R X)^{-1}_{j,j}}\right)] \leq 1-\alpha \]\fi}
\end{corollary}

We compare the confidence interval of Corollary~\ref{cor:JL_confidence_interval} to the confidence interval of the standard OLS model, whose length is $c_\alpha \tfrac {\|\vec\zeta\|}{\sqrt{n-p}} \sqrt{(X^\T X)^{-1}_{j,j}}$. As $R$ is a JL-matrix, known results regarding the JL transform 
\ifdefined\fullversion (see~Claim~\ref{clm:comparability_sigmaXX_lMM}) 
\fi give that $\|\tvec\zeta\| = \Theta\left(\|\vec\zeta\|\right)$, and that $\sqrt{(r-p)(X^\T R^\T R X)^{-1}_{j,j}}= \Theta\left(\sqrt{(X^\T X)^{-1}_{j,j}}\right)$. We therefore have that 
\ifdefined\fullversion 
\[ \tilde\sigma \sqrt{(X^\T R^\T R X)^{-1}_{j,j}} = \tfrac {\|\tvec\zeta\|}{\sqrt{r-p}} \sqrt{r} \sqrt{(X^\T R^\T R X)^{-1}_{j,j}} = \Theta\left(\sqrt{\tfrac {r\cdot (n-p)} {(r-p)^2}}\cdot \tfrac {\|\vec\zeta\|}{\sqrt{n-p}}  \sqrt{(X^\T X)^{-1}_{j,j}}\right) \]
\else
$\tilde\sigma \sqrt{(X^\T R^\T R X)^{-1}_{j,j}} = \tfrac {\|\tvec\zeta\|\sqrt{r}}{\sqrt{r-p}}  \sqrt{(X^\T R^\T R X)^{-1}_{j,j}} = \sqrt{\tfrac {r\cdot (n-p)} {(r-p)^2}}\cdot \Theta\left(\tfrac {\|\vec\zeta\|}{\sqrt{n-p}}  \sqrt{(X^\T X)^{-1}_{j,j}}\right)$.
\fi
So for values of $r$ for which $\tfrac r {r-p} = \Theta(1)$ we get that the confidence interval of Theorem~\ref{thm:main_theorem_matrix_unaltered} is a factor of $\Theta\left(\tfrac {\tilde c_\alpha}{c_\alpha}\sqrt{\tfrac {n-p} {r-p}}\right)$-larger than the standard OLS confidence interval. Observe that when $\alpha = \Theta(1)$, which is the common case, the dominating factor is $\sqrt{(n-p)/(r-p)}$. This bound intuitively makes sense: we have contracted $n$  observations to $r$ observations, hence our model is based on confidence intervals derived from $T_{r-p}$ rather than $T_{n-p}$.

In \ifx\confversion\undefined Section~\ref{apx_subsec:main_thm_discussion} \else the supplementary material \fi we give further discussion, in which we compare our work to the more straight-forward bounds one gets by ``plugging in''  Sarlos' work~\mycitet{Sarlos06}; and we also compare ourselves to the bounds derived from alternative works in differentially private linear regression.
\ifx\confversion\undefined The proof of Theorem~\ref{thm:main_theorem_matrix_unaltered} appears in Appendix~\ref{apx_subsec:proof_main_thm}. \fi
\cut{
The above comparison shows that we'd like to set $r$ as close as possible to $n$. However, as $r$ gets bigger, we require the matrix $A=[X;\vec y]$ to have larger and larger singular values. We discuss the tradeoff of $r$ and $n$ in Section~\ref{subsec:value_of_r}. 
\ifx\fullversion\undefined
\ifdefined\confversion
The proof of Theorem~\ref{thm:main_theorem_matrix_unaltered} appears in Appendix~\ref{apx_subsec:proof_main_thm}. We
\else
The proof of Theorem~\ref{thm:main_theorem_matrix_unaltered} appears in the full version. We
\fi
\else
Previous to that, we provide the proof of Theorem~\ref{thm:main_theorem_matrix_unaltered} in Section~\ref{subsec:proof_main_thm}, and
\fi
compare our baseline for rejecting the null-hypothesis in Section~\ref{subsec:simple_JL_reject_null_hypothesis}.
}
\DeclareRobustCommand{\JLDiscussion}{
\myparagraph{Comparison with Existing Bounds.} Sarlos' work~\mycitet{Sarlos06} utilizes the fact that when $r$, the numbers of rows in $R$, is large enough, then $\tfrac 1 {\sqrt r} R$ is a Johnson-Lindenstrauss matrix. Specifically, given $r$ and $\nu\in(0,1)$ we denote $\eta = \Omega( \sqrt{\frac{p\ln(p)\ln(1/\nu)}{r}} )$, and so $r = O(\tfrac {p\ln(p)\ln(1/\nu)} {\eta^2})$. Let us denote $\tilde {\vec\beta} = \arg\min_{\vec z} \tfrac 1 r\|RX\vec z - R\vec y\|^2$. In this setting, Sarlos' work~\mycite{Sarlos06} (Theorem 12(3)) guarantees that w.p. $\geq 1-\nu$ we have $\|\hat{\vec\beta}-\tilde{\vec\beta}\|_2 \leq \eta \|\vec\zeta\| / \svmin(X) = O\left( \sqrt{\frac{p\log(p)\log(1/\nu)}{r\svmin(X^\T X)}}\|\vec\zeta\|\right)$. Na\"ively bounding $|\hat\beta_j - \tilde\beta_j| \leq \|\hat{\vec\beta}-\tilde{\vec\beta}\|$ and using the confidence interval for $\hvec\beta_j - \vec\beta_j$ from Section~\ref{apx_subsec:OLS_background}\footnote{Where we approximate $c_\alpha$, the tail bound of the $T_{n-p}$-distribution with the tail bound on a Gaussian, i.e., use the approximation $c_\alpha \approx O(\sqrt{\ln(1/\alpha)})$.} gives a confidence interval of level $1-(\alpha+\nu)$ centered at $\tilde\beta_j$ with length of $O\left( \sqrt{\frac{p\ln(p)\log(1/\nu)}{r\svmin(X^\T X)}}\|\vec\zeta\|\right) + O\left( \sqrt{(X^\T X)^{-1}_{j,j}\frac{\log(1/\alpha)}{n-p}} \|\vec\zeta\|\right) = O\left( \sqrt{\frac{p\ln(p)\log(1/\nu) + \log(1/\alpha)}{r\svmin(X^\T X)}}\|\vec\zeta\|\right) $. This implies that our confidence interval has decreased its degrees of freedom from $n-p$ to roughly $r/p\ln(p)$, and furthermore, that it no longer depends on $(X^TX)^{-1}_{j,j}$ but rather on $1/\svmin(X^\T X)$. 
It is only due to the fact that we rely on Gaussians and by mimicking carefully the original proof that we can deduce that the $\tilde t$-value has (roughly) $r-p$ degrees of freedom and depends solely on $(X^TX)^{-1}_{j,j}$.

(In the worst case, we have that $(X^\T X)^{-1}_{j,j}$ is proportional to $\svmin(X^\T X)^{-1}$, but it is not uncommon to have matrices where the former is much larger than the latter.) As mentioned in the introduction, alternative techniques (\mycite{ChaudhuriMS11, BassilyST14, Ullman15}) for finding a DP estimator $\vec\beta^{dp}$ of the linear regression give a data-independent\footnote{In other words, independent of $X,\vec\zeta$.} bound of $\|\vec\beta^{dp}-\hvec\beta\| = \tilde O(p/\epsilon)$. Such bounds are harder to compare with the interval length given by Corollary~\ref{cor:JL_confidence_interval}. Indeed, as we discuss in
\ifdefined \fullversion Section~\ref{subsec:simple_JL_reject_null_hypothesis},
\else Section~\ref{sec:OLS_for_JL_data} under ``Rejecting the null-hypothesis,''
\fi
enough samples from a multivariate Gaussian whose covariance-matrix is well conditioned give a bound which is well below the worst-upper bound of $O(p/\epsilon)$. (Yet, it is possible that these techniques also do much better on such ``well-behaved'' data.) What the works of Sarlos and alternative works regrading differentially private linear regression do not take into account are questions such as generating a likelihood for $\beta_j$ nor do they discuss rejecting the null hypothesis.
}

\DeclareRobustCommand{\proofmainthm}{
We now turn to our analysis of $\tvec\beta$ and $\tvec\zeta$, where our goal is to show that the distribution of the $\tilde t$-values as specified in Theorem~\ref{thm:main_theorem_matrix_unaltered} is well-approximated by the $T_{r-p}$-distribution. For now, we assume the existence of fixed vectors $\vec \beta \in \mathbb{R}^p$ and $\vec e \in \mathbb{R}^{n}$ s.t. $\vec y= X\vec\beta + \vec e$. (Later, we will return to the homoscedastic model where each coordinate of $\vec e$  is sampled i.i.d from $\gaussian(0,\sigma^2)$ for some $\sigma^2$.) In other words, we first examine the case where $R$ is the sole source of randomness in our estimation.
Based on the assumption that $\vec e$ is fixed, we argue the following.
\begin{claim}
\label{clm:beta_and_zeta_distribution}
In our model, given $X$ and the output $M = RX$, we have that 
\ifx\fullversion\undefined
$\tvec\beta \sim \gaussian\left(\vec\beta + X^\mpi\vec e, \|P_{U^\perp}\vec e\|^2 (M^\T M)^{-1}\right)$ and $\tvec\zeta \sim \gaussian\left(\vec 0_n, \tfrac {\|P_{U^\perp}\vec e\|^2} r (I_{r\times r} - M(M^\T M)^{-1}M^\T)\right)$.
\else
\begin{eqnarray*}
& \tvec\beta & \sim \gaussian\left(\vec\beta + X^\mpi\vec e, \|P_{U^\perp}\vec e\|^2 (M^\T M)^{-1}\right)\cr
& \tvec\zeta & \sim \gaussian\left(\vec 0_n, \tfrac {\|P_{U^\perp}\vec e\|^2} r (I - M(M^\T M)^{-1}M^\T)\right)
\end{eqnarray*}
\fi
Where $P_{U^\perp}$ denotes the projection operator onto the subspace orthogonal to $colspan(X)$; i.e., $P_U = XX^\mpi$ and $P_{U^\perp} = (I_{r\times r}- XX^{\mpi})$.
\end{claim}

\begin{proof}
The matrix $R$ is sampled from $\gaussian( 0_{r\times p}, I_{r\times r}, I_{p\times p})$. Given $X$ and $RX=M$, we learn the projection of each row in $R$ onto the subspace spanned by  the columns of $X$. That is, denoting $\vec u^\T$ as the $i$-th row of $R$ and $\vec v^\T$ as the $i$-th row of $M$, we have that $X^\T \vec u = \vec v$. Recall, initially $\vec u \sim \gaussian(\vec 0_n, I_{n\times n})$ -- a spherically symmetric Gaussian. As a result, we can denote $\vec u = P_U \vec u \times P_{U^\perp} \vec u$ where the two projections are independent samples from $\gaussian( \vec 0_n, P_U)$ and $\gaussian (\vec 0_n, P_{U^\perp})$ resp. However, once we know that $\vec v = X^\T \vec u$ we have that $P_U \vec u = X(X^\T X)^{-1} X^\T \vec u =  X(X^\T X)^{-1} \vec v $ so we learn $P_U\vec u$ exactly, whereas we get no information about $P_{U^\perp}$ so $P_{U^\perp} \vec u$ is still sampled from a Gaussian $\gaussian(\vec 0_n , P_{U^\perp})$. As we know for each row of $R$ that $\vec u^\T P_U = \vec v^\T X^\mpi$, we therefore have that
\ifx\twocols\undefined
\[ R = RP_U + RP_{U^\perp} = MX^\mpi + RP_{U^\perp}, ~~~\textrm{ where } ~~~ RP_{U^\perp} \sim \gaussian(0_{r\times n}, I_{r\times r}, P_{U^\perp})\]
\else 
\[ R = RP_U + RP_{U^\perp} = MX^\mpi + RP_{U^\perp}\] where $RP_{U^\perp} \sim \gaussian(0_{r\times n}, I_{r\times r}, P_{U^\perp})$.
\fi 
From here on, we just rely on the existing results about the linearity of Gaussians.
\begin{align*}
 & R \sim \gaussian (MX^\mpi, I_{r\times r}, P_{U^\perp}) \cr & ~~~\Rightarrow  R\vec e \sim \gaussian(MX^\mpi\vec e, \|P_{U^\perp}\vec e\|^2 I_{r\times r}) \cr &~~~\Rightarrow M^\mpi R\vec e \sim \gaussian (X^\mpi \vec e, \|P_{U^\perp}\vec e\|^2 (M^\T M)^{-1}) 
\end{align*}
so $\tvec\beta = \vec\beta + M^\mpi R\vec e$ implies $\tvec\beta \sim \gaussian( \vec\beta + X^\mpi \vec e, \|P_{U^\perp}\vec e\|^2 (M^\T M)^{-1})$. And as $\tvec \zeta = \tfrac 1 {\sqrt r} (I_{r\times r}~-~M(M^\T M)^{-1}M^\T) R\vec e$ then we have $\tvec \zeta \sim \gaussian( \vec 0_r, \tfrac {\|P_{U^\perp}\vec e\|^2 } {r} (I_{r\times r}- M M^{\mpi}) )$ as $(I_{r\times r}-M M^\mpi)M = 0_{r\times p}$.
\end{proof}

Claim~\ref{clm:beta_and_zeta_distribution} was based on the assumption that $\vec e$ is fixed. However, given $X$ and $\vec y$ there are many different ways to assign vectors $\vec\beta$ and $\vec e$ s.t. $\vec y = X\vec \beta + \vec e$. However, the distributions we get in Claim~\ref{clm:beta_and_zeta_distribution} are \emph{unique}. To see that, recall Equations~\eqref{eq:hat_beta} and~\eqref{eq:zeta}: $\vec \beta + X^\mpi\vec e = X^\mpi \vec y = \hvec\beta$ and $P_{U^\perp}\vec e = P_{U^\perp}\vec y = (I-XX^\mpi)\vec y = \vec\zeta$. We therefore have $\tvec\beta\sim\gaussian(\hvec\beta, \|\vec\zeta\|^2(M^\T M)^{-1})$ and $\tvec\zeta\sim~\gaussian(\vec 0_n, \tfrac{ \|\vec\zeta\|^2}r (I-MM^\mpi))$. We will discuss this further, in Section~\ref{sec:ridge_regression}, where we will not be able to better analyze the  explicit distributions of our estimators. But in this section, we are able to argue more about the distributions of $\tvec\beta$ and $\tvec\zeta$.

So far we have considered the case that $\vec e$ is fixed, whereas our goal is to argue about the case where each coordinate of $\vec e$ is sampled i.i.d from $\gaussian(0,\sigma^2)$. To that end, we now switch to an intermediate model, in which $P_U\vec e$ is sampled from a multivariate Gaussian while $P_{U^\perp}\vec e$ is fixed as some arbitrary vector of length $l$. Formally, let $\mathcal{D}_l$ denote the distribution where $P_U\vec e \sim \gaussian (0, \sigma^2 P_U)$ and $P_{U^\perp}\vec e$ is fixed as some specific vector whose length is denoted by $\|P_{U^\perp}\vec e\|=l$.

\begin{claim}
\label{clm:homoscedastic_dist_of_JL_beta}
Under the same assumptions as in Claim~\ref{clm:beta_and_zeta_distribution}, given that $\vec e \sim \mathcal{D}_l$, we have that $\tvec \beta \sim \gaussian\left(\vec\beta, \sigma^2 (X^T X)^{-1} + l^2 (M^\T M)^{-1}\right)$ and $\tvec\zeta\sim\gaussian\left(\vec 0_n, \tfrac {l^2} r (I - MM^\mpi)\right)$.
\end{claim}

\begin{proof}
Recall, $\tvec \beta = \vec\beta + M^\mpi R\vec e = \vec\beta + M^\mpi(MX^\mpi + RP_{U^\perp})\vec e = \vec\beta + X^\mpi\vec e + M^\mpi R (P_{U^\perp}\vec e)$. Now, under the assumption $\vec e \sim \mathcal{D}_l$ we have that $\beta$ is the sum of two \emph{independent} Gaussians:
\begin{eqnarray*}
& \vec\beta + X^\mpi \vec e &\sim \gaussian (\vec\beta, \sigma^2 \left(X^\mpi \cdot P_U \cdot (X^\mpi)^\T\right)) \cr 
&&= \gaussian(\vec\beta,\sigma^2 (X^\T X)^{-1}) \cr 
& RP_{U^\perp} \vec e &\sim \gaussian (\vec 0_r, \|P_{U^\perp}\vec e\|^2 I_{r\times r})\cr
& \Rightarrow M^\mpi R \vec e &\sim \gaussian (\vec 0_p, \|P_{U^\perp}\vec e\|^2 (M^\T M)^{-1})
\end{eqnarray*}
Summing the two independent Gaussians' means and variances gives the distribution of $\tvec\beta$. Furthermore, in Claim~\ref{clm:beta_and_zeta_distribution} we have already established that for any fixed $\vec e$ we have $\tvec\zeta  \sim \gaussian\left(\vec 0_n, \tfrac {\|P_{U^\perp}\vec e\|^2} r (I - MM^\mpi)\right)$. Hence, for $\vec e \sim \mathcal{D}_l$ we still have $\tvec\zeta  \sim \gaussian\left(\vec 0_n, \tfrac {l^2} r (I - MM^\mpi)\right)$. (It is easy to verify that the same chain of derivations is applicable when $\vec e \sim \mathcal{D}_l$.)
\end{proof}
\begin{corollary}
\label{cor:tildebeta_j_zeta_squared}
Given that $\vec e\sim \mathcal{D}_l$ we have that $\tilde\beta_j \sim \gaussian(\beta_j,~ \sigma^2 (X^\T X)^{-1}_{j,j} + l^2(M^\T M)^{-1}_{j,j})$ for any coordinate $j$, and that $\|\tvec\zeta\|^2 \sim \frac{l^2}r \cdot\chi^2_{r-p}$.
\end{corollary}
\begin{proof}The corollary follows immediately from the fact that $\beta_j = \vec e_j^\T \tvec\beta$, and from the definition of the $\chi^2$-distribution, as $\tvec\zeta$ is a spherically symmetric Gaussian defined on the subspace $colspan(M)^\perp$ of dimension $r-p$.\end{proof}
To continue, we need the following claim.
\begin{claim}
\label{clm:beta_zeta_independent}
Given $X$ and $M=RX$, and given that $\vec e \sim\mathcal{D}_l$ we have that $\tvec\beta$ and $\tvec\zeta$ are independent.
\end{claim}
\begin{proof}
Recall, $\tvec\beta=\vec\beta + X^\mpi\vec e + M^\mpi R (P_{U^\perp}\vec e)$. And so, given $X$, $M$ and a specific vector $P_{U^\perp}\vec e $ we have that the distribution of $\tvec\beta$ depends on (i) the projection of $\vec e$ on $U = colspan(X)$ and on (ii) the projection of each row in $R$ onto $\tilde U = colspan(M)$. The distribution of $\tvec\zeta = \tfrac 1 {\sqrt r} P_{\tilde U^\perp} R \vec e = \tfrac 1 {\sqrt r} P_{\tilde U^\perp} (MX^\mpi + R P_{U^\perp})\vec e = \tfrac 1 {\sqrt r} P_{\tilde U^\perp} R P_{U^\perp}\vec e$ depends on (i) the projection of $\vec e$ onto $U^{\perp}$ (which for the time being is fix to some specific vector of length $l$) and on (ii) the projection of each row in $R$ onto $\tilde U^\perp$. Since $P_U\vec e$ is independent from $P_{U^\perp} \vec e$, and since for any row $\vec u^\T$ of $R$ we have that $P_{\tilde U}\vec u$ is independent of $P_{\tilde U^{\perp}}\vec u$, and since $\vec e$ and $R$ are chosen independently, we have that $\tvec\beta$ and $\tvec\zeta$ are independent.

Formally, consider any pair of coordinates $\tilde\beta_j$ and $\tilde\zeta_k$, and we have
\begin{eqnarray*}
&\tilde\beta_j - \beta_j & = \vec e_j^\T X^\mpi \vec e + \vec e_j^\T M^\mpi (R P_{U^{\perp}}\vec e)\cr
&\tilde\zeta_k &=  \vec e_k^\T P_{\tilde U^\perp} (RP_{U^\perp}\vec e)
\end{eqnarray*}
Recall, we are given $X$ and $M = RX$. Therefore, we know $P_U$ and $P_{\tilde U}$. And so
\ifx\twocols\undefined 
\begin{align*}
\textrm{Cov}[\tilde\beta_j,\tilde\zeta_k] & = \E[(\tilde\beta_j-\beta_j)(\tilde\zeta_k - 0)] \cr
&= \E [ \vec e_j^\T X^\mpi \vec e (RP_{U^\perp}\vec e)^\T P_{\tilde U^\perp} \vec e_k ] + \E[ \vec e_j^\T M^\mpi  (RP_{U^\perp}\vec e)(RP_{U^\perp}\vec e)^\T P_{\tilde U^\perp} \vec e_k ] \cr
& = \vec e_j^\T X^\mpi \E[ \vec e \vec e^\T P_{U^\perp}] \E[ R^\T] P_{\tilde U^\perp}\vec e_k + \vec e_j^\T M^\mpi \E[(RP_{U^\perp}\vec e)(RP_{U^\perp}\vec e)^\T] P_{\tilde U^\perp} \vec e_k \cr
&= \vec e_j^\T X^\mpi \E[ \vec e \vec e^\T P_{U^\perp}] \left((MX^\mpi)^\T+ \E[(RP_{U^\perp})^\T] \right) P_{\tilde U^\perp}\vec e_k + \vec e_j^\T M^\mpi \left( \|P_{U^\perp}\vec e\|^2 I_{r\times r}\right) P_{\tilde U^\perp} \vec e_k \cr
&= \vec e_j^\T X^\mpi \E[ \vec e \vec e^\T P_{U^\perp}] (X^\mpi)^\T \left(M^\T P_{\tilde U^\perp}\right) \vec e_k + 0 + l^2 \cdot \vec e_j^\T \left(M^\mpi P_{\tilde U^\perp}\right) \vec e_k = 0+0+0= 0
\end{align*}
\else
\begin{align*}
& \textrm{Cov}[\tilde\beta_j,\tilde\zeta_k] \cr & = \E[(\tilde\beta_j-\beta_j)(\tilde\zeta_k - 0)] \cr
&= \E [ \vec e_j^\T X^\mpi \vec e (RP_{U^\perp}\vec e)^\T P_{\tilde U^\perp} \vec e_k ] \cr &~~~+ \E[ \vec e_j^\T M^\mpi  (RP_{U^\perp}\vec e)(RP_{U^\perp}\vec e)^\T P_{\tilde U^\perp} \vec e_k ] \cr
& = \vec e_j^\T X^\mpi \E[ \vec e \vec e^\T P_{U^\perp}] \E[ R^\T] P_{\tilde U^\perp}\vec e_k \cr &~~~ + \vec e_j^\T M^\mpi \E[(RP_{U^\perp}\vec e)(RP_{U^\perp}\vec e)^\T] P_{\tilde U^\perp} \vec e_k \cr
&= \vec e_j^\T X^\mpi \E[ \vec e \vec e^\T P_{U^\perp}] \left((MX^\mpi)^\T+ \E[(RP_{U^\perp})^\T] \right) P_{\tilde U^\perp}\vec e_k \cr &~~~ + \vec e_j^\T M^\mpi \left( \|P_{U^\perp}\vec e\|^2 I_{r\times r}\right) P_{\tilde U^\perp} \vec e_k \cr
&= \vec e_j^\T X^\mpi \E[ \vec e \vec e^\T P_{U^\perp}] (X^\mpi)^\T \left(M^\T P_{\tilde U^\perp}\right) \vec e_k +0 \cr &~~~ + l^2 \cdot \vec e_j^\T \left(M^\mpi P_{\tilde U^\perp}\right) \vec e_k \cr & = 0+0+0= 0
\end{align*}
\fi
And as $\tvec\beta$ and $\tvec\zeta$ are Gaussians, having their covariance $=0$ implies independence.
\end{proof}
\cut{
\begin{corollary}
\label{cor:beta_zeta_independent}
Given $X$ and $M=RX$, and given that $\vec e$ is sampled i.i.d s.t. $e_i \sim\gaussian(0,\sigma^2)$, we have that $\tvec\beta$ and $\tvec\zeta$ are independent.
\end{corollary}
\begin{proof}
The same argument for $\vec e \sim \mathcal{D}_l$ applies here too. In particular, we have that, under the same notation as in the proof of Claim~\ref{clm:beta_zeta_independent},
\begin{eqnarray*}
& \textrm{Cov}[\tilde\beta_j,\tilde\zeta_k] &= 0 + \E[\vec e_j^\T M^\mpi (RP_{U^\perp}\vec e)(RP_{U^\perp}\vec e)^\T P_{\tilde U^\perp} \vec e_k]\cr
&& = \int_0^\infty \E\left[ \vec e_j^\T M^\mpi (RP_{U^\perp}\vec e)(RP_{U^\perp}\vec e)^\T P_{\tilde U^\perp}\vec e_k ~\Big| ~\|P_{U^\perp}\vec e\| = l \right] \PDF_{\| P_{U^\perp}\vec e\|}(l)dl
\cr && = \int_{0}^\infty 0 \cdot dl = 0\end{eqnarray*}
\end{proof}
}
Having established that $\tvec\beta$ and $\tvec\zeta$ are independent Gaussians and specified their distributions, we continue with the proof of Theorem~\ref{thm:main_theorem_matrix_unaltered}. We assume for now that there exists some small $a>0$ s.t. 
\ifx\twocols\undefined
\begin{equation} l^2 (M^\T M)^{-1}_{j,j}\leq \sigma^2 (X^T X)^{-1}_{j,j} + l^2 (M^\T M)^{-1}_{j,j}\leq  e^{2a} \cdot l^2 (M^\T M)^{-1}_{j,j}\label{eq:exists_bounding_a}\end{equation} 
\else 
\begin{eqnarray} l^2 (M^\T M)^{-1}_{j,j}\leq \sigma^2 (X^T X)^{-1}_{j,j} + l^2 (M^\T M)^{-1}_{j,j} \cr \leq  e^{2a} \cdot l^2 (M^\T M)^{-1}_{j,j}\label{eq:exists_bounding_a}\end{eqnarray} 
\fi
Then, due to Corollary~\ref{cor:probabilities_ratio_close_gaussians}, denoting the distributions $\gaussian_1 = \gaussian(0,~ l^2(M^\T M)^{-1}_{j,j})$ and $\gaussian_2 = \gaussian(0,~ \sigma^2 (X^\T X)^{-1}_{j,j} + l^2(M^\T M)^{-1}_{j,j})$, we have that for any $S\subset\mathbb{R}$ it holds that\footnote{In fact, it is possible to use standard techniques from differential privacy, and argue a similar result~--- that the probabilities of any event that depends on some function $f(\beta_j)$ under $\beta_j \sim \gaussian_1$ and  under $\beta_j\sim\gaussian_2$ are close in the differential privacy sense.}
\begin{equation} e^{-a}\Pr_{\tilde\beta_j \sim \gaussian_1}[S]\leq\Pr_{\tilde\beta_j\sim \gaussian_2 }[S]\leq e^a \Pr_{\tilde\beta_j \sim \gaussian_1}[S/e^{a}] \label{eq:ratios_N1_N2}\end{equation}

More specifically, denote the function
\ifx\twocols\undefined 
\begin{eqnarray*}&\tilde t(\psi, \|\vec\xi\|,\beta_j) &= \frac {\psi-\beta_j} {\|\vec\xi\| \sqrt{\tfrac r {r-p} (M^\T M)^{-1}_{j,j}}}=\frac {\psi-\beta_j} {l\sqrt{(M^\T M)^{-1}_{j,j}}} \Big/ \frac {\|\vec\xi\|\sqrt{\tfrac r {r-p}}} l 
\end{eqnarray*}
\else 
\begin{eqnarray*}\tilde t(\psi, \|\vec\xi\|,\beta_j) &&= \frac {\psi-\beta_j} {\|\vec\xi\| \sqrt{\tfrac r {r-p} (M^\T M)^{-1}_{j,j}}}\cr &&=\frac {\psi-\beta_j} {l\sqrt{(M^\T M)^{-1}_{j,j}}} \Big/ \frac {\|\vec\xi\|\sqrt{\tfrac r {r-p}}} l 
\end{eqnarray*}
\fi 
and observe that when we sample $\psi,\vec\xi$ independently s.t. $\psi \sim \gaussian( \beta_j , ~ l^2 (M^\T M)^{-1}_{j,j})$ and $\|\vec\xi\|^2 \sim \tfrac{l^2}{r}\chi^2_{r-p}$ then $\tilde t(\psi,\|\vec\xi\|,\beta_j)$ is distributed like a $T$-distribution with $r-p$ degrees of freedom. And so, for any $\tau >0 $ we have that under such way to sample $\psi,\vec\xi$ we have $\Pr[\tilde t(\psi,\|\vec\xi\|,\beta_j) > \tau] = 1-\CDF_{T_{r-p}}(\tau)$.

For any $\tau \geq 0$ and for any non-negative real value $z$ let $S^\tau_{z}$ denote the suitable set of values s.t.
\ifx\twocols\undefined 
\[ \Pr_{\left\{\substack{\psi \sim \gaussian( \beta_j , ~ l^2 (M^\T M)^{-1}_{j,j})\\ \|\vec\xi\|^2 \sim \tfrac{l^2}{r}\chi^2_{r-p}}\right\} }[\tilde t(\psi,\|\vec\xi\|,\beta_j) > \tau]
 = \int\limits_0^{\infty} \PDF_{\tfrac {l^2}r \chi^2_{r-p}}(z) \cdot \Prx_{\left\{\psi-\beta_j \sim \gaussian( 0 , ~ l^2 (M^\T M)^{-1}_{j,j})\right\}}[S^\tau_z]\,dz\]
\else
\begin{eqnarray*} 
&&\Pr_{\left\{\substack{\psi \sim \gaussian( \beta_j , ~ l^2 (M^\T M)^{-1}_{j,j})\\ \|\vec\xi\|^2 \sim \tfrac{l^2}{r}\chi^2_{r-p}}\right\} }[\tilde t(\psi,\|\vec\xi\|,\beta_j) > \tau]
\cr &&= \int\limits_0^{\infty} \PDF_{\tfrac {l^2}r \chi^2_{r-p}}(z) \cdot \Prx_{\left\{\psi-\beta_j \sim \gaussian( 0 , ~ l^2 (M^\T M)^{-1}_{j,j})\right\}}[S^\tau_z]\,dz\end{eqnarray*}
\fi
That is, $S^\tau_z = \left(\tau\cdot z\sqrt{\tfrac r {r-p} (M^\T M)^{-1}_{j,j}},~\infty\right)$.

We now use Equation~\eqref{eq:ratios_N1_N2} (Since $\gaussian( 0 , ~ l^2 (M^\T M)^{-1}_{j,j})$ is precisely $\gaussian_1$) to deduce that
\ifx\twocols\undefined
\begin{align*}
&\Pr_{\left\{\substack{\psi \sim \gaussian( \beta_j , ~ l^2 (M^\T M)^{-1}_{j,j} + \sigma^2 (X^\T X)^{-1}_{j,j})\\ \|\vec\xi\|^2 \sim \tfrac{l^2}{r}\chi^2_{r-p}}\right\} }[\tilde t(\psi,\|\vec\xi\|,\beta_j) > \tau] 
\cr &~~~~~~~=  \int_0^{\infty} \PDF_{\tfrac {l^2}r \chi^2_{r-p}}(z) \Pr_{\psi-\beta_j \sim \gaussian( 0 , ~ l^2 (M^\T M)^{-1}_{j,j} + \sigma^2 (X^\T X)^{-1}_{j,j})}[S^\tau_z]dz
\cr &~~~~~~~ \leq e^a \int_0^{\infty} \PDF_{\tfrac {l^2}r \chi^2_{r-p}}(z) \Pr_{\psi-\beta_j \sim \gaussian( 0 , ~ l^2 (M^\T M)^{-1}_{j,j} )}[S^\tau_z/e_a]dz
\cr &~~~~~~~ \stackrel {(*)}= e^a \int_0^{\infty} \PDF_{\tfrac {l^2}r \chi^2_{r-p}}(z) \Pr_{\psi-\beta_j \sim \gaussian( 0 , ~ l^2 (M^\T M)^{-1}_{j,j} )}[S^{\tau/e^a}_{z}]dz 
\cr &~~~~~~~ = e^{a}\Pr_{\left\{\substack{\psi \sim \gaussian( \beta_j , ~ l^2 (M^\T M)^{-1}_{j,j})\\ \|\vec\xi\|^2 \sim \tfrac{l^2}{r}\chi^2_{r-p}}\right\} }[\tilde t(\psi,\|\vec\xi\|,\beta_j) > \tau/e^a] = e^a\left(1 - \CDF_{T_{r-p}}(\tau/e^a)\right)
\end{align*}
\else 
\begin{align*}
&\Pr_{\left\{\substack{\psi \sim \gaussian( \beta_j , ~ l^2 (M^\T M)^{-1}_{j,j} + \sigma^2 (X^\T X)^{-1}_{j,j})\\ \|\vec\xi\|^2 \sim \tfrac{l^2}{r}\chi^2_{r-p}}\right\} }[\tilde t(\psi,\|\vec\xi\|,\beta_j) > \tau] 
\cr &~=  \int_0^{\infty} \PDF_{\tfrac {l^2}r \chi^2_{r-p}}(z) \Prx_{\resizebox*{0.25\textwidth}{!}{$\psi-\beta_j \sim \gaussian( 0 , ~ l^2 (M^\T M)^{-1}_{j,j} + \sigma^2 (X^\T X)^{-1}_{j,j})$}}[S^\tau_z]dz
\cr &~ \leq e^a \int_0^{\infty} \PDF_{\tfrac {l^2}r \chi^2_{r-p}}(z) \Prx_{\psi-\beta_j \sim \gaussian( 0 , ~ l^2 (M^\T M)^{-1}_{j,j} )}[S^\tau_z/e_a]dz
\cr &~ \stackrel {(*)}= e^a \int_0^{\infty} \PDF_{\tfrac {l^2}r \chi^2_{r-p}}(z) \Prx_{\psi-\beta_j \sim \gaussian( 0 , ~ l^2 (M^\T M)^{-1}_{j,j} )}[S^{\tau/e^a}_{z}]dz 
\cr &~ = e^{a}\Pr_{\left\{\substack{\psi \sim \gaussian( \beta_j , ~ l^2 (M^\T M)^{-1}_{j,j})\\ \|\vec\xi\|^2 \sim \tfrac{l^2}{r}\chi^2_{r-p}}\right\} }[\tilde t(\psi,\|\vec\xi\|,\beta_j) > \tau/e^a] 
\cr &~= e^a\left(1 - \CDF_{T_{r-p}}(\tau/e^a)\right)
\end{align*}
\fi 
where the equality $(*)$ follows from the fact that $S^\tau_z / c = S^{\tau/c}_z$ for any $c>0$, since it is a non-negative interval. Analogously, we can also show that
\begin{align*}
& \Pr_{\left\{\substack{\psi \sim \gaussian( \beta_j , ~ l^2 (M^\T M)^{-1}_{j,j} + \sigma^2 (X^\T X)^{-1}_{j,j})\\ \|\vec\xi\|^2 \sim \tfrac{l^2}{r}\chi^2_{r-p}}\right\} }[\tilde t(\psi,\|\vec\xi\|,\beta_j) > \tau]
\cr &~~~~~~~ \geq e^{-a}\Pr_{\left\{\substack{\psi \sim \gaussian( \beta_j , ~ l^2 (M^\T M)^{-1}_{j,j})\\ \|\vec\xi\|^2 \sim \tfrac{l^2}{r}\chi^2_{r-p}}\right\} }[\tilde t(\psi,\|\vec\xi\|,\beta_j) > \tau] 
\ifx\twocols\undefined = \else \cr &= \fi e^{-a}\left(1-\CDF_{T_{r-p}}(\tau) \right) 
 \end{align*}

In other words, we have just shown that for any interval $I = (\tau,\infty)$ with $\tau\geq 0$ we have
\ifx\twocols\undefined 
\[ e^a \int\limits_I \PDF_{T_{r-p}}(z)dz \leq \Pr_{\left\{\substack{\psi \sim \gaussian( \beta_j , ~ l^2 (M^\T M)^{-1}_{j,j} + \sigma^2 (X^\T X)^{-1}_{j,j})\\ \|\vec\xi\|^2 \sim \tfrac{l^2}{r}\chi^2_{r-p}}\right\} }[\tilde t(\psi,\|\vec\xi\|,\beta_j) \in I] \leq e^a \int\limits_{I/e^a} \PDF_{T_{r-p}}(z)dz\]
\else
that $\Pr_{\left\{\substack{\psi \sim \gaussian( \beta_j , ~ l^2 (M^\T M)^{-1}_{j,j} + \sigma^2 (X^\T X)^{-1}_{j,j})\\ \|\vec\xi\|^2 \sim \tfrac{l^2}{r}\chi^2_{r-p}}\right\} }[\tilde t(\psi,\|\vec\xi\|,\beta_j) \in I]$ is lower bounded by $e^a \int\limits_I \PDF_{T_{r-p}}(z)dz$ and upper bounded by $e^a \int\limits_{I/e^a} \PDF_{T_{r-p}}(z)dz$.
\fi 
We can now repeat the same argument for $I = (\tau_1,\tau_2)$ with $0\leq \tau_1<\tau_2$ (using an analogous definition of $S^{\tau_1,\tau_2}_z$), and again for any $I= (\tau_1,\tau_2)$ with $\tau_1 < \tau_2 \leq 0$, and deduce that the $\PDF$ of the function $\tilde t(\psi,\|\vec\xi\|,\beta_j)$  at $x$~--- where we sample  $\psi \sim \gaussian( \beta_j , ~ l^2 (M^\T M)^{-1}_{j,j} + \sigma^2 (X^\T X)^{-1}_{j,j})$ and $\|\vec\xi\|^2 \sim \tfrac{l^2}{r}\chi^2_{r-p}$ independently~--- lies in the range $\left(e^{-a} \PDF_{T_{r-p}}(x), e^a\PDF_{T_{r-p}}(x/e^a)\right)$.
And so, using Corollary~\ref{cor:tildebeta_j_zeta_squared} and Claim~\ref{clm:beta_zeta_independent}, we have that when $\vec e\sim \mathcal{D}_l$, the distributions of $\tilde\beta_j$ and $\|\tvec\zeta\|^2$ are precisely as stated above, and so we have that the distribution of $\tilde t(\beta_j) \stackrel{\rm def} = \tilde t(\tilde \beta_j, \|\tvec\zeta\|, \beta_j)$
has a $\PDF$ that at the point $x$ is ``sandwiched'' between $e^{-a} \PDF_{T_{r-p}}(x)$ and $e^a\PDF_{T_{r-p}}(x/e^a)$.

Next, we aim to argue that this characterization of the $\PDF$ of $\tilde t(\beta_j)$ still holds when $e\sim~\gaussian(\vec 0_n, \sigma^2 I_{n\times n})$. It would be convenient to think of $\vec e$ as a sample in $\gaussian(\vec 0_n, \sigma^2 P_U) \times \gaussian(\vec 0_n, \sigma^2 P_{U^\perp})$. (So while in $\mathcal{D}_l$ we have $P_U \vec e \sim \gaussian(\vec 0_n, \sigma^2P_U)$ but $P_{U^\perp} \vec e $ is fixed, now both $P_U \vec e$ and $P_{U^\perp}\vec e$ are sampled from spherical Gaussians.) The reason why the above still holds lies in the fact that $\tilde t(\beta_j)$ does not depend on $l$. In more details:
\ifx\twocols\undefined 
\begin{align*}
\Pr_{\vec e \sim \gaussian(\vec 0_n, \sigma^2 I_{n\times n})} \left[\tilde t(\beta_j) \in I\right] &= \int_{\vec v} \Pr_{\vec e \sim \gaussian(\vec 0_n, \sigma^2 I_{n\times n})}\left[\tilde t(\beta_j) \in I \;|~ P_{U^\perp}\vec e = \vec v\right] \PDF_{P_{U^\perp}\vec e} (\vec v)d\vec v  
\cr & = \int_{\vec v} \Pr_{\vec e \sim \mathcal{D}_{l}}\left[\tilde t(\beta_j) \in I \;|~ l = \|\vec v\|\right] \PDF_{P_{U^\perp}\vec e} (\vec v)d\vec v  
\cr & \leq \int_{\vec v} \left(e^a \int_{I/e^a} \PDF_{T_{r-p}}(z)dz \right) \PDF_{P_{U^\perp}\vec e} (\vec v)d\vec v
\cr & =  \left(e^a \int_{I/e^a} \PDF_{T_{r-p}}(z)dz \right) \int_{\vec v}\PDF_{P_{U^\perp}\vec e} (\vec v)d\vec v =  e^a \int_{I/e^a} \PDF_{T_{r-p}}(z)dz
\end{align*}
\else 
\begin{align*}
&\Pr_{\vec e \sim \gaussian(\vec 0_n, \sigma^2 I_{n\times n})} \left[\tilde t(\beta_j) \in I\right] \cr &= \int_{\vec v} \resizebox*{0.3\textwidth}{!}{$\Prx_{\vec e \sim \gaussian(\vec 0_n, \sigma^2 I_{n\times n})}\left[\tilde t(\beta_j) \in I \;|~ P_{U^\perp}\vec e = \vec v\right]$} \PDF_{P_{U^\perp}\vec e} (\vec v)d\vec v  
\cr & = \int_{\vec v} \Prx_{\vec e \sim \mathcal{D}_{l}}\left[\tilde t(\beta_j) \in I \;|~ l = \|\vec v\|\right] \PDF_{P_{U^\perp}\vec e} (\vec v)d\vec v  
\cr & \leq \int_{\vec v} \left(e^a \int_{I/e^a} \PDF_{T_{r-p}}(z)dz \right) \PDF_{P_{U^\perp}\vec e} (\vec v)d\vec v
\cr & =  \left(e^a \int_{I/e^a} \PDF_{T_{r-p}}(z)dz \right) \int_{\vec v}\PDF_{P_{U^\perp}\vec e} (\vec v)d\vec v 
\cr &=  e^a \int_{I/e^a} \PDF_{T_{r-p}}(z)dz
\end{align*}
\fi 
where the last transition is possible precisely because $\tilde t$ is independent of $l$ (or $\|\vec v\|$)~--- which is precisely what makes this $t$-value a pivot quantity. The proof of the lower bound is symmetric. 

To conclude, we have shown that if Equation~\eqref{eq:exists_bounding_a} holds, then for every interval $I \subset \mathbb{R}$ we have
\ifx\twocols\undefined 
\[ e^{-a} \Pr_{z \sim T_{r-p}}\left[z\in I\right] \leq  \Pr_{\vec e \sim \gaussian(\vec 0_n, \sigma^2 I_{n\times n})} \left[\tilde t(\beta_j) \in I\right] \leq e^a \Pr_{z \sim T_{r-p}}\left[z\in (I/e^a)\right]  \]
\else 
that $\Pr_{\vec e \sim \gaussian(\vec 0_n, \sigma^2 I_{n\times n})} \left[\tilde t(\beta_j) \in I\right]$ is lower bounded by $e^{-a} \Pr_{z \sim T_{r-p}}\left[z\in I\right]$ and upper bounded by $e^a \Pr_{z \sim T_{r-p}}\left[z\in (I/e^a)\right]$.
\fi 
So to conclude the proof of Theorem~\ref{thm:main_theorem_matrix_unaltered}, we need to show that w.h.p such $a$ as in Equation~\eqref{eq:exists_bounding_a} exists.
\begin{claim}
\label{clm:comparability_sigmaXX_lMM}
In the homoscedastic model with Gaussian noise, if both $n$ and $r$ satisfy $n,r\geq~p~+~\Omega(\log(1/\nu))$, then we have that 
\ifx\twocols\undefined \[ l^2 (M^\T M)^{-1}_{j,j}\leq \sigma^2 (X^T X)^{-1}_{j,j} + l^2 (M^\T M)^{-1}_{j,j}\leq  (1 + \tfrac {2(r-p)}{n-p}) \cdot l^2 (M^\T M)^{-1}_{j,j} \leq e^{\tfrac {2(r-p)}{n-p}}\cdot l^2 (M^\T M)^{-1}_{j,j} \]  \end{claim}
\else 
$\sigma^2 (X^T X)^{-1}_{j,j} + l^2 (M^\T M)^{-1}_{j,j} \geq l^2 (M^\T M)^{-1}_{j,j}$ and 
 \[ \sigma^2 (X^T X)^{-1}_{j,j} + l^2 (M^\T M)^{-1}_{j,j}\leq  (1 + \tfrac {2(r-p)}{n-p}) \cdot l^2 (M^\T M)^{-1}_{j,j}\]  \end{claim} Using  $(1 + \tfrac {2(r-p)}{n-p})\leq e^{\tfrac {2(r-p)}{n-p}}$, 
\fi 
 Theorem~\ref{thm:main_theorem_matrix_unaltered} now follows from plugging $a= \tfrac {r-p}{n-p}$ to our above discussion.
\begin{proof}
The lower bound is immediate from non-negativity of $\sigma^2$ and of $(X^\T X)^{-1}_{j,j} = \| (X^\T X)^{-1/2}\vec e_j\|^2$. We therefore prove the upper bound.

First, observe that $l^2 = \|P_{U^\perp}\vec e\|^2$ is sampled from $\sigma^2\cdot \chi^2_{n-p}$ as $U^\perp$ is of dimension $n-p$. Therefore, it holds that w.p. $\geq 1- \nu/2$ that
\[ \sigma^2\left(\sqrt{n-p}-\sqrt{2\ln(2/\nu)}\right)^2 \leq l^2 \] and assuming $n > p + 100\ln(2/\nu)$ we therefore have $\sigma^2 \leq \tfrac {4} {3(n-p)}l^2$.

Secondly, we argue that when $r > p + 300 \ln(4/\nu)$ we have that w.p. $\geq 1 -\nu/2$ it holds that $\tfrac 3 4 (X^\T X)^{-1}_{j,j} \leq (r-p) (X^\T R^\T R X)^{-1}_{j,j}$. To see this, first observe that by picking $R\sim \gaussian(0_{r\times n}, I_{r\times r}, I_{n\times n})$ the distribution of the product $RX \sim \gaussian(0_{r\times d}, I_{r\times r}, X^\T X)$ is identical to picking $Q\sim \gaussian (0_{r\times d}, I_{r\times r}, I_{d\times d})$ and taking the product $Q(X^\T X)^{1/2}$. Therefore, the distribution of $(X^\T R^T R X)^{-1}$ is identical to $\left((X^\T X)^{1/2} Q^\T Q (X^\T X)^{1/2} \right)^{-1} =(X^\T X)^{-1/2} (Q^\T Q)^{-1} (X^\T X)^{-1/2}$. Denoting $\vec v = (X^\T X)^{-1/2}\vec e_j$ we have $\|\vec v\|^2 = (X^\T X)^{-1}_{j,j}$. Claim A.1 from~\mycite{Sheffet15} gives that w.p. $\geq 1-\nu/2$ we have 
\ifx\twocols\undefined 
\[ (r-p)\cdot \vec e_j^\T \left((X^\T X)^{1/2} Q^\T Q (X^\T X)^{1/2} \right)^{-1} \vec e_j = \vec v^\T (\tfrac 1 {r-p}Q^\T Q)^{-1} \vec v \geq \tfrac 3 4 \vec v^\T \vec v = \tfrac 3 4(X^\T X)^{-1}_{j,j}\]
\else 
\begin{eqnarray*}
&&(r-p)\cdot \vec e_j^\T \left((X^\T X)^{1/2} Q^\T Q (X^\T X)^{1/2} \right)^{-1} \vec e_j \cr &&~~~~~= \vec v^\T (\tfrac 1 {r-p}Q^\T Q)^{-1} \vec v \geq \tfrac 3 4 \vec v^\T \vec v = \tfrac 3 4(X^\T X)^{-1}_{j,j}
\end{eqnarray*}
which implies the required.

Combining the two inequalities we get: 
\ifx\twocols\undefined
\[ \sigma^2 (X^\T X)^{-1}_{j,j} \leq \tfrac {16l^2(r-p)}{n-p} (X^\T R^\T RX)^{-1}_{j,j} \leq \tfrac {2(r-p)}{n-p}l^2 (X^\T R^\T RX)^{-1}_{j,j}\]
\else 
\begin{eqnarray*}
\sigma^2 (X^\T X)^{-1}_{j,j} &&\leq \tfrac {16l^2(r-p)}{n-p} (X^\T R^\T RX)^{-1}_{j,j} \cr &&~~\leq \tfrac {2(r-p)}{n-p}l^2 (X^\T R^\T RX)^{-1}_{j,j}
\end{eqnarray*}
\fi
and as we denote $M=RX$ we are done.\end{proof}

We comment that our analysis in the proof of Claim~\ref{clm:comparability_sigmaXX_lMM} implicitly assumes $r \ll n$ (as we do think of the projection $R$ as dimensionality reduction), and so the ratio $\tfrac {r-p}{n-p}$ is small. However, a similar analysis holds for $r$ which is comparable to $n$~---  in which we would argue that $\frac {\sigma^2 (X^T X)^{-1}_{j,j} + l^2 (M^\T M)^{-1}_{j,j}} {\sigma^2 (X^T X)^{-1}} \in [1, 1+\eta]$ for some small $\eta$.
}
\ifdefined\fullversion 
\subsection{Proof of Theorem~\ref{thm:main_theorem_matrix_unaltered}}
\label{subsec:proof_main_thm}
\proofmainthm 
\fi

\ifdefined\fullversion
\subsection{Rejecting the Null Hypothesis}
\label{subsec:simple_JL_reject_null_hypothesis}
\else
\myparagraph{Rejecting the Null Hypothesis.}
\fi
Due to Theorem~\ref{thm:main_theorem_matrix_unaltered}, we can mimic OLS'  technique for rejecting the null hypothesis. I.e., we denote $\tilde t_0 = \frac{\tilde\beta_j}{\tilde\sigma\sqrt{(X^\T R^\T R X)^{-1}_{j,j}}}$ and reject the null-hypothesis if indeed the associated $\tilde p_0$, denoting $p$-value of the slightly truncated $e^{-\tfrac {r-p}{n-p}}\tilde t_0$, 
is below $\alpha\cdot  e^{-\tfrac {r-p}{n-p}}$. Much like Theorem~\ref{thm:baseline_rejecting_nh} we now establish a lower bound on $n$ so that w.h.p we end up (correctly) rejecting the null-hypothesis. \ifdefined\fullversion
\comparisonboundproof
\else \ifx\confversion\undefined The proof of Theorem~\ref{thm:JL_simple_case_rejecting_null_hypothesis} is deferred to Appendix~\ref{apx_subsec:comparison_bound}.\fi 
\fi

\begin{theorem}
\label{thm:JL_simple_case_rejecting_null_hypothesis}
Fix a positive definite matrix $\Sigma \in \mathbb{R}^{p\times p}$. Fix parameters $\vec\beta\in\mathbb{R}^p$ and $\sigma^2 > 0$ and a coordinate $j$ s.t. $\beta_j\neq 0$. Let $X$ be a matrix whose $n$ rows are sampled i.i.d from $\gaussian(\vec 0_p, \Sigma)$. Let $\vec y$ be a vector s.t. $y_i-(X\vec\beta)_i$ is sampled i.i.d from $\gaussian(0,\sigma^2)$. Fix $\nu \in (0,1/2)$ and $\alpha \in (0,1/2)$. Then there exist constants $C_1$, $C_2$, $C_3$ and $C_4$ such that when we run Algorithm~\ref{alg:private_JL} over $[X;\vec y]$ with parameter $r$  w.p. $\geq 1-\alpha-\nu$ we (correctly) reject the null hypothesis using $\tilde p_0$ (i.e., Algorithm~\ref{alg:private_JL} returns $\mathtt{matrix~unaltered}$ and we can estimate $\tilde t_0$ and verify that indeed $\tilde p_0 < \alpha\cdot e^{-\tfrac {r-p}{n-p}}$) provided
\ifx\fullversion\undefined
$ r \geq p + \max\left\{C_1\frac{\sigma^2 (\tilde c_\alpha^2+\tilde\tau_\alpha^2)}{\beta_j^2\svmin(\Sigma)},~C_2\ln(1/\nu)\right\}\textrm{~,~and~~~}    n \geq \max\left\{r, ~ C_3  \frac {w^2} {\min\{\svmin(\Sigma),\sigma^2\}}, ~ C_4 p \ln(1/\nu)\right\} $
\else
\[  r \geq p + \max\left\{C_1\frac{\sigma^2 (\tilde c_\alpha^2+\tilde\tau_\alpha^2)}{\beta_j^2\svmin(\Sigma)},~C_2\ln(1/\nu)\right\}\textrm{,~and~}    n \geq \max\left\{r, ~ C_3  \frac {w^2} {\min\{\svmin(\Sigma),\sigma^2\}}, ~ C_4 (p+ \ln(1/\nu))\right\} \]
\fi
where $\tilde c_\alpha$, $\tilde\tau_\alpha$ defined s.t.
$\Pr_{X \sim T_{r-p}}[X > \tilde c_\alpha / e^{\tfrac {r-p}{n-p}}] = \Pr_{X\sim \gaussian(0,1)}[X > \tilde \tau_\alpha / e^{\tfrac {r-p}{n-p}}] = \tfrac \alpha 2 e^{-\tfrac {r-p}{n-p}}$.
\end{theorem}

\omitfromversion{
}
\DeclareRobustCommand{\comparisonboundproof}{ 
\begin{proof}
First we need to use the lower bound on $n$ to show that indeed Algorithm~\ref{alg:private_JL} does not alter $A$, and that various quantities are not far from their expected values. Formally, we claim the following.

\begin{proposition}
\label{pro:n_large_enough}
Under the same lower bounds on $n$ and $r$ as in Theorem~\ref{thm:JL_simple_case_rejecting_null_hypothesis}, w.p. $1-\alpha-\nu$ we have that Theorem~\ref{thm:main_theorem_matrix_unaltered} holds and also that
\ifx\twocols\undefined 
\begin{align*}\|\tvec\zeta\|^2 = \Theta(\tfrac{r-p}{r} \| P_{U^\perp}\vec e\|^2) = \Theta(\tfrac{r-p}{r}  (n-p)\sigma^2) \textrm{~, and~ } (X^\T R^\T R X)^{-1}_{j,j} = \Theta(\tfrac 1 {r-p} (X^\T X)^{-1}_{j,j})\end{align*}
\else \[\tvec\zeta\|^2 = \Theta(\tfrac{r-p}{r} \| P_{U^\perp}\vec e\|^2) = \Theta(\tfrac{r-p}{r}  (n-p)\sigma^2)\] and \[ (X^\T R^\T R X)^{-1}_{j,j} = \Theta(\tfrac 1 {r-p} (X^\T X)^{-1}_{j,j})\]
\fi
\end{proposition}
\begin{proof}[Proof of Proposition~\ref{pro:n_large_enough}.]
 First, we need to argue that we have enough samples as to have the gap $\svmin^2([X;y]) - w^2$ sufficiently large. 

Since $\vec x_i\sim\gaussian(0,\Sigma)$, and $y_i = \vec\beta ^\T \vec x_i + e_i$ with $e_i\sim\gaussian(0,\sigma^2)$, we have that the concatenation $(\vec x_i \circ y_i)$ is also sampled from a Gaussian. Clearly, $\E[y_i] = \vec\beta^\T \E[\vec x_i] + \E[e_i] = 0$. Similarly, $\E[x_{i,j} y_i] = \E [x_{i,j}\cdot (\vec\beta^\T \vec x_i + e_i)] = (\Sigma \vec\beta)_j$ and $\E[y_i^2] =  \E[e_i^2] + \E[ \|X\vec\beta\|^2] = \sigma^2 + \E[\vec\beta^\T X^\T X \vec\beta] = \sigma^2 + \vec\beta^\T\Sigma\vec\beta$. Therefore, each row of $A$ is an i.i.d sample of $\gaussian(\vec 0_{p+1}, \Sigma_A)$, with 
\[\Sigma_A = \left( \begin{array}{c|c}
\Sigma~~~ & \Sigma\vec\beta \cr \hline \vphantom{T^{T^T}}  \vec\beta^\T\Sigma~~~ & \scriptstyle{\sigma^2 + \vec\beta^\T \Sigma\vec\beta}
\end{array}\right)
\] Denote $\lambda^2 =\svmin(\Sigma)$. Then, to argue that $\svmin(\Sigma_A)$ is large we use the lower bound from~\mycite{MaZ95} (Theorem 3.1) 
\ifx\twocols\undefined 
to argue that:
\newcommand{\bsb}{\vec\beta^\T \Sigma\vec\beta}
\begin{eqnarray*}
&\svmin(\Sigma_A) &\geq \frac {(\bsb + \sigma^2)+\lambda^2} 2 - \sqrt{\frac{((\bsb+\sigma^2)+\lambda^2)^2}{4} - \Big((\bsb+\sigma^2)-(\vec\beta^\T \Sigma) \Sigma^{-1} (\Sigma \vec \beta)\Big)\lambda^2}
\cr
&& = \frac {\bsb + \sigma^2+\lambda^2} 2 - \sqrt{\frac {(\bsb+\sigma^2+\lambda^2)^2 - 4\lambda^2(\bsb+\sigma^2-\bsb)} 4 }\cr
&& = \frac {\bsb + \sigma^2+\lambda^2} 2 - \sqrt{\frac {(\bsb)^2 +2\bsb(\sigma^2+\lambda^2) + (\sigma^2+\lambda^2)^2 - 4\lambda^2\sigma^2} 4 }\cr
&& = \frac {\bsb + \sigma^2+\lambda^2} 2 - \sqrt{\frac {(\bsb)^2 +2\bsb(\sigma^2+\lambda^2) + (\sigma^2-\lambda^2)^2} 4 }\cr
&&= \frac {\bsb + \sigma^2+\lambda^2} 2 - \sqrt{\frac {(\bsb)^2 +2\bsb|\sigma^2-\lambda^2|+ (\sigma^2-\lambda^2)^2 + 4\bsb \min\{\lambda^2,\sigma^2\}} 4 } \cr
&& \geq \frac {\bsb + \sigma^2+\lambda^2} 2 - \sqrt{\frac {(\bsb)^2 +2\bsb |\sigma^2-\lambda^2| + (\sigma^2-\lambda^2)^2} 4 } \cr
&& = \frac {\bsb + \sigma^2+\lambda^2} 2 - \sqrt{\frac {(\bsb + |\sigma^2 - \lambda^2|)^2} 4 } \cr
&& = \frac {\bsb + \sigma^2+\lambda^2} 2 - \frac {\bsb + |\sigma^2-\lambda^2|} 2  \geq \min\{ \lambda^2,\sigma^2 \} = \min\{ \svmin(\Sigma), \sigma^2\}
\end{eqnarray*}
\else
combining with some simple arithmetic manipulations to deduce that $\svmin(\Sigma_A) \geq \min\{ \svmin(\Sigma), \sigma^2\}$.
\fi

Having established a lower bound on $\svmin(\Sigma_A)$, it follows that with $n = \Omega(p\ln(1/\nu))$ i.i.d draws from $\gaussian(\vec 0_{p+1},\Sigma_A)$ we have w.p. $\leq \nu/4$ that $\svmin(A^\T A) = o(n)\cdot\min\{\svmin(\Sigma), \sigma^2\}$. Conditioned on $\svmin(A^\T A) = \Omega( n \svmin(\Sigma_A)) = \Omega(w^2)$ being large enough, we have that w.p. $\leq \nu/4$ over the randomness of Algorithm~\ref{alg:private_JL} the matrix $A$ does not pass the if-condition and the output of the algorithm is not $RA$.  Conditioned on Algorithm~\ref{alg:private_JL} outputting $RA$, and due to the lower bound $r = p + \Omega(\ln(1/\nu))$, we have that the result of Theorem~\ref{thm:main_theorem_matrix_unaltered} does not hold w.p. $\leq \alpha+ \nu/4$. All in all we deduce that w.p. $\geq 1-\alpha-3\nu/4$ the result of Theorem~\ref{thm:main_theorem_matrix_unaltered} holds. 
And since we argue Theorem~\ref{thm:main_theorem_matrix_unaltered} holds, then the following two bounds that are used in the proof\footnote{More accurately, both are bounds shown in Claim~\ref{clm:comparability_sigmaXX_lMM}.} also hold: 
\begin{align*}
&(X^\T R^\T R X)^{-1}_{j,j} = \Theta(\tfrac 1 {r-p} (X^\T X)^{-1}_{j,j})\cr
&\| P_{U^\perp}\vec e\|^2 = \Theta((n-p)\sigma^2)
\end{align*} 
Lastly, in the proof of Theorem~\ref{thm:main_theorem_matrix_unaltered} we argue that for a given $P_{U^\perp}\vec e$ the length $\|\tvec\zeta\|^2$ is distributed like $\tfrac {\| P_{U^\perp}\vec e\|^2} {r}\chi^2_{r-p}$. Appealing again to the fact that $r = p + \Omega(\ln(1/\nu)$ we have that  w.p. $\geq \nu/4$ it holds that $\|\tvec \zeta\|^2 > 2(r-p)\tfrac{\| P_{U^\perp}\vec e\|^2}{r}$. Plugging in the value of $\|P_{U^\perp}\vec e\|^2$ concludes the proof of the proposition.
\end{proof}

Based on Proposition~\ref{pro:n_large_enough}, we now show that we indeed reject the null-hypothesis (as we should). 
When Theorem~\ref{thm:main_theorem_matrix_unaltered} holds, reject the null-hypothesis iff $\tilde p_0 < \alpha \cdot e^{-\tfrac {r-p}{n-p}}$ which holds iff $|\tilde t_0| > e^{\tfrac {r-p}{n-p}} \tilde \tau_\alpha$. This implies we reject that null-hypothesis when $|\tilde\beta_j| >  e^{\tfrac {r-p}{n-p}}\tilde \tau_\alpha\cdot \tilde\sigma \sqrt{(X^\T R^\T R X)^{-1}_{j,j}})$. Note that this bound is based on Corollary~\ref{cor:JL_confidence_interval} that determines that $|\tilde\beta_j - \beta_j| = O \left( e^{\tfrac {r-p}{n-p}}\tilde c_\alpha\cdot \tilde\sigma \sqrt{(X^\T R^\T R X)^{-1}_{j,j}})\right)$. And so we have that w.p. $\geq 1-\nu$ we $\alpha$-reject the null hypothesis when it holds that $|\beta_j| > 3(\tilde c_\alpha+\tilde \tau_\alpha)\cdot \tilde\sigma \sqrt{(X^\T R^\T R X)^{-1}_{j,j}}) \geq e^{\tfrac {r-p}{n-p}}(\tilde c_\alpha+\tilde \tau_\alpha) \tilde\sigma \sqrt{(X^\T R^\T R X)^{-1}_{j,j}})$ (due to the lower bound $n\geq r$).

Based on the bounds stated above we have that  \[\tilde\sigma = \|\tvec \zeta\| \sqrt{\tfrac{r}{r-p}} = \Theta(\sigma \sqrt{n-p} \sqrt{\tfrac{r-p}{r}} \sqrt{\tfrac r{r-p}}) = \Theta(\sigma\sqrt{n-p})\] and that \[(X^\T R^\T R X)^{-1}_{j,j} =\Theta (\tfrac 1 {r-p}(X^\T X)^{-1}_{j,j}) = O\left( \tfrac 1 {r-p} \cdot \tfrac 1 {n\svmin(\Sigma)}\right)\] And so, a sufficient condition for rejecting the null-hypothesis is to have
\ifx\twocols\undefined
\[ |\beta_j| =  \Omega\left( (\tilde c_\alpha+\tilde \tau_\alpha) \sigma \sqrt{ \frac {n-p}{r-p}} \cdot \sqrt{\tfrac 1 {n\svmin(\Sigma)}} \right) = \Omega(e^{\tfrac {r-p}{n-p}}(\tilde c_\alpha+\tilde \tau_\alpha) \tilde\sigma \sqrt{(X^\T R^\T R X)^{-1}_{j,j}})) \]
\else 
\begin{eqnarray*} |\beta_j| &&=  \Omega\left( (\tilde c_\alpha+\tilde \tau_\alpha) \sigma \sqrt{ \frac {n-p}{r-p}} \cdot \sqrt{\tfrac 1 {n\svmin(\Sigma)}} \right) \cr&&= \Omega(e^{\tfrac {r-p}{n-p}}(\tilde c_\alpha+\tilde \tau_\alpha) \tilde\sigma \sqrt{(X^\T R^\T R X)^{-1}_{j,j}})) \end{eqnarray*}
\fi 
which, given the lower bound $r = p + \Omega\left( \frac{(\tilde c_\alpha+\tilde \tau_\alpha)^2 \sigma^2} {\beta_j^2 \svmin(\Sigma)} \right)$ indeed holds.
\end{proof}
}

\subsection{Setting the Value of $r$, Deriving a Bound on $n$}
\label{subsec:value_of_r}

Comparing the lower bound on $n$ given by Theorem~\ref{thm:JL_simple_case_rejecting_null_hypothesis} to the bound of Theorem~\ref{thm:baseline_rejecting_nh}, we have that 
the data-dependent bound of $\Omega\left( \frac{(\tilde c_\alpha+\tilde \tau_\alpha)^2 \sigma^2} {\beta_j^2 \svmin(\Sigma)} \right)$ should now hold for $r$ rather than $n$. Yet, Theorem~\ref{thm:JL_simple_case_rejecting_null_hypothesis} also introduces an additional dependency between $n$ and $r$: we require $n = \Omega( \tfrac {w^2}{\sigma^2} + \tfrac {w^2}{\svmin(\Sigma)})$ (since otherwise we do not have $\svmin(A) \gg w$ and Algorithm~\ref{alg:private_JL} might alter $A$ before projecting it) and by definition $w^2$ is proportional to $\sqrt{r\ln(1/\delta)}/\epsilon$. This is precisely the focus of our discussion in this subsection. We would like to set $r$'s value as high as possible~--- the larger $r$ is, the more observations we have in $RA$ and the better our confidence bounds (that depend on $T_{r-p}$) are~--- while satisfying $n = \Omega(\tfrac {\sqrt{r}}{\epsilon\cdot \min\{\sigma^2, \svmin(\Sigma)\}})$.

Recall that if each sample point is drawn i.i.d $\vec x\sim \gaussian(\vec 0_p,\Sigma)$, then each sample $(\vec x_i\circ y_i)$ is sampled from $\gaussian(\vec 0_{p+1}, \Sigma_A)$ for $\Sigma_A$ defined in the proof of Theorem~\ref{thm:JL_simple_case_rejecting_null_hypothesis}, that is: $\Sigma_A = \left( \begin{array}{c|c}
\Sigma~~~ & \Sigma\vec\beta \cr \hline \phantom{^{T^T}}\vec\beta^\T\Sigma~~~ & \scriptstyle{\sigma^2 + \vec\beta^\T \Sigma\vec\beta}
\end{array}\right)$. So, Theorem~\ref{thm:JL_simple_case_rejecting_null_hypothesis} gives the lower bound $r-p = \Omega\left( \frac{\sigma^2(\tilde c_\alpha+\tilde \tau_\alpha)^2}{\beta_j^2 \svmin(\Sigma)} \right)$ and the following lower bounds on $n$:
\ifx\fullversion\undefined
$n\geq r$ and $n = \Omega\left( \frac{B^2(\sqrt{r\ln(1/\delta)} + \ln(1/\delta))}{\epsilon \svmin(\Sigma_A)} \right)$, which means $r = \min\left\{n, \frac{\epsilon^2\svmin^2(\Sigma_A)}{B^4\ln(1/\delta)} (n-\ln(1/\delta))^2\right\}$.
\else
\begin{align*}
& n-p \geq r-p \textrm{~, and~ } n = \Omega\left( \frac{B^2(\sqrt{r\ln(1/\delta)} + \ln(1/\delta))}{\epsilon \svmin(\Sigma_A)} \right) \Rightarrow r = \min\left\{n, \frac{\epsilon^2\svmin^2(\Sigma_A)}{B^4\ln(1/\delta)} (n-\ln(1/\delta))^2\right\}
\end{align*}
\fi
This discussion culminates in the following corollary.
\begin{corollary}
	\label{cor:bound_n}
Denoting $\widetilde{LB_{\ref{thm:baseline_rejecting_nh}}} = \tfrac{\sigma^2(\tilde c_\alpha+\tilde \tau_\alpha)^2}{\beta_j^2 \svmin(\Sigma)}$, we thus conclude that if
\ifx\fullversion\undefined
$n-p \geq  \Omega\left( \widetilde{LB_{\ref{thm:baseline_rejecting_nh}}} \right)$ and $n = \Omega\left( \frac {B^2 \ln(1/\delta)}{\epsilon\svmin(\Sigma_A)} \cdot \sqrt{\widetilde{LB_{\ref{thm:baseline_rejecting_nh}}}} \right)$, 
\else
\begin{align*}
& n-p \geq  \Omega\left( \widetilde{LB_{\ref{thm:baseline_rejecting_nh}}} \right) \textrm{~, and~ } n = \Omega\left( \frac {B^2 \ln(1/\delta)}{\epsilon\svmin(\Sigma_A)} \cdot \sqrt{\widetilde{LB_{\ref{thm:baseline_rejecting_nh}}}} \right) 
\end{align*}
\fi
then the result of Theorem~\ref{thm:JL_simple_case_rejecting_null_hypothesis} holds by setting $r= \min\left\{n, \frac{\epsilon^2\svmin^2(\Sigma_A)}{B^4\ln(1/\delta)} (n-\ln(1/\delta))^2\right\}$.
\end{corollary}

It is interesting to note that when we know $\Sigma_A$, we also have a bound on $B$. Recall $\Sigma_A$, the variance of the Gaussian $(\vec x \circ y)$. Since every sample is an independent draw from $\gaussian(\vec 0_{p+1},\Sigma_A)$ then we have an upper bound of $B^2 \leq \log(np)\svmax(\Sigma_A)$. So our lower bound on $n$ (using $\kappa(\Sigma_A)$ to denote the condition number of $\Sigma_A$) is given by
\ifx\fullversion\undefined
$n \geq \max \left\{  \Omega\left(\widetilde{LB_{\ref{thm:baseline_rejecting_nh}}}\right), \tilde\Omega\left(\tfrac {\kappa(\Sigma_A)\ln(1/\delta)}{\epsilon} \cdot \sqrt{\widetilde{LB_{\ref{thm:baseline_rejecting_nh}}}} \right) \right\}$.
\else
\[ n \geq \max \left\{ \Omega\left(p+\ln(1/\nu)\right),~~ \Omega\left(\widetilde{LB_{\ref{thm:baseline_rejecting_nh}}}\right), ~~\tilde\Omega\left(\frac {\kappa(\Sigma_A)\ln(1/\delta)}{\epsilon} \cdot \sqrt{\widetilde{LB_{\ref{thm:baseline_rejecting_nh}}}} \right) \right\}\]
\fi
\omitfromversion{
Note that if we have no apriori bound on $\svmin(A)$, then, much like it is done in Algorithm~\ref{alg:private_JL}, we can privately estimate $\lambda = \svmin(A^\T A)+Z$ by adding Laplace noise $Z\sim Lap( 4 B^2/\epsilon )$. We now have that w.p. $\geq 1-\nu$ it holds that $\svmin(A^\T A) \geq \lambda -  4 B^2 \ln(1/\nu)/\epsilon \stackrel {\rm def} =\underline \lambda$. We then upper bound $r$ using $n$ and $\underline{\lambda}$ replacing $\svmin(\Sigma_A)$.
}
Observe, overall this result is similar in nature to many other results in differentially private learning \mycite{BassilyST14} which are of the form 
``without privacy, in order to achieve a total loss of $\leq \eta$ we have a sample complexity bound of some $N_\eta$; and with differential privacy the sample complexity increases to $N_\eta + \Omega(\sqrt{N_\eta}/\epsilon)$.'' However, there's a subtlety here worth noting. $\widetilde{LB_{\ref{thm:baseline_rejecting_nh}}}$ is proportional to $\tfrac 1 {\svmin(\Sigma_A)}$ but not to $\kappa(\Sigma_A) = \tfrac {\svmax(\Sigma_A)}{\svmin(\Sigma_A)}$. The additional dependence on $\svmax$ follows from the fact that differential privacy adds noise proportional to the upper bound on the norm of each row.

\myvspace{-8pt}
\section{Projected Ridge Regression}
\label{sec:ridge_regression}
\myvspace{-3pt}
We now turn to deal with the case that our matrix does not pass the if-condition of Algorithm~\ref{alg:private_JL}. In this case, the matrix is appended with a $d\times d$-matrix which is $wI_{d\times d}$. Denoting $A' = \left[ \begin{array}{c} A \cr w\cdot I_{d\times d}\end{array}\right]$ we have that the algorithm's output is $RA'$. 
Similarly to before, we are going to denote $d = p+1$ and decompose  $A = [X; \vec y]$ with $X\in \mathbb{R}^{n\times p}$ and $\vec y \in \mathbb{R}^n$, with the standard assumption of $\vec y = X\vec\beta + \vec e$ and $e_i$ sampled i.i.d from $\gaussian(0,\sigma^2)$. We now need to introduce some additional notation. We denote the appended matrix and vectors $X'$ and $\vec y'$ s.t. $A' = [X';\vec y']$. 
\ifx\fullversion\undefined
\else
Meaning:
\newcommand{\thisnotation}{X' = \left[ \begin{array}{c}
X \cr w I_{p\times p} \cr \vec 0_p^\T
\end{array} \right]   \textrm{~, and } ~~ \vec y' = \left[\begin{array}{c}
\vec y \cr \vec 0_p \cr w
\end{array} \right] = \left[\begin{array}{c}
\vec X\vec\beta + \vec e \cr \vec 0_p \cr w
\end{array} \right] = X'\vec\beta + \left[\begin{array}{c}
\vec e \cr -w\vec \beta \cr w
\end{array} \right] \stackrel{\rm def}= X'\vec\beta + \vec e'}
\[ \thisnotation
\]
And so we respectively denote $R = \left[ R_1 ; R_2 ; R_3 \right]$ with $R_1\in \mathbb{R}^{r\times n}$, $R_2\in \mathbb{R}^{r\times p}$ and $R_3\in \mathbb{R}^{r\times 1}$ (so $R_3$ is a vector denoted as a matrix). Hence: 
\[ M' = RX' = R_1 X + w R_2 \textrm{~, and }~~ R\vec y' = RX'\vec\beta + R\vec e' = R_1\vec y + w R_3 = R_1X\vec \beta + R_1 \vec e + wR_3
\] 
\fi
And so, using the output $RA'$ of Algorithm~\ref{alg:private_JL}, we solve the linear regression problem derived from $\tfrac 1 {\sqrt r}RX'$ and $\tfrac 1 {\sqrt r}R\vec y'$. I.e., we set
\begin{eqnarray} &&\vec\beta' =  (X'^\T R^\T R X')^{-1}(RX')^\T (R\vec y') 
\cr
&&\vec\zeta' = \tfrac 1 {\sqrt{r}} (R\vec y' - RX'\vec\beta') \label{eq:jl_ridge_regression_zeta}
\end{eqnarray}
Sarlos' results~\mycitet{Sarlos06} regarding the Johnson Lindenstrauss transform give that, when $R$ has sufficiently many rows, solving the latter optimization problem gives a good approximation for the solution of the optimization problem
\ifx\fullversion\undefined
$\vec\beta^R = \arg\min_{\vec z} \|\vec y' - X'\vec z\|^2 = \arg\min_{\vec z} \left(\| \vec y-X\vec z \|^2 + w^2\|\vec z\|^2 \right)$.
\else
\[ \vec\beta^R = \arg\min_{\vec z} \|\vec y' - X'\vec z\|^2 = \arg\min_{\vec z} \left(\| \vec y-X\vec z \|^2 + w^2\|\vec z\|^2 \right)    \] 
\fi
The latter problem is known as the Ridge Regression problem. Invented in the 60s~\mycite{Tikhonov63, HoerlK70}, Ridge Regression is often motivated from the perspective of penalizing linear vectors whose coefficients are too large. It is also often applied in the case where $X$ doesn't have full rank or is close to not having full-rank: one can show that the minimizer $\vec\beta^R = (X^\T X + w^2 I_{p\times p})^{-1}X^\T \vec y$ is the unique solution of the Ridge Regression problem and that the RHS is always well-defined. 

\cut{The original focus of Ridge Regression is on penalizing $\vec\beta^R$ for having large coefficients. Therefore, Ridge Regression actually poses a family of linear regression problems: $\min_{\vec z} \|y-X\vec z\| + \lambda \|\vec z\|^2$, where one may set $\lambda$ to be any non-negative scalar. And so, much of the literature on Ridge Regression is devoted to the art of fine-tuning this penalty term~--- either empirically or based on the $\lambda$ that yields the best risk: $ \|\E[\vec\beta^R]-\vec\beta\|^2+ \textrm{Var}(\vec\beta^R)$.\footnote{Ridge Regression, as opposed to OLS, does not yield an unbiased estimator. I.e., $\E[\vec\beta^R] \neq \vec\beta$.}
Here we propose a fundamentally different approach for the choice of the normalization factor~--- we set it so that solution of the regression problem would satisfy $(\epsilon,\delta)$-differential privacy (by projecting the problem onto a lower dimension).}

While the solution of the Ridge Regression problem might have smaller risk than the OLS solution, it is not known how to derive $t$-values and/or reject the null hypothesis under Ridge Regression (except for using $X$ to manipulate $\vec\beta^R$ back into $\hvec\beta = (X^\T X)^{-1} X^\T \vec y$ and relying on OLS). In fact, prior to our work there was no need for such analysis! For confidence intervals one could just use the standard OLS, because access to $X$ and $\vec y$ was given.  

Therefore, much for the same reason, we are unable to derive $t$-values under projected Ridge Regression.\footnote{Note: The na\"ive approach of using $RX'$ and $R\vec y'$ to interpolate $RX$ and $R\vec y$ and then apply Theorem~\ref{thm:main_theorem_matrix_unaltered} using these estimations of $RX$ and $R\vec y$ ignores the noise added from appending the matrix $A$ into $A'$, and therefore leads to inaccurate estimations of the $t$-values.} Clearly, there are situations where such confidence bounds simply cannot be derived.\ifx\fullversion\undefined\else(Consider for example the case where $X = 0_{n\times p}$ and $\vec y$ is just i.i.d draws from $\gaussian(0,\sigma^2)$, so obviously $[X;y]$ gives no information about $\vec\beta$.) \fi Nonetheless, under additional assumptions about the data, our work can give confidence intervals for $\beta_j$, and in the case where the interval doesn't intersect the origin~--- assure us that $sign(\beta_j') = sign(\beta_j)$ w.h.p. This is detailed in \ifx\confversion\undefined Section~\ref{apx_sec:ridge_regression}.\else the supplementary material.\fi

To give an overview of our analysis, we first discuss a model where $\vec e = \vec y - X\vec \beta$ is fixed (i.e., the data is fixed and the algorithm is the sole source of randomness), and prove that in this model $\vec{\beta'}$ is as an approximation to $\hvec{\beta}$.
\begin{theorem}
	\label{thm:approximating_hatbeta_body}
	Fix $X\in \mathbb{R}^{n\times p}$ and $\vec y\in \mathbb{R}$. Define $\hvec\beta = X^\mpi \vec y$ and $\zeta = (I- XX^\mpi)\vec y$. Let $RX'=M'$ and $R\vec y'$ denote the result of applying Algorithm~\ref{alg:private_JL} to the matrix $A = [X;\vec y]$ when the algorithm appends the data with a $w\cdot I$ matrix. Fix a coordinate $j$ and any $\alpha\in (0,1/2)$. When computing $\vec\beta'$ and $\vec\zeta'$ as in ~\eqref{eq:jl_ridge_regression_zeta}, we have that w.p. $\geq 1-\alpha$ it holds that
	$ \hat\beta_j \in \left( \beta'_j \pm c'_\alpha \|\vec \zeta'\|\sqrt{\tfrac r {r-p}\cdot (M'^\T M')^{-1}_{j,j}} \right)$ where $c'_\alpha$ denotes the number such that $(-c'_\alpha,c'_\alpha)$ contains $1-\alpha$ mass of the $T_{r-p}$-distribution.
\end{theorem}

However, our goal remains to argue that $\beta_j'$ serves as a good approximation for $\beta_j$. To that end, we combine the standard OLS confidence interval~--- which says that w.p. $\geq 1-\alpha$ over the randomness of picking $\vec e$ in the homoscedastic model we have $|\beta_j-\hat\beta_j| \leq c_\alpha \|\vec\zeta\|\sqrt{\tfrac{(X^\T X)^{-1}_{j,j}}{n-p}}$~--- with the confidence interval of Theorem~\ref{thm:approximating_hatbeta_body} above,
\ifx\confversion\undefined and deduce that
\begin{equation}
\Pr\left[ |\beta'_j - \beta_j| = O\left(c_\alpha \frac{\|\vec\zeta\|}{\sqrt{n-p}} \sqrt{(X^\T X)^{-1}_{j,j}} +  c'_\alpha \frac {\|\vec\zeta'\|} {\sqrt{r-p}} \sqrt{r(M'^\T M')^{-1}_{j,j}} \right)\right] \geq 1-\alpha \label{eq:dist_beta'j_betaj_body}
\end{equation}
\else
and denoting $I = c_\alpha \frac{\|\vec\zeta\|}{\sqrt{n-p}} \sqrt{(X^\T X)^{-1}_{j,j}} +  c'_\alpha \frac {\|\vec\zeta'\|} {\sqrt{r-p}} \sqrt{r(M'^\T M')^{-1}_{j,j}}$ we have that $\Pr[ |\beta'_j - \beta_j| = O(I)]\geq 1-\alpha$.
\fi
And so, in summary, in Section~\ref{apx_sec:ridge_regression} we give conditions under which  the length of \ifx\confversion\undefined the interval in Equation~\eqref{eq:dist_beta'j_betaj_body} \else the interval $I$ \fi  is dominated by  the $c'_\alpha \tfrac {\|\vec\zeta'\|} {\sqrt{r-p}} \sqrt{r(M'^\T M')^{-1}_{j,j}}$ factor derived from Theorem~\ref{thm:approximating_hatbeta_body}. 

\omitfromversion{
Clearly, Sarlos' work~\mycite{Sarlos06} gives an upper bound on the distance $\|\vec\beta' - \vec\beta^R\|$.  However, such distance bound doesn't come with the coordinate by coordinate confidence guarantee we would like to have. In fact, it is not even clear from Sarlos' work that $\E[\vec\beta'] = \vec\beta^R$ (though it is obvious to see that $\E[ (X'^\T R^\T R X')] \vec\beta^R = \E[(RX')^\T R\vec y' ]$). Here, we show that $\E[\vec\beta'] = \hvec\beta$ which, more often than not, does not equal $\vec\beta^R$.

\myparagraph{Comment about notation.} Throughout this section we assume $X$ is of full rank and so $(X^\T X)^{-1}$ is well-defined. If $X$ isn't full-rank, then one can simply replace any occurrence of $(X^\T X)^{-1}$ with $X^\mpi (X^\mpi)^\T$. This  makes all our formulas well-defined in the general case. 
\subsection{Running OLS on the Projected Data}
\label{subsec:linear_regression_on_projected_altered_data}

In this section, we analyze the projected Ridge Regression, under the assumption (for now) that $\vec e$ is fixed. That is, for now we assume that the only source of randomness comes from picking the matrix $R = [R_1 ; R_2 ; R_3]$. As before, we analyze the distribution over $\vec\beta'$ (see Equation~\eqref{eq:jl_ridge_regression_problem}), and the value of the function we optimize at $\vec\beta'$. Denoting $M'=RX'$, we can formally express the estimators:
\begin{eqnarray}
&\vec\beta' &= (M'^\T M')^{-1} M'^\T R\vec y' \label{eq:beta'}\\
&\vec\zeta' &= \tfrac 1 {\sqrt{r}} (R\vec y' - RX'\vec\beta') \label{eq:zeta'}
\end{eqnarray}
\begin{claim}
\label{clm:dist_beta'_zeta'}
Given that $\vec y = X\vec \beta + \vec e$ for a fixed $\vec e$, and given $X$ and $M' = RX'= R_1X + wR_2$ we have that
\begin{align*}
&& \vec\beta' &\sim \gaussian\left( \vec\beta + X^\mpi\vec e,  ~~  (w^2(\|\vec\beta+X^\mpi\vec e\|^2+1)  + \|P_{U^\perp}\vec e\|^2 ) (M'^\T M')^{-1} \right)\cr
&&\vec\zeta' &\sim \gaussian\left( \vec 0_r, ~~\tfrac 1 r (w^2(\|\vec\beta+X^\mpi\vec e\|^2+1)  + \|P_{U^\perp}\vec e\|^2 ) (I_{r\times r}-M'M'^\mpi) \right)
\end{align*}
and furthermore, $\vec\beta'$ and $\vec\zeta'$ are independent of one another.
\end{claim}
\begin{proof}
First, we write $\vec\beta'$ and $\vec \zeta'$ explicitly, based on $\vec e$ and projection matrices:
\begin{eqnarray*}
&\vec\beta' &= (M'^\T M')^{-1} M'^\T R\vec y' = M'^\mpi (R_1 X)\vec\beta + M'^\mpi (R_1\vec e + wR_3) \\
&\vec\zeta' &= \tfrac 1 {\sqrt{r}} (R\vec y' - RX'\vec\beta') = \tfrac 1 {\sqrt{r}}(I_{r\times r}-M'M'^\mpi) R\vec e'= \tfrac 1 {\sqrt r} P_{U'^\perp}(R_1\vec e - wR_2\vec\beta + wR_3)
\end{eqnarray*}with $U'$ denoting $colspan(M')$ and $P_{U'^\perp}$ denoting the projection onto the subspace $U'^\perp$.

Again, we break $\vec e$ into an orthogonal composition: $\vec e = P_U\vec e + P_{U^\perp}\vec e$ with $U = colspan(X)$ (hence $P_{U} = XX^\mpi$) and $U^\perp = colspan(X)^\perp$. Therefore,
\begin{eqnarray}
&\vec\beta' &= M'^\mpi (R_1 X)\vec\beta + M'^\mpi (R_1X X^\mpi\vec e + R_1P_{U^\perp}\vec e + wR_3) \cr
&&= M'^\mpi(R_1 X)(\vec\beta + X^\mpi \vec e) +M'^\mpi (R_1P_{U^\perp}\vec e + wR_3)  \\
&\vec\zeta' &= \tfrac 1 {\sqrt r} (I_{r\times r}- M'M'^\mpi)(R_1X X^\mpi\vec e + R_1 P_{U^\perp}\vec e- wR_2\vec\beta + wR_3)\cr
&& \stackrel{(\ast)}=\tfrac 1 {\sqrt r} (I_{r\times r}- M'M'^\mpi)(R_1X X^\mpi\vec e + R_1 P_{U^\perp}\vec e + (M'- wR_2)\vec\beta + wR_3) \cr
&& =\tfrac 1 {\sqrt r} (I_{r\times r}- M'M'^\mpi)( R_1X (\vec \beta + X^\mpi\vec e) + R_1 P_{U^\perp}\vec e +  wR_3) 
\end{eqnarray}
where equality $(\ast)$ holds because $(I- M'M'^\mpi)M'\vec v = \vec 0$ for any  $\vec v$.

We now aim to describe the distribution of $R$ given that we know $X'$ and $M' = RX'$. Since 
\begin{align*}
&M' = R_1 X + wR_2 + 0\cdot R_3 = R_1 X (X^\mpi X) +w R_2 = (R_1 P_U) X + w R_2
\end{align*} then $M'$ is independent of $R_3$ and independent of $R_1P_{U^\perp}$. Therefore, given $X$ and $M'$ the induced distribution over $R_3$ remains $R_3\sim\gaussian(\vec 0_r, I_{r\times r})$, and similarly, given $X$ and $M'$ we have $R_1P_{U^\perp} \sim \gaussian(0_{r\times n}, I_{r\times r}, P_{U^\perp})$ (rows remain independent from one another, and each row is distributed like a spherical Gaussian in $colspan(X)^\perp$). And so, we have that $R_1X = R_1 P_U X = M' - wR_2$, which in turn implies: 
\begin{align*}
&& R_1X &\sim \gaussian\left( M', ~~ I_{r\times r}, ~~w^2\cdot I_{p\times p}\right) \cr
&&\Rightarrow R_1X(\vec\beta + X^\mpi\vec e) &\sim \gaussian\left(M'\vec\beta + M'X^\mpi\vec e, ~~ w^2\|\vec\beta + X^{\mpi}\vec e\|^2 I_{r\times r}\right)\cr
&&\Rightarrow M'^\mpi R_1X(\vec\beta + X^{\mpi}\vec e) &\sim \gaussian\left(\vec\beta+ X^\mpi\vec e, ~~ w^2\|\vec\beta+X^\mpi\vec e\|^2  (M'^\T M)^{-1}\right)\cr
&& & = \|\vec\beta + X^\mpi \vec e\|\cdot \gaussian( \vec u, w^2 (M'^\T M)^{-1})
\end{align*}
where $\vec u$ denotes a unit-length vector in the direction of $\vec\beta + X^\mpi \vec e$. 

Similar to before we have
\begin{align*}
& RP_{U^\perp} \sim \gaussian(0_{r\times n}, I_{r\times r}, P_{U^\perp}) 
&& \Rightarrow M'^\mpi (R P_{U^\perp} \vec e) \sim \gaussian (\vec 0_d, ~ \|P_{U^\perp} e\|^2 (M'^\T M')^{-1}) 
\cr & wR_3 \sim \gaussian(\vec 0_r, w^2 I_{r\times r}) &&\Rightarrow M'^\mpi(wR_3) \sim \gaussian( \vec 0_d, w^2 (M'^\mpi M')^{-1})
\end{align*}
Therefore, the distribution of  $\vec\beta'$, which is the sum of the $3$ independent Gaussians, is as required.

Similarly, $\vec\zeta' = \tfrac 1 {\sqrt r} P_{U'^\perp} \left( R_1 X(\vec\beta + X^\mpi\vec e) + R_1 P_{U^\perp}\vec e + wR_3 \right)$ is the sum of $3$ independent Gaussians, which implies its distribution is
\[ \gaussian\left(\tfrac 1 {\sqrt r}P_{U'^\perp}M'(\vec\beta+X^\mpi\vec e), ~~ \tfrac 1 r (w^2(\|\vec\beta+X^\mpi\vec e\|^2 +1) + \|P_{U^\perp}\vec e\|^2)P_{U'^\perp} \right) \] which is exactly $\gaussian\left(\vec 0_r, ~~ \tfrac 1 r (w^2(\|\vec\beta+X^\mpi\vec e\|^2 +1) + \|P_{U^\perp}\vec e\|^2)P_{U'^\perp} \right)$ as $P_{U'^\perp} M'= 0_{r\times r}$.

Finally, observe that $\vec\beta'$ and $\vec \zeta'$ are independent as the former depends on the projection of the spherical Gaussian $R_1 X(\beta + X^\mpi\vec e) + R_1 P_{U^\perp}\vec e + wR_3$ on $U'$, and the latter depends on the projection of the same multivariate Gaussian on $U'^\perp$.
\end{proof}

Observe that Claim~\ref{clm:dist_beta'_zeta'} assumes $\vec e$ is given. This may seem somewhat strange, since without assuming anything about $\vec e$ there can be many combinations of $\vec\beta$ and $\vec e$ for which $\vec y = X\vec\beta + \vec e$. However, we always have that $\vec\beta + X^\mpi \vec e = X^\mpi\vec y = \hvec\beta$. Similarly, it is always the case the $P_{U^\perp}\vec e = (I-XX^\mpi) \vec y = \vec\zeta$. (Recall OLS definitions of $\hvec\beta$ and $\vec\zeta$ in Equation~\eqref{eq:hat_beta} and~\eqref{eq:zeta}.) Therefore, the distribution of $\vec\beta'$ and $\vec\zeta'$ is unique (once $\vec y$ is set):
\begin{align*}
&& \vec\beta' &\sim \gaussian\left( \hvec\beta ,    (w^2(\|\hvec\beta\|^2+1)  + \|\vec\zeta\|^2  ) (M'^\T M')^{-1} \right)\cr
&&\vec\zeta' &\sim \gaussian\left( \vec 0_r, \tfrac 1 r\cdot (w^2(\|\hvec\beta\|^2+1)  + \|\vec\zeta\|^2 ) (I_{r\times r}-M'M'^\mpi) \right)
\end{align*}
And so for a given dataset $[X;\vec y]$ we have that $\vec\beta'$ serves as an approximation for $\hvec\beta$.

An immediate corollary of Claim~\ref{clm:dist_beta'_zeta'} is that for any fixed $\vec e$ it holds that the quantity
$t'(\beta_j) = \frac {\beta'_j - (\beta_j + (X^\mpi\vec e)_j)} {   \|\vec \zeta'\|\sqrt{\tfrac r {r-p}\cdot (M'^\T M')^{-1}_{j,j}}  }=\frac {\beta'_j - \hat\beta_j} {   \|\vec \zeta'\|\sqrt{\tfrac r {r-p}\cdot (M'^\T M')^{-1}_{j,j}}  }$ is distributed like a $T_{r-p}$-distribution. Therefore, the following theorem therefore follows immediately.
\begin{theorem}
\label{thm:approximating_hatbeta}
Fix $X\in \mathbb{R}^{n\times p}$ and $\vec y\in \mathbb{R}$. Define $\hvec\beta = X^\mpi \vec y$ and $\zeta = (I- XX^\mpi)\vec y$. Let $RX'$ and $R\vec y'$ denote the result of applying Algorithm~\ref{alg:private_JL} to the matrix $A = [X;\vec y]$ when the algorithm appends the data with a $w\cdot I$ matrix. Fix a coordinate $j$ and any $\alpha\in (0,1/2)$. When computing $\vec\beta'$ and $\vec\zeta'$ as in Equations~\eqref{eq:beta'} it and~\eqref{eq:zeta'}, we have that w.p. $\geq 1-\alpha$ it holds that
\[ \hat\beta_j \in \left( \beta'_j - c'_\alpha \|\vec \zeta'\|\sqrt{\tfrac r {r-p}\cdot (M'^\T M')^{-1}_{j,j}} , ~~\beta'_j + c'_\alpha \|\vec \zeta'\|\sqrt{\tfrac r {r-p}\cdot (M'^\T M')^{-1}_{j,j}}  \right)\] where $c'_\alpha$ denotes the number such that $(-c'_\alpha,c'_\alpha)$ contains $1-\alpha$ mass of the $T_{r-p}$-distribution.
\end{theorem}

Note that Theorem~\ref{thm:approximating_hatbeta}, much like the rest of the discussion in this Section, builds on $\vec y$ being fixed, which means $\beta_j'$ serves as an approximation for $\hat\beta_j$. Yet our goal is to argue about similarity (or proximity) between $\beta_j'$ and $\beta_j$. To that end, we combine the standard OLS confidence interval~--- which says that w.p. $\geq 1-\alpha$ over the randomness of picking $\vec e$ in the homoscedastic model we have $|\beta_j-\hat\beta_j| \leq c_\alpha \|\vec\zeta\|\sqrt{\tfrac{(X^\T X)^{-1}_{j,j}}{n-p}}$~--- with the confidence interval of Theorem~\ref{thm:approximating_hatbeta} above, and deduce that
\begin{equation}
\Pr\left[ |\beta'_j - \beta_j| = O\left( c_\alpha \frac{\|\vec\zeta\|}{\sqrt{n-p}} \sqrt{(X^\T X)^{-1}_{j,j}} +  c'_\alpha \frac {\|\vec\zeta'\|} {\sqrt{r-p}} \sqrt{r(M'^\T M')^{-1}_{j,j}}\right)\right] \geq 1-\alpha \label{eq:dist_beta'j_betaj}
\end{equation}
\footnote{Observe that w.p. $\geq 1- \alpha$ over the randomness of $\vec e$ we have that $|\beta_j-\hat\beta_j| \leq c_\alpha \|\vec\zeta\|\sqrt{\tfrac{(X^\T X)^{-1}_{j,j}}{n-p}}$, and w.p. $\geq1-\alpha$ over the randomness of $R$ we have that $|\beta_j'-\hat\beta_j|\leq c'_\alpha \|\vec \zeta'\|\sqrt{\tfrac r {r-p}\cdot (M'^\T M')^{-1}_{j,j}}$. So technically, to give a $(1-\alpha)$-confidence interval around $\beta_j'$ that contains $\beta_j$ w.p. $\geq 1-\alpha$, we need to use $c_{\alpha/2}$ and $c'_{\alpha/2}$ instead of $c_\alpha$ and $c'_\alpha$ resp. To avoid overburdening the reader with what we already see as too many parameters, we switch to asymptotic notation.}And so, in the next section, our goal is to give conditions under which the interval of Equation~\eqref{eq:dist_beta'j_betaj} isn't much larger in comparison to the interval length of $c'_\alpha \tfrac {\|\vec\zeta'\|} {\sqrt{r-p}} \sqrt{r(M'^\T M')^{-1}_{j,j}}$ we get from Theorem~\ref{thm:approximating_hatbeta}; and more importantly~--- conditions that make the interval of Theorem~\ref{thm:approximating_hatbeta} useful and not too large. (Note, in expectation $\tfrac {\|\vec\zeta'\|} {\sqrt{r-p}}$ is about $\sqrt{(w^2 + w^2\|\hvec\beta\|^2 + \|\vec\zeta\|^2)/r}$. So, for example, in situations where $\|\hvec\beta\|$ is very large, this interval isn't likely to inform us as to the sign of $\beta_j$.) 

\myparagraph{Motivating Example.} A good motivating example for the discussion in the following section is when $[X;\vec y]$ is a strict submatrix of the dataset $A$. That is, our data contains many variables for each entry (i.e., the dimensionality $d$ of each entry is large), yet our regression is made only over a modest subset of variables out of the $d$. In this case, the least singular value of $A$ might be too small, causing the algorithm to alter $A$; however, $\svmin(X^\T X)$ could be sufficiently large so that had we run Algorithm~\ref{alg:private_JL} only on $[X;\vec y]$ we would not alter the input. (Indeed, a differentially private way for finding a subset of the variables that induce a submatrix with high $\svmin$ is an interesting open question, partially answered~--- for a single regression~--- in the work of Thakurta and Smith~\mycite{ThakurtaS13}.) Indeed, the conditions we specify in the following section depend on $\svmin(\tfrac 1 n X^\T X)$, which, for a zero-mean data, the minimal variance of the data in any direction. For this motivating example, indeed such variance isn't necessarily small.

\subsection{Conditions for Deriving a Confidence Interval for Ridge Regression}
\label{subsec:analyzing_beta'}

Looking at the interval specified in Equation~\eqref{eq:dist_beta'j_betaj}, we now give an upper bound on the the random quantities in this interval: $\|\vec\zeta\|, \|\vec\zeta'\|$, and $(M'^\T M')^{-1}_{j,j}$. First, we give bound that are dependent on the randomness in $R$ (i.e., we continue to view $\vec e$ as fixed). 
\begin{proposition}
\label{pro:zeta'_M'TM'}
For any $\nu\in(0,1/2)$, if we have $r = p + \Omega(\ln(1/\nu))$ then with probability $\geq 1-\nu$ over the randomness of $R$ we have $(r-p)(M'^\T M)^{-1}_{j,j} = \Theta\left( (w^2 I_{p\times p} + X^\T X)^{-1}_{j,j}\right)$ and $\tfrac {\|\vec\zeta'\|^2}{r-p} = \Theta( \tfrac{w^2 + w^2\|\hvec\beta\|^2 + \|\vec\zeta\|^2}r)$.
\end{proposition}
\begin{proof}
The former bound follows from known results on the Johnson-Lindenstrauss transform (as were shown in the proof of Claim~\ref{clm:comparability_sigmaXX_lMM}). The latter bound follows from standard concentration bounds of the $\chi^2$-distribution.
\end{proof}

Plugging in the result of Proposition~\ref{pro:zeta'_M'TM'} to Equation~\eqref{eq:dist_beta'j_betaj} we get that w.p. $\geq1-\nu$ 
\begin{equation}
 |\beta'_j - \beta_j| = O\left( c_\alpha \frac{\|\vec\zeta\|}{\sqrt{n-p}} \sqrt{(X^\T X)^{-1}_{j,j}} +  c'_\alpha \sqrt{\frac {w^2 + w^2\|\hvec\beta\|^2 + \|\vec\zeta\|^2}{r-p} } \sqrt{(w^2 I_{p\times p} + X^\T X)^{-1}_{j,j}}\right) \label{eq:dist_beta'j_betaj_revised}
\end{equation} 
We will also use the following proposition.
\begin{proposition}
\label{pro:XTX}
\[ (X^\T X)^{-1}_{j,j} \leq \left(1 + \frac{w^2}{\svmin(X^\T X)}\right) (w^2 I_{p\times p} + X^\T X)^{-1}_{j,j} \]
\end{proposition}
\begin{proof}
We have that
\begin{eqnarray*}
& (X^\T X)^{-1} &= (X^\T X)^{-1}(X^\T X + w^2 I_{p\times p})(X^\T X + w^2 I_{p\times p})^{-1} \cr
&& = (X^\T X + w^2 I_{p\times p})^{-1} + w^2(X^\T X)^{-1}(X^\T X + w^2 I_{p\times p})^{-1}\cr
&& = (I_{p\times p} + w^2 (X^\T X)^{-1})(X^\T X + w^2 I_{p\times p})^{-1}\cr
&& = (X^\T X + w^2 I_{p\times p})^{-1/2} (I_{p\times p} + w^2 (X^\T X)^{-1}) (X^\T X + w^2 I_{p\times p})^{-1/2}
\end{eqnarray*}
where the latter holds because $(I_{p\times p} + w^2 (X^\T X)^{-1})$ and $(X^\T X + w^2 I_{p\times p})^{-1}$ are diagonalizable by the same matrix $V$ (the same matrix for which $(X^\T X) = V S^{-1} V^\T$).
Since we have $\|I_{p\times p} + w^2 (X^\T X)^{-1} \| = 1 + \tfrac{w^2}{\svmin^{2}(X)}$, it is clear that $(I_{p\times p} + w^2 (X^\T X)^{-1}) \preceq (1 + \tfrac{w^2}{\svmin^{2}(X)}) I_{p\times p}$. 
We deduce that $(X^\T X)^{-1}_{j,j} = \vec e_j^\T (X^\T X)^{-1} \vec e_j \leq (1 + \tfrac{w^2}{\svmin^{2}(X)}) (X^\T X+w^2I_{p\times p})^{-1}_{j,j}$.
\end{proof}

Based on Proposition~\ref{pro:XTX} we get from Equation~\eqref{eq:dist_beta'j_betaj_revised} that
\begin{equation}
 |\beta'_j - \beta_j| = O(\left( c_\alpha \sqrt{\frac{\|\vec\zeta\|^2(1+\tfrac {w^2}{\svmin(X^\T X)})}{n-p}} +  c'_\alpha \sqrt{\frac {w^2 + w^2\|\hvec\beta\|^2 + \|\vec\zeta\|^2}{r-p} } \right)\sqrt{(w^2 I_{p\times p} + X^\T X)^{-1}_{j,j}}) \label{eq:dist_beta'j_betaj_revised2}
\end{equation}
And so, if it happens to be the case that exists some small $\eta>0$ for which $\hvec\beta, \vec\zeta$ and $w^2$ satisfy
\begin{equation}
\frac{\|\vec\zeta\|^2(1+\tfrac {w^2}{\svmin(X^\T X)})}{n-p} \leq \eta^2\left( \frac {w^2 + w^2\|\hvec\beta\|^2 + \|\vec\zeta\|^2}{r-p} \right)
\label{eq:constraints_on_zeta_betahat}
\end{equation} then we have that 
$\Pr[\beta_j \in \left(\beta_j' \pm O((1+\eta)\cdot c'_\alpha \|\vec \zeta'\|\sqrt{\tfrac r {r-p}\cdot (M'^\T M')^{-1}_{j,j}})\right)] \geq 1-\alpha$.\footnote{We assume $n\geq r$ so $c_\alpha<c'_\alpha$ as the $T_{n-p}$-distribution is closer to a normal Gaussian than the $T_{r-p}$-distribution.} Moreover, if in this case $|\beta_j| > c'_\alpha(1+\eta)\sqrt{\frac {w^2 + w^2\|\hvec\beta\|^2 + \|\vec\zeta\|^2}{r-p} }\sqrt{(w^2 I_{p\times p} + X^\T X)^{-1}_{j,j}}$ then $\Pr[sign(\beta_j') = sign(\beta_j)]\geq 1-\alpha$. This is precisely what Claims~\ref{clm:sufficient_condition_for_interval_on_beta'} and~\ref{clm:sufficient_conditions_for_finding_right_sign_beta'} below do.

\begin{claim}
\label{clm:sufficient_condition_for_interval_on_beta'}
If there exists $\eta>0$ s.t. $n-p \geq \tfrac{2}{\eta^2}(r-p)$ and $n^2 = \Omega\left(r^{3/2}\cdot\frac{B^2\ln(1/\delta)}\epsilon \cdot  \frac 1 {\eta^2\svmin(\tfrac 1 n X^\T X)}\right)$, then $\Pr[\beta_j \in \left(\beta_j' \pm O((1+\eta)\cdot c'_\alpha \|\vec \zeta'\|\sqrt{\tfrac r {r-p}\cdot (M'^\T M')^{-1}_{j,j}})\right)] \geq 1-\alpha$.
\end{claim}
\begin{proof}
Based on the above discussion, it is enough to argue that under the conditions of the claim, the constraint of Equation~\eqref{eq:constraints_on_zeta_betahat} holds. Since we require $\tfrac{\eta^2}2 \geq \tfrac{r-p}{n-p}$ then it is evident that $\tfrac{\|\vec\zeta\|^2}{n-p} \leq \tfrac{\eta^2\|\vec\zeta\|^2}{2(r-p)}$. So we now show  that $\tfrac{\|\vec\zeta\|^2}{n-p} \cdot \tfrac {w^2}{\svmin(X^\T X)}\leq \tfrac{\eta^2\|\vec\zeta\|^2}{2(r-p)}$ under the conditions of the claim, and this will show the required.
All that is left is some algebraic manipulations. It suffices to have:
\begin{align*}
\tfrac{\eta^2}{2} \cdot \tfrac{n-p}{r-p} \svmin(X^\T X) & \geq \tfrac{\eta^2}{2} \cdot \tfrac{n^2}{r} \svmin(\tfrac 1 n X^\T X) \geq \frac {32B^2\sqrt{r}\ln(8/\delta)}\epsilon \geq w^2 
\end{align*} which holds for $n^2 \geq r^{3/2}\cdot\frac {64B^2\ln(1/\delta)} {\epsilon \eta^2}\svmin(\tfrac 1 n X^\T X)^{-1}$, as we assume to hold.
\end{proof}
\cut{
Therefore, all that is left to show is that w.p. $\geq 1-\nu$ it holds that $\tfrac{\|\vec\zeta\|^2}{n-p} \cdot \tfrac 1{\svmin(X^\T X)}\leq \tfrac{\eta^2 (1+\|\hvec\beta\|^2)}{r-p}$.

To see that, we now invoke Claim~\ref{clm:sizes_hatbeta_zeta}, which states that w.p. $\geq 1-\nu$ we have $\|\vec\beta-\hvec\beta\|^2 = O(\sigma^2\|X^\mpi\|_F^2\ln(\tfrac p \nu))$ and $\|\vec\zeta\|^2= \Theta( \sigma^2 (n-p) )$. Therefore, for $n = p + \Omega(\ln(1/\nu))$ we need to have $\tfrac{2\sigma^2}{\svmin(X^\T X)} \leq \tfrac {\eta^2 (1+\|\vec\beta\|^2)}{2(r-p)}$. Due to the lower bound on $\eta^2$ this indeed holds.

}

\begin{claim}
\label{clm:sufficient_conditions_for_finding_right_sign_beta'}
Fix $\nu \in (0,\tfrac1 2)$. If (i) $n= p + \Omega(\ln(1/\nu))$, (ii) $\|\vec\beta\|^2  = \Omega(\sigma^2 \|X^\mpi\|^2_F \ln(\tfrac p \nu))$ and (iii) $r-p = \Omega\left( \tfrac{(c'_\alpha)^2(1+\eta)^2}{\beta_j^2} \left(1+\|\vec\beta\|^2 +  \frac {\sigma^2}{\svmin(\tfrac 1 n X^\T X)} \right)\right)$, then in the homoscedastic model, with probability $\geq 1-\nu-\alpha$ we have that $sign(\beta_j) = sign(\beta_j')$.
\end{claim}
\begin{proof}
Based on the above discussion, we aim to show that in the homoscedastic model (where each coordinate $e_i\sim\gaussian(0,\sigma^2)$ independently) w.p. $\geq 1- \nu$ it holds that \[|\beta_j| > c'_\alpha(1+\eta)\sqrt{\frac {w^2 + w^2\|\hvec\beta\|^2 + \|\vec\zeta\|^2}{r-p} }\sqrt{(w^2 I_{p\times p} + X^\T X)^{-1}_{j,j}}\] To show this, we invoke Claim~\ref{clm:sizes_hatbeta_zeta} to argue that w.p. $\geq 1-\nu$ we have (i) $\|\vec\zeta\|^2 \leq 2\sigma^2(n-p)$ (since $n = p + \Omega(\ln(1/\nu))$), and (ii) $\|\hvec\beta\|^2 \leq 2\|\vec\beta\|^2$ (since $\|\vec\beta-\hvec\beta\|^2 \leq \sigma^2 \|X^\mpi\|^2_F \ln(\tfrac p \nu)$ whereas $\|\vec\beta\|^2  = \Omega(\sigma^2 \|X^\mpi\|^2_F \ln(\tfrac p \nu))$). We also use the fact that $(w^2 I_{p\times p}+ X^\T X)^{-1}_{j,j} \leq (w^2 + \svmin^{-1}(X^\T X))$, and then deduce that 
\begin{align*}
&(1+\eta)c'_\alpha\sqrt{\frac {w^2 + w^2\|\hvec\beta\|^2 + \|\vec\zeta\|^2}{r-p} }\sqrt{(w^2 I_{p\times p} + X^\T X)^{-1}_{j,j}} 
\cr&~~~~\leq \frac{(1+\eta)c'_\alpha}{\sqrt{r-p}} \sqrt{2\frac{w^2(1+\|\vec\beta\|^2) + \sigma^2(n-p) }{w^2 + \svmin(X^\T X)}} \leq \frac{(1+\eta)c'_\alpha}{\sqrt{r-p}} \sqrt{2(1+\|\vec\beta\|^2)+ \frac{2\sigma^2(n-p)}{\svmin(X^\T X)} } \leq |\beta_j|   
\end{align*}
due to our requirement on $r-p$.
\end{proof}

Observe, out of the $3$ conditions specified in Claim~\ref{clm:sufficient_conditions_for_finding_right_sign_beta'}, condition (i) merely guarantees that the sample is large enough to argue that estimations are close to their expect value; and condition (ii) is there merely to guarantee that $\|\hvec\beta\|\approx \|\vec\beta\|$. It is condition (iii) which is non-trivial to hold, especially together with the conditions of Claim~\ref{clm:sufficient_condition_for_interval_on_beta'} that pose other constraints in regards to $r$, $n$, $\eta$ and the various other parameters in play. It is interesting to compare the requirements on $r$ to the lower bound we get in Theorem~\ref{thm:JL_simple_case_rejecting_null_hypothesis}~--- especially the latter bound. The two bounds are strikingly similar, with the exception that here we also require $r-p$ to be greater than $\frac{1+\|\vec\beta\|^2}{\beta_j^2}$. This is part of the unfortunate effect of altering the matrix $A$: we cannot give confidence bounds only for the coordinates $j$ for which $\beta_j^2$ is very small \emph{relative to $\|\vec\beta\|^2$}.

In summary, we require to have $n = p  + \Omega(\ln(1/\nu))$ and that $X$ contains enough sample points to have $\|\hvec\beta\|$ comparable to $\|\vec\beta\|$, and then set $r$ and $\eta$ such that (it is convenient to think of $\eta$ as a small constant, say, $\eta = 0.1$) 
\begin{itemize}
\item $r - p = O(\eta^2(n-p))$ (which implies $r = O(n)$)
\item $r = O(\left( \eta^2 \frac{\epsilon n^2}{B^2 \ln(1/\delta)}\svmin(\tfrac 1 n X^\T X)\right)^{\tfrac 2 3})$ 
\item $r - p = \Omega( \frac{1+\|\vec\beta\|^2}{\beta_j^2} + \frac{\sigma^2}{\beta_j^2} \cdot {\svmin^{-1}(\tfrac 1 n X^\T X)})$
\end{itemize}
to have that the $(1-\alpha)$-confidence interval around $\beta_j'$ does not intersect the origin. Once again, we comment that these conditions are sufficient but not necessary, and furthermore~--- even with these conditions holding~--- we do not make any claims of optimality of our confidence bound. That is because from Proposition~\ref{pro:XTX} onwards our discussion uses upper bounds that do not have corresponding lower bounds, to the best of our knowledge.

}

\newcommand{\oo}[1]{\widetilde{#1}}
\newcommand{\oovec}[1]{\oo{\pmb{#1}}}
\newcommand{\invxx}{(X^\T X)^{-1}}
\newcommand{\ooinvxx}{\oo{X^\T X}^{\,-1}}
\newcommand{\ooinvxxsq}{\oo{X^\T X}^{\,-2}}

\DeclareRobustCommand{\AGsection}{
\section{Confidence Intervals for ``Analyze Gauss'' Algorithm}
\label{sec:analyze_gauss}

To complete the picture, we now analyze the ``Analyze Gauss'' algorithm of Dwork et al~\mycite{DworkTTZ14}. Algorithm~\ref{alg:AG} works by adding random Gaussian noise to $A^\T A$, where the noise is symmetric with each coordinate above the diagonal sampled i.i.d from $\gaussian(0,\Delta^2)$ with $\Delta^2 = O\left(B^4 \tfrac{\log(1/\delta)}{\epsilon^2}\right)$.\footnote{It is easy to see that the $l_2$-global sensitivity of the mapping $A \mapsto A^\T A$ is $\propto B^4$. Fix any $A_1, A_2$ that differ on one row which is some vector $\vec v$ with $\|\vec{v}\|=B$ in $A_1$ and the all zero vector in $A_2$. Then $GS_2^2 = \|A_1^\T A_1 - A_2^\T A_2\|_F^2 = \|\vec v\vec v^T\|_F^2 = \trace(\vec v \vec v^\T \cdot \vec v \vec v^\T) = (\vec v^\T \vec v)^2 = B^4$.} Using the same notation for a sub-matrix of $A$ as $[X;\vec y]$ as before, with $X \in \R^{n\times p}$ and $y\in \R^n$, we denote the output of Algorithm~\ref{alg:AG} as
\begin{equation} 
\resizebox*{0.475\textwidth}{!}{$\left(\begin{array}{ccc|c} & & & \cr & \oo{X^\T X} & &\oo{X^\T\vec y}\\ & & & \\ \hline \vphantom{^{\large{2}}} & \oo{\vec y^\T X} & &\oo{\vec y^\T \vec y}\end{array}\right) = \left(\begin{array}{ccc|c} & & & \cr & {X^\T X}+N & &{X^\T\vec y}+\vec n\\ & & & \\ \hline \vphantom{^{\displaystyle{2}}}& {\vec y^\T X}+\vec n^\T & &{\vec y^\T \vec y}+m\end{array}\right)$}\label{eq:AG_notation}\end{equation} where $N$ is a symmetric $p\times p$-matrix, $\vec n$ is a $p$-dimensional vector and $m$ is a scalar, whose coordinates are sampled i.i.d from $\gaussian(0,\Delta^2)$. 

Using the output of Algorithm~\ref{alg:AG}, it is simple to derive analogues of $\hvec\beta$ and $\|\vec{\zeta}\|^2$ (Equations~\eqref{eq:hat_beta} and~\eqref{eq:zeta})
\begin{align}
&\oovec{\beta} = \left(\oo{X^\T X}\right)^{-1} \oo{X^\T \vec y} = \left( X^\T X + N\right)^{-1} (X^T \vec y + \vec n) \label{eq:oo_beta}\\
&\oo{\| \vec{\zeta}\|^2} = \oo{\vec y^\T \vec y} -2\, \oo{\vec y^T X}\,\oovec{\beta} + \oovec{\beta}^\T \,\oo{X^\T X} \,\oovec{\beta}  
\ifx\twocols\undefined =\else \cr &~~~~~~~~~~~=  \fi 
\oo{\vec y^T \vec y} -  \oo{\vec y^T X}\, \oo{X^\T X}^{-1}\, \oo{X^T \vec y} \label{eq:oo_zeta}
\end{align}
We now argue that it is possible to use $\oo{\beta_j}$ and $\oo{\| \vec{\zeta}\|^2}$ to get a confidence interval for $\beta_j$ under certain conditions.

\begin{theorem}
	\label{thm:AG_interval}
	Fix $\alpha,\nu\in (0,\tfrac 1 2)$. Assume that there exists $\eta\in (0,\tfrac 1 2)$ s.t. $\svmin(X^\T X) > \Delta\sqrt{p\ln(1/\nu)}/\eta$. Under the homoscedastic model, given $\vec{\beta}$ and $\sigma^2$, if we assume also that $\|\vec \beta\|\leq B$ and $\|\hvec{\beta}\| = \|\invxx X^\T \vec y\|\leq B$, then w.p. $\geq 1-\alpha-\nu$ it holds that
	\ifx\twocols\undefined  
	\begin{eqnarray*} \left|\beta_j - \oo{\beta}_j \right| &=&  O\Big(  \rho \cdot \sqrt{\left( \ooinvxx_{j,j} + \Delta \sqrt{p\ln(1/\nu)} \cdot\ooinvxxsq_{j,j} \right) \ln(1/\alpha)} \cr
		&&~~~~+ \Delta \sqrt{\ooinvxxsq_{j,j}\cdot \ln(1/\nu)} \cdot (B\sqrt{p}+1)   \Big)
	\end{eqnarray*}
	\else 
	$|\beta_j - \oo{\beta}_j|$ it at most
	\begin{eqnarray*} O\Big(  \rho \cdot \sqrt{\left( \ooinvxx_{j,j} + \Delta \sqrt{p\ln(1/\nu)} \cdot\ooinvxxsq_{j,j} \right) \ln(1/\alpha)} \cr
		~~~~+ \Delta \sqrt{\ooinvxxsq_{j,j}\cdot \ln(1/\nu)} \cdot (B\sqrt{p}+1)   \Big)
	\end{eqnarray*}
	\fi 
	where $\rho$ is such that $\rho^2$ is w.h.p an upper bound on $\sigma^2$, defined as
	\ifx\twocols\undefined 
	\[ \rho^2 \stackrel{\rm def}{=}  \left(\tfrac{1}{\sqrt{n-p} - 2\sqrt{\ln(4/\alpha)}}\right)^2 \cdot \left( \oo{\| \vec{\zeta}\|^2} - C\cdot\left(  \Delta \tfrac{ B^2\sqrt{p}}{1-\eta}\sqrt{\ln(1/\nu)}   + \Delta^2 \|\ooinvxx\|_F\cdot  \ln(p/\nu) \right)   \right) \]
	\else
	\begin{align*}
	&\rho^2 \stackrel{\rm def}{=}  \left(\tfrac{1}{\sqrt{n-p} - 2\sqrt{\ln(4/\alpha)}}\right)^2 \cdot
	\cr &~~\resizebox{0.5\textwidth}{!}{$ \left( \oo{\| \vec{\zeta}\|^2} - C\cdot\left(  \Delta \tfrac{ B^2\sqrt{p}}{1-\eta}\sqrt{\ln(1/\nu)}   + \Delta^2 \|\ooinvxx\|_F\cdot  \ln(p/\nu) \right)   \right)$}
	\end{align*}
	\fi for some large constant $C$.
\end{theorem}

We comment that in practice, instead of using $\rho$, it might be better to use the MLE of $\sigma^2$, namely: \[\overline{\sigma^2}~ \stackrel{\rm def}{=}~\tfrac {1}{n-p}\left( \oo{\|\vec{\zeta}\|^2} + \Delta^2 \|\ooinvxx\|_F\right)\]
instead of $\rho^2$, the upper bound we derived for $\sigma^2$. (Replacing an unknown variable with its MLE estimator is a common approach in applied statistics.) 
Note that the assumption that $\|\vec{\beta}\|\leq B$ is fairly benign once we assume each row has bounded $l_2$-norm. The assumption $\|\hvec{\beta}\| \leq B$ simply assumes that $\hvec\beta$ is a reasonable estimation of $\vec\beta$, which is likely to hold if we assume that $X^\T X$ is well-spread. The assumption about the magnitude of the least singular value of $X^\T X$ is therefore the major one. Nonetheless, in the case we considered before where each row in $X$ is sampled i.i.d from $\gaussian(\vec 0, \Sigma)$, this assumption merely means that $n$ is large enough s.t. $n = \tilde{\Omega}(\tfrac{\Delta\sqrt{p\ln(1/\nu)}}{\eta\cdot\svmin(\Sigma)})$.

In order to prove Theorem~\ref{thm:AG_interval}, we require the following proposition. 
\begin{proposition}
\label{pro:norms_of_gaussians} Fix any $\nu \in (0,\tfrac 1 2)$. Fix any matrix $M\in\R^{p\times p}$. Let $\vec v \in\R^p$ be a vector with each coordinate sampled independently from a Gaussian $\gaussian(0, \Delta^2)$. Then we have that $\Pr \left[ \|M\vec  v\| > \Delta \cdot \|M\|_F\sqrt{2\ln(2p/\nu)} \right] < \nu$.
\end{proposition}
\begin{proof}
Given $M$, we have that $M\vec v \sim \gaussian(\vec 0, \Delta^2 \cdot M M^\T)$. Denoting $M$'s singular values as $sv_1, \ldots, sv_p$, we can rotate $M\vec v$ without affecting its $l_2$-norm and infer that $\|M \vec v|^2$ is distributed like a sum on $p$ independent Gaussians, each sampled from $\gaussian(0,\Delta^2 \cdot sv_i^2)$. Standard union bound gives that w.p. $\geq 1-\nu$ non of the $p$ Gaussians exceeds its standard deviation by a factor of $\sqrt{2\ln(2p/\nu)}$. Hence, w.p. $\geq 1-\nu$ it holds that $\|M\vec v\|^2 \leq 2\Delta^2\sum_i sv_i^2 \ln(2p/\nu) = 2\Delta^2\cdot \trace(M M^\T)\cdot \ln(2p/\nu)$. 
\end{proof}
Our proof also requires the use of the following equality, that holds for any invertible $A$ and any matrix $B$ s.t. $I+B\cdot A^{-1}$ is invertible:
\begin{equation*} \left(A+B\right)^{-1} = A^{-1} - A^{-1} \left( I + BA^{-1}\right)^{-1} B A^{-1}\end{equation*}
In our case, we have
\ifx\twocols\undefined 
\begin{eqnarray}
& \ooinvxx &= (X^\T X+N)^{-1} = \invxx - \invxx \left( I + N\invxx\right)^{-1} N \invxx \cr &&= \invxx \left(I -  \left( I + N\invxx\right)^{-1} N \invxx \right) \stackrel{\rm def}{=} \invxx \left( I - Z\cdot \invxx\right)\cr &\label{eq:inv_sum}
\end{eqnarray}
\else 
\begin{align}
& \ooinvxx 
\cr &= (X^\T X+N)^{-1} 
\cr &= \invxx - \invxx \resizebox*{0.175\textwidth}{!}{$\left( I + N\invxx\right)^{-1} N$} \invxx 
\cr &= \invxx \left(I -  \left( I + N\invxx\right)^{-1} N \invxx \right) 
\cr &\stackrel{\rm def}{=} \invxx \left( I - Z\cdot \invxx\right) \label{eq:inv_sum}
\end{align}
\fi 

\begin{proof}[Proof of Theorem~\ref{thm:AG_interval}]
	Fix $\nu>0$. First, we apply to standard results about Gaussian matrices, such as~\cite{Tao12} (used also by~\cite{DworkTTZ14} in their analysis), to see that w.p. $\geq 1- \nu/6$ we have $\|N\| = O(\Delta \sqrt{p \ln(1/\nu)})$. And so, for the remainder of the proof we fix $N$ subject to having bounded operator norm. Note that by fixing $N$ we fix $\oo{X^\T X}$. 
	
	Recall that in the homoscedastic model, $\vec y = X\vec \beta + \vec e$ with each coordinate of $\vec e$ sampled i.i.d from $\gaussian(0,\sigma^2)$. We therefore have that
	\begin{align}
	\oovec{\beta} & = \ooinvxx (X^\T \vec y + \vec n) = \ooinvxx (X^\T X \vec\beta + X^\T \vec e + \vec n)\cr
	& = \ooinvxx (\oo{X^\T X} - N)\vec{\beta} + \ooinvxx X^\T \vec e + \ooinvxx \vec n \cr
	&= \vec \beta - \ooinvxx N \vec{\beta} + \ooinvxx X^\T \vec e + \ooinvxx \vec n\cr
	\intertext{Denoting the $j$-th row of $\ooinvxx$ as $\ooinvxx_{j\rightarrow}$ we deduce:}
	\oo\beta_j &=  \beta_j - \ooinvxx_{j\rightarrow} N \vec{\beta} + \ooinvxx_{j\rightarrow} X^\T \vec e + \ooinvxx_{j\rightarrow} \vec n \label{eq:oobeta_j}
	\end{align}
	
	We na\"ively bound the size of the term $\ooinvxx_{j\rightarrow} N \vec{\beta}$ by $\left\| \ooinvxx_{j\rightarrow}\right\| \|N\| \|\vec{\beta}\| = O\left(\left\| \ooinvxx_{j\rightarrow}\right\| \cdot B \Delta \sqrt{p\ln(1/\nu)}   \right)$.
	
	To bound $\ooinvxx_{j\rightarrow} X^\T \vec e$ note that $\vec e$ is chosen independently of $\oo{X^\T X}$ and since $\vec e \sim \gaussian (\vec 0, \sigma^2 I)$ we have $\ooinvxx_{j\rightarrow} X^\T \vec e \sim \gaussian\left(\vec 0, \sigma^2 \cdot \vec e_j^\T \ooinvxx\cdot X^\T X\cdot \ooinvxx \vec e_j\right)$. Since we have
	\ifx\twocols\undefined 
	\[ \ooinvxx\cdot X^\T X\cdot \ooinvxx = \ooinvxx\cdot (\oo{X^\T X} - N)\cdot \ooinvxx = \ooinvxx - \ooinvxx \cdot N \cdot \ooinvxx\] 
	\else
	\begin{align*}
	&\ooinvxx\cdot X^\T X\cdot \ooinvxx 
	\cr &~= \ooinvxx\cdot (\oo{X^\T X} - N)\cdot \ooinvxx 
	\cr &~= \ooinvxx - \ooinvxx \cdot N \cdot \ooinvxx
	\end{align*} 
	\fi  we can bound the variance of $\ooinvxx_{j\rightarrow} X^\T \vec e$ by $\sigma^2 \left( \ooinvxx_{j,j} + \|N\| \cdot \left\| \ooinvxx_{j\rightarrow}\right\|^2 \right)$. Appealing to Gaussian concentration bounds, we have that w.p. $\geq 1-\alpha/2$ the absolute value of this Gaussian is at most 
	\ifx\twocols\undefined
	$O \left(\sqrt{\left( \ooinvxx_{j,j} + \Delta \sqrt{p\ln(1/\nu)} \cdot \left\| \ooinvxx_{j\rightarrow}\right\|^2 \right) \sigma^2\ln(1/\alpha)}~\right)$.
	\else
	\resizebox{0.5\textwidth}{!}{$O \left(\sqrt{\left( \ooinvxx_{j,j} + \Delta \sqrt{p\ln(1/\nu)} \cdot \left\| \ooinvxx_{j\rightarrow}\right\|^2 \right) \sigma^2\ln(1/\alpha)}~\right)$.}
	\fi 
	
	To bound  $\ooinvxx_{j\rightarrow} \vec n$ note that $\vec n \sim \gaussian(\vec 0, \Delta^2 I)$ is sampled independently of $\oo{X^\T X}$. We therefore have that $\ooinvxx_{j\rightarrow} \vec n\sim \gaussian (0, \Delta^2 \left\| \ooinvxx_{j\rightarrow}\right\|^2)$. Gaussian concentration bounds give that w.p $\geq 1-\nu/6$ we have $|\ooinvxx_{j\rightarrow} \vec n| = O\left(\Delta \left\| \ooinvxx_{j\rightarrow}\right\|\sqrt{\ln(1/\nu)} \right) $.
	
	Plugging this into our above bounds on all terms that appear in Equation~\eqref{eq:oobeta_j} we have that w.p. $\geq 1-\nu/2-\alpha/2$ we have that $\left| \oo\beta_j - \beta_j\right|$ is at most
	\ifx\twocols\undefined 
	\begin{align*}
	& O\left(\left\| \ooinvxx_{j\rightarrow}\right\| \cdot B \Delta \sqrt{p\ln(1/\nu)}   \right) \cr
	&~~+O \left(\sigma \sqrt{\left( \ooinvxx_{j,j} + \Delta \sqrt{p\ln(1/\nu)} \cdot \left\| \ooinvxx_{j\rightarrow}\right\|^2 \right)\ln(1/\alpha)}~\right) \cr
	&~~+  O\left(\Delta \left\| \ooinvxx_{j\rightarrow}\right\|\sqrt{\ln(1/\nu)} \right)
	\end{align*}
	\else
	\begin{align*}
	& O\left(\left\| \ooinvxx_{j\rightarrow}\right\| \cdot B \Delta \sqrt{p\ln(1/\nu)}   \right) \cr
	&\resizebox{0.5\textwidth}{!}{$~~+O \left(\sigma \sqrt{\left( \ooinvxx_{j,j} + \Delta \sqrt{p\ln(1/\nu)} \cdot \left\| \ooinvxx_{j\rightarrow}\right\|^2 \right)\ln(1/\alpha)}~\right)$} \cr
	&~~+  O\left(\Delta \left\| \ooinvxx_{j\rightarrow}\right\|\sqrt{\ln(1/\nu)} \right)
	\end{align*}
	\fi 
	Note that due to the symmetry of $\oo{X^\T X}$ we have $\left\| \ooinvxx_{j\rightarrow}\right\|^2 = \ooinvxxsq_{j,j}$ (the $(j,j)$-coordinate of the matrix $\ooinvxxsq$), thus
	\ifx\twocols\undefined 
	\begin{align}
	\left| \oo\beta_j - \beta_j\right| &= O\Big( \sigma \cdot \sqrt{\left( \ooinvxx_{j,j} + \Delta \sqrt{p\ln(1/\nu)} \cdot\ooinvxxsq_{j,j} \right) \ln(1/\alpha)} \cr 
	& ~~~~+ \Delta \sqrt{\ooinvxxsq_{j,j}\cdot \ln(1/\nu)} \cdot (B\sqrt{p}+1) \ \Big) \label{eq:dist_oobeta_beta}
	\end{align}
	\else 
	$|\oo\beta_j - \beta_j|$ is at most 
	\begin{align}
	&O\Big( \sigma \cdot \sqrt{\left( \ooinvxx_{j,j} + \Delta \sqrt{p\ln(1/\nu)} \cdot\ooinvxxsq_{j,j} \right) \ln(1/\alpha)} \cr 
	& ~~~~+ \Delta \sqrt{\ooinvxxsq_{j,j}\cdot \ln(1/\nu)} \cdot (B\sqrt{p}+1) \ \Big) \label{eq:dist_oobeta_beta}
	\end{align}
	\fi 
	
	All of the terms appearing in Equation~\eqref{eq:dist_oobeta_beta} are known given $\oo{X^\T X}$, except for $\sigma$ --- which is a parameter of the model. Next, we derive an upper bound on $\sigma$ which we can then plug into Equation~\eqref{eq:dist_oobeta_beta} to complete the proof of the theorem and derive a confidence interval for $\beta_j$.
	
	\ifx\twocols\undefined 
	Recall Equation~\eqref{eq:oo_zeta}, according to which we have
	\begin{align}
	\oo{\|\vec{\zeta}\|^2	} &= \oo{\vec y^T \vec y} -  \oo{\vec y^T X}\, \oo{X^\T X}^{-1}\, \oo{X^T \vec y}  
	\cr & \stackrel{\rm Eq\eqref{eq:inv_sum}}= \vec y^\T \vec y + m - (\vec y^\T X + \vec n^\T)\invxx (I - Z \cdot \invxx) (X^\T \vec y + \vec n) 
	\cr 	& =  \vec y^\T \vec y + m  - \vec y^\T X \invxx X^\T \vec y + \vec y^\T X \invxx Z \invxx X^\T  \vec y \cr
	& \qquad  - 2 \vec y^\T X\invxx \vec n + 2 \vec y^\T X  \invxx Z \invxx \vec n - \vec n^\T \invxx (I - Z \cdot \invxx)\vec n	\notag
	\intertext{Recall that $\hvec{\beta} = \invxx X^\T \vec y$, and so we have}
	&= \vec y^\T \left( I - X\invxx X^\T\right) \vec y + m - \hvec{\beta}^\T Z \hvec{\beta} - 2\hvec \beta^\T (I - Z\invxx )\vec n - \vec n^\T \ooinvxx \vec n 
	\label{eq:oozeta_breakdown}
	\end{align}
	and of course, both $\vec n$ and $m$ are chosen independently of $\oo{X^\T X}$ and $\vec y$.
	\else
	Recall Equation~\eqref{eq:oo_zeta}, according to which we have
	\begin{align}
	\oo{\|\vec{\zeta}\|^2	} &= \oo{\vec y^T \vec y} -  \oo{\vec y^T X}\, \oo{X^\T X}^{-1}\, \oo{X^T \vec y}  
	\cr & \stackrel{\rm \eqref{eq:inv_sum}}= \vec y^\T \vec y + m 
	\cr &~~~~~ - (\vec y^\T X + \vec n^\T)\resizebox*{0.3\textwidth}{!}{$\invxx (I - Z \cdot \invxx) (X^\T \vec y + \vec n)$} 
	\cr 	& =  \vec y^\T \vec y + m  
	\cr &~~~~~ - \vec y^\T X \invxx X^\T \vec y 
	\cr &~~~~~ + \vec y^\T X \invxx Z \invxx X^\T  \vec y 
	\cr &~~~~~  - 2 \vec y^\T X\invxx \vec n 
	\cr &~~~~~ + 2 \vec y^\T X  \invxx Z \invxx \vec n 
	\cr &~~~~~ - \vec n^\T \invxx (I - Z \cdot \invxx)\vec n	\notag
	\intertext{Recall that $\hvec{\beta} = \invxx X^\T \vec y$, and so we have}
	&= \vec y^\T \left( I - X\invxx X^\T\right) \vec y + m - \hvec{\beta}^\T Z \hvec{\beta} 
	\cr & ~~~~ - 2\hvec \beta^\T (I - Z\invxx )\vec n - \vec n^\T \ooinvxx \vec n 
	\label{eq:oozeta_breakdown}
	\end{align}
	and of course, both $\vec n$ and $m$ are chosen independently of $\oo{X^\T X}$ and $\vec y$.
	\fi 
	
	Before we bound each term in Equation~\eqref{eq:oozeta_breakdown}, we first give a bound on $\|Z\|$. Recall, $Z  = \left( I + N \invxx\right)^{-1} N$. Recall our assumption (given in the statement of Theorem~\ref{thm:AG_interval}) that $\svmin(X^\T X) \geq \tfrac{\Delta}{\eta}\sqrt{p\ln(1/\nu)}$. This implies that $\|N \invxx\| \leq \|N\|\cdot \svmin(X^\T X)^{-1} = O(\eta)$. Hence
	\[ \|Z\| \leq (\|I + N\invxx\|)^{-1} \cdot \|N\| = O\left( \tfrac{\Delta \sqrt{p\ln(1/\nu)}} {1-\eta} \right)\]
	Moreover, this implies that $\|Z \invxx\| \leq O\left(\tfrac{\eta}{1-\eta} \right)$ and that $\| I - Z\invxx\| \leq O\left( \tfrac 1 {1-\eta}\right)$.
	
	Armed with these bounds on the operator norms of $Z$ and $(I-Z\invxx)$ we bound the magnitude of the different terms in Equation~\eqref{eq:oozeta_breakdown}.
	\begin{itemize}
		\item The term $\vec y^T \left(I - X X^\mpi \right) \vec y$ is the exact term from the standard OLS, and we know it is distributed like $\sigma^2 \cdot \chi^2_{n-p}$ distribution. Therefore, it is greater than $\sigma^2(\sqrt{n-p}~-~2\sqrt{\ln(4/\alpha)})^2$ w.p. $\geq 1- \alpha/2$.
		\item The scalar $m$ sampled from $m \sim \gaussian(0,\Delta^2)$ is bounded by $O(\Delta\sqrt{\ln(1/\nu)})$ w.p. $\geq 1- \nu/8$.
		\item Since we assume $\|\hvec{\beta}\|\leq B$, the term $\hvec{\beta}^\T Z \hvec \beta$ is upper bounded by $B^2 \|Z\| = O\left( \tfrac{B^2\Delta \sqrt{p\ln(1/\nu)}} {1-\eta} \right)$.
		\item Denote $\vec z^\T \vec n = 2\hvec\beta^\T (I - Z\invxx) \vec n$. We thus have that $\vec z^\T \vec n \sim \gaussian(0,\Delta^2 \|\vec z\|^2)$ and that its magnitude is at most $O(\Delta\cdot\|\vec z\| \sqrt{\ln(1/\nu)})$ w.p. $\geq 1- \nu/8$. We can upper bound $\|\vec z\|~\leq~2\|\hvec{\beta}\|~\|I~-~Z\invxx\| = O(\tfrac {B}{1-\eta})$, and so this term's magnitude is upper bounded by $ O\left(\tfrac {\Delta\cdot B\sqrt{\ln(1/\nu)}}{1-\eta}\right)$.
		\item Given our assumption about the least singular value of $X^\T X$ and with the bound on $\|N\|$, we have that $\svmin(\oo{X^\T X}) \geq \svmin(X^\T X) - \|N\| >0$ and so the symmetric matrix $\oo{X^\T X}$ is a PSD. Therefore, the term $\vec n ^\T \ooinvxx \vec n = \| \oo{X^\T X}^{-1/2} \vec n\|^2$ is strictly positive.  Applying Proposition~\ref{pro:norms_of_gaussians} we have that w.p. $\geq 1- \nu/8$ it holds that $\vec n^\T \ooinvxx \vec n \leq O\left(\Delta^2 \|\ooinvxx\|_F\cdot  \ln(p/\nu)\right)$.
	\end{itemize}
	Plugging all of the above bounds into Equation~\eqref{eq:oozeta_breakdown} we get that w.p. $\geq 1- \nu/2-\alpha/2$ it holds that
	\ifx\twocols\undefined 
	\[ 
	\sigma^2 \leq \left(\tfrac{1}{\sqrt{n-p} - 2\sqrt{\ln(4/\alpha)}}\right)^2 \left( \oo{\| \vec{\zeta}\|^2} + O\left( (1+\tfrac{ B^2\sqrt{p} + B}{1-\eta})\Delta\sqrt{\ln(1/\nu)} + \Delta^2 \|\ooinvxx\|_F\cdot  \ln(p/\nu) \right)   \right)
	\]	
	\else 
	\begin{align*}
		&\sigma^2 \leq \left(\tfrac{1}{\sqrt{n-p} - 2\sqrt{\ln(4/\alpha)}}\right)^2\cdot 
		\cr &\resizebox*{0.5\textwidth}{!}{$\left( \oo{\| \vec{\zeta}\|^2} + O\left( (1+\tfrac{ B^2\sqrt{p} + B}{1-\eta})\Delta\sqrt{\ln(1/\nu)} + \Delta^2 \|\ooinvxx\|_F\cdot  \ln(p/\nu) \right)   \right)$}
	\end{align*}
	\fi
	and indeed, the RHS is the definition of $\rho^2$ in the statement of Theorem~\ref{thm:AG_interval}.
\end{proof}
}

\DeclareRobustCommand{\AGshort}{
\section{Confidence Intervals for ``Analyze Gauss''}
\label{sec:analyze_gauss_short}

In this section we analyze the ``Analyze Gauss'' algorithm of Dwork et al~\mycitet{DworkTTZ14}. Algorithm~\ref{alg:AG} works by adding random Gaussian noise to $A^\T A$, where the noise is symmetric with each coordinate above the diagonal sampled i.i.d from $\gaussian(0,\Delta^2)$ with $\Delta^2 = O\left(B^4 \tfrac{\log(1/\delta)}{\epsilon^2}\right)$. Using the same notation for a sub-matrix of $A$ as $[X;\vec y]$ as before, we denote the output of Algorithm~\ref{alg:AG} as
\ifx\twocols\undefined $\left(\begin{array}{ccc|c} & & & \cr & \oo{X^\T X} & &\oo{X^\T\vec y}\\ & & & \\ \hline \vphantom{^{\large{2}}} & \oo{\vec y^\T X} & &\oo{\vec y^\T \vec y}\end{array}\right)$
\else
\resizebox{!}{8mm}{$\left(\begin{array}{ccc|c} & & & \cr & \oo{X^\T X} & &\oo{X^\T\vec y}\\ & & & \\ \hline \vphantom{^{\large{2}}} & \oo{\vec y^\T X} & &\oo{\vec y^\T \vec y}\end{array}\right)$}.
\fi
Thus, we approximate $\vec{\beta}$ and $\|\vec{\zeta}\|$ by $\oovec{\beta} = \left(\oo{X^\T X}\right)^{-1} \oo{X^\T \vec y}$ and $\oo{\| \vec{\zeta}\|^2} = \oo{\vec y^\T \vec y} -2\, \oo{\vec y^T X}\,\oovec{\beta} + \oovec{\beta}^\T \,\oo{X^\T X} \,\oovec{\beta}$ resp.
We now argue that it is possible to use $\oo{\beta_j}$ and $\oo{\| \vec{\zeta}\|^2}$ to get a confidence interval for $\beta_j$ under certain conditions.

\begin{theorem}
	\label{thm:AG_interval_body}
	Fix $\alpha,\nu\in (0,\tfrac 1 2)$. Assume that there exists $\eta\in (0,\tfrac 1 2)$ s.t. $\svmin(X^\T X) > \Delta\sqrt{p\ln(1/\nu)}/\eta$. Under the homoscedastic model, given $\vec{\beta}$ and $\sigma^2$, if we assume also that $\|\vec \beta\|\leq B$ and $\|\hvec{\beta}\| = \|\invxx X^\T \vec y\|\leq B$, then w.p. $\geq 1-\alpha-\nu$ it holds that 
	\ifx\confversion\undefined
	\begin{eqnarray*} \left|\beta_j - \oo{\beta}_j \right| &=&  O\Big(  \rho \cdot \sqrt{\left( \ooinvxx_{j,j} + \Delta \sqrt{p\ln(1/\nu)} \cdot\ooinvxxsq_{j,j} \right) \ln(1/\alpha)} \cr
		&&~~~~+ \Delta \sqrt{\ooinvxxsq_{j,j}\cdot \ln(1/\nu)} \cdot (B\sqrt{p}+1)   \Big)
	\end{eqnarray*}
	where $\rho$ is such that $\rho^2$ is w.h.p an upper bound on $\sigma^2$, defined (using some large constant $C$) as
	\[ \rho^2 \stackrel{\rm def}{=}  \left(\tfrac{1}{\sqrt{n-p} - 2\sqrt{\ln(4/\alpha)}}\right)^2 \cdot \left( \oo{\| \vec{\zeta}\|^2} - C\cdot\left(  \Delta \tfrac{ B^2\sqrt{p}}{1-\eta}\sqrt{\ln(1/\nu)}   + \Delta^2 \|\ooinvxx\|_F\cdot  \ln(p/\nu) \right)   \right) \] 
	\else
	$\left|\beta_j - \oo{\beta}_j \right|$ is at most
	\vspace{-2mm}
	\begin{eqnarray*}  O\Big(  \rho \cdot \sqrt{\left( \ooinvxx_{j,j} + \Delta \sqrt{p\ln(1/\nu)} \cdot\ooinvxxsq_{j,j} \right) \ln(1/\alpha)} \cr
	~~~~+ \Delta \sqrt{\ooinvxxsq_{j,j}\cdot \ln(1/\nu)} \cdot (B\sqrt{p}+1)   \Big)
	\end{eqnarray*}
	where $\rho$ is w.h.p an upper bound on $\sigma$ (details appear in the Supplementary material).
	\fi
\end{theorem}

Note that the assumptions that $\|\vec{\beta}\|\leq B$ and $\|\hvec{\beta}\| \leq B$ are fairly benign once we assume each row has bounded $l_2$-norm. The key assumption is that $X^\T X$ is well-spread. Yet in the model where each row in $X$ is sampled i.i.d from $\gaussian(\vec 0, \Sigma)$, this assumption merely means that $n$ is large enough~--- namely, that $n = \tilde{\Omega}(\tfrac{\Delta\sqrt{p\ln(1/\nu)}}{\eta\cdot\svmin(\Sigma)})$.\ifx\confversion\undefined The proof of Theorem~\ref{thm:AG_interval_body} appears in Section~\ref{sec:analyze_gauss}.\fi
}

\myvspace{-4mm}
\DeclareRobustCommand{\conclusions}{
\section{Conclusions and Future Directions}
\label{sec:conclusions}
\myvspace{-2mm}

This work is the first, to the best of our knowledge, to provide an analysis of a differentially private technique for OLS' statistical inference. Furthermore, our work is the first to approach Ridge Regression from a completely different perspective. We propose that one should set the penalty term in Ridge Regression to a specific value (namely, $w^2$) so that by projecting the problem using the Johnson-Lindenstrauss transform, we still satisfy differential privacy.

\omitfromversion{We believe this work should be applicable in practice and curious to see its performance over real datasets. (Initial investigation was done in~\mycite{Sheffet15}, however, the experiments there look at the distance $\|\tvec\beta-\vec\beta\|$ rather than $t$-values and $p$-values.) In particular, we are curious to see whether the conditions posed in Section~\ref{sec:ridge_regression} hold in practice, and if indeed one is able to use the JLT version of Ridge Regression without having $\vec\beta'$ far from $\vec\beta$ or $\hvec\beta$. We are curious also to see if one is able to give a better characterization of the distances between of any pair of the following $4$ vectors: $\vec\beta$ (the true coefficients), $\hvec\beta$ (the linear regression estimator from the data), $\vec\beta^R$ (the Ridge Regression estimator) and $\vec\beta'$ (the estimator from the projected Ridge Regression problem). Also, observe that the statistical analysis in our work follows the frequentist approach. However, Ridge Regression is also motivated from a Bayesian perspective (where $\vec\beta$ has a prior of a spherical Gaussian). Deriving a Bayesian analysis of private least squares seems to be both important and challenging. As ever, the question of matching lower bounds is of importance. Does there exist a sample of points from a multivariate Gaussian for which, without privacy we are likely to $\alpha$-reject the null-hypothesis, but no differentially private algorithm is likely to $\alpha$-reject the null-hypothesis? Is it the case that the increase in sample complexity of our algorithm is optimal?}

We believe there is much work to be done in order to bridge the gap between TCS' standard utility analysis of differentially private algorithms and the statistical inference techniques used in practice in data analysis. Statistical inference is often done using deterministic estimators, where the sole source of randomness lies in the underlying model of data generation. In contrast, differentially private estimators are inherently random in their computation. Statistical inference that considers \emph{both} the randomness in the data and the randomness in the computation is highly uncommon, and this work deals solely with one particular analysis. 
As noted before, OLS is just the first out of many variants of linear regression applied in data analysis, for which confidence estimations should be derived. And even beyond linear regression~--- OLS is only one of many MLE techniques which can be associated with confidence estimations, based on the general recipe of estimating the information matrix of the loss function (the expected Hessian of the loss function, whose computation is often fairly complicated even without privacy). Computing confidence estimations for other differentially private estimators pose a difficult and challenging problem. 

Lastly, we also believe that other applications of the Johnson-Lindenstrauss transform, previously studied for computational speedups, can and should also be analyzed also for their privacy preserving implications.
}
\ifdefined\fullversion
\conclusions
\else
\AGshort
\fi

\DeclareRobustCommand{\acksfull}{
The bulk of this work was done when the author was a postdoctoral fellow at Harvard University, supported by NSF grant CNS-123723;
and also an unpaid collaborator on NSF grant 1565387. 
The author wishes to wholeheartedly thank Prof.~Salil Vadhan, for his tremendous help in shaping this paper. The author would also like to thank Prof.~Jelani Nelson and the members of the ``Privacy Tools for Sharing Research Data'' project at Harvard University (especially 
James Honaker, Vito D'Orazio, Vishesh Karwa, Prof.~Kobbi Nissim and Prof.~Gary King) for many helpful discussions and suggestions; as well as Abhradeep Thakurta for clarifying the similarity between our result. Lastly the author thanks the anonymous referees for many helpful suggestions in general and for a reference to~\mycite{Ullman15} in particular.}

\section{Experiment: $t$-Values of Output}
\label{sec:experiment}

\myparagraph{Goal.} We set to experiment with the outputs of Algorithms~\ref{alg:private_JL} and~\ref{alg:AG}. While Theorem~\ref{thm:main_theorem_matrix_unaltered} guarantees that computing the $t$-value from the output of Algorithm~\ref{alg:private_JL} in the \texttt{matrix unaltered} case does give a good approximation of the $t$-value -- we were wondering if by computing the $t$-value directly from the output we can (a) get a good approximation of the true (non-private) $t$-value and (b) get the same ``higher-level conclusion'' of rejecting the null-hypothesis. The answers are, as ever, mixed. The two main observations we do notice is that both algorithms improve as the number of examples increases, and that Algorithm~\ref{alg:private_JL} is more conservative then Algorithm~\ref{alg:AG}.

\myparagraph{Setting.} We tested both algorithms in two settings. The first is over synthetic data. Much like the setting in Theorems~\ref{thm:baseline_rejecting_nh} and~\ref{thm:JL_simple_case_rejecting_null_hypothesis}, $X$ was generated using $p=3$ independent normal Gaussian features, and $\vec y$ was generated using the homoscedastic model. We chose $\vec\beta = (0.5,-0.25,0)$ so the first coordinate is twice as big a the second but of opposite sign, and moreover, $\vec y$ is independent of the $3$rd feature. The variance of the label is also set to $1$, and so the variance of the homosedastic noise equals to $\sigma^2 = 1-(0.5)^2 -(-0.25)^2$. The number of observations $n$ ranges from $n=1000$ to $n=100000$.

The second setting is over real-life data. We ran the two algorithms over diabetes dataset collected over ten years (1999-2008) taken from the UCI repository~\cite{UCIdiabetes}. We truncated the data to 4 attributes: sex (binary), age (in buckets of 10 years), number medications (numeric, 0-100), and a diagnosis (numeric, 0-1000). Naturally, we added a $5^{\rm th}$ column of all-$1$ (intercept). Omitting any entry with missing  or non-numeric values on these nine attributes we were left with $N=91842$ entries, which we shuffled and fed to the algorithm in varying sizes --- from $n=30,000$ to $n=90,000$. Running OLS over the entire $N$ observation yields $\beta \approx (14.07, 0.54, -0.22, 482.59)$, and $t$-Values of $(10.48, 1.25, -2.66, 157.55)$.

\myparagraph{The Algorithms.} We ran a version of Algorithm~\ref{alg:private_JL} that  uses a DP-estimation of $\svmin$, and finds the largest $r$ the we can use without altering the input, yet if this $r$ is below $25$ then it does alter the input and approximates Ridge regression. We ran Algorithm~\ref{alg:AG} verbatim. We set $\epsilon=0.25$ and $\delta = 10^{-6}$. We repeated each algorithm $100$ times.

\myparagraph{Results.} We plot the $t$-values we get from Algorithms~\ref{alg:private_JL} and~\ref{alg:AG} and decide to reject the null-hypothesis based on $t$-value larger than $2.8$ (which corresponds to a fairly conservative $p$-value of $0.005$). Not surprisingly, as $n$ increases, the $t$-values become closer to their expected value -- the $t$-value of Analyze Gauss is close to the non-private $t$-value and the $t$-value from Algorithm~\ref{alg:private_JL} is a factor of $\sqrt{\tfrac r n}$ smaller as detailed above (see after Corollary~\ref{cor:JL_confidence_interval}). As a result, when the null-hypothesis is false, Analyze Gauss tends to produce larger $t$-values (and thus reject the null-hypothesis) for values of $n$ under which Algorithm~\ref{alg:private_JL} still does not reject, as shown in Figure~\ref{subfig:decision-synth-coord1}. This is exacerbated in real data setting, where its actual least singular value ($\approx 500$) is fairly small in comparison to its size ($N=91842$).

However, what is fairly surprising is the case where the null-hypothesis should not be rejected --- since $\beta_j =0$ (in the synthetic case) or its non-private $t$-value is close to $0$ (in the real-data case). Here, the Analyze Gauss' $t$-value approximation has fairly large variance, and we still get fairly high (in magnitude) $t$-values. As the result, we falsely reject the null-hypothesis based on the $t$-value of Analyze Gauss quite often, even for large values of $n$. This is shown in Figure~\ref{subfig:decision-synth-coord3}. Additional figures (including plotting the distribution of the $t$-value approximations) appear in the supplementary material.

The results show that $t$-value approximations that do not take into account the inherent randomness in the DP-algorithms lead to erroneous conclusions. One approach would be to follow the more conservative approach we advocate in this paper, where Algorithm~\ref{alg:private_JL} may allow you to get true approximation of the $t$-values and otherwise reject the null-hypothesis only based on the confidence interval (of Algorithm~\ref{alg:private_JL} or~\ref{alg:AG}) not intersecting the origin. Another approach, which we leave as future work, is to replace the $T$-distribution with a new distribution, one that takes into account the randomness in the estimator as well. This, however, has been an open and long-standing challenge since the first works on DP and statistics (see~\cite{VuS09, DworkL09}) and requires we move into non-asymptotic hypothesis testing.

\vspace{-0mm}
\def \SCALE {0.375} 
\begin{figure}[H]
	\begin{center}
	\begin{subfigure}{0.5\textwidth}
		\includegraphics[scale=\SCALE]{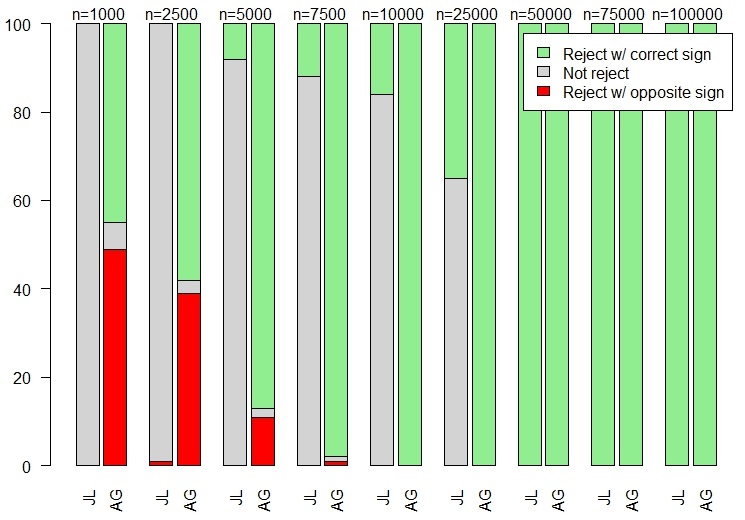}
		\caption{\label{subfig:decision-synth-coord1} Synthetic data, coordinate $\beta_1$}
	\end{subfigure}
	\begin{subfigure}{0.5\textwidth}
		\includegraphics[scale=\SCALE]{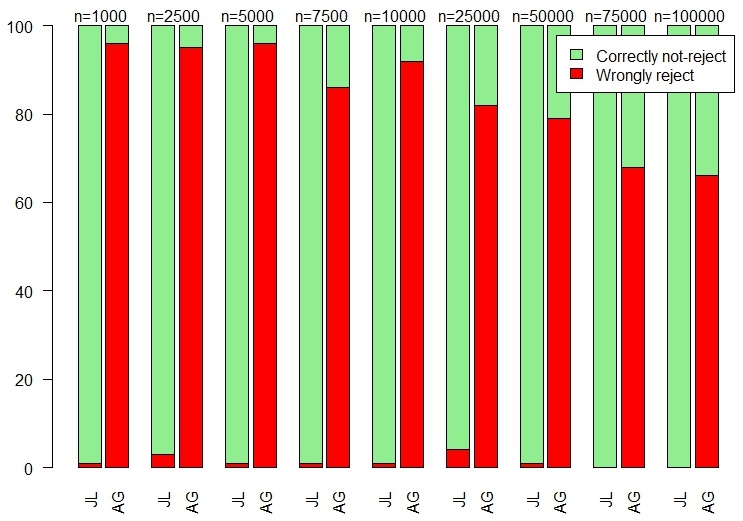}
		\caption{\label{subfig:decision-synth-coord3} Synthetic data, coordinate $\beta_3$}
	\end{subfigure}
	\end{center}
	\caption{Correctly and Wrongly Rejecting the Null-Hypothesis}
\end{figure}

\newpage

\section*{Acknowledgements} 
\acksfull

\ifx\forappendix\undefined
\bibliography{paper}
\bibliographystyle{icml2017}
\newpage
\appendix 
\fi

\omitfromversion{
\section{Omitted Proofs.}
\label{apx:privacy_of_alg}
\begin{theorem}
\label{apx_thm:alg_privacy}
Algorithm~\ref{alg:private_JL} is $(\epsilon,\delta)$-differentially private.
\end{theorem}
\begin{proof}
The proof of the theorem is based on the fact the Algorithm~\ref{alg:private_JL} is the result of composing the differentially private Propose-Test-Release algorithm of~\mycite{DworkL09} with the differentially private analysis of the Johnson-Lindenstrauss transform of~\mycite{Sheffet15}.
\omitproof{
\omitproof{More specifically, we use Theorem B.1 from~\mycite{Sheffet15} that states that given a matrix $A$ whose all of its singular values at greater than $T(\epsilon,\delta)$ where $T(\epsilon,\delta)^2 = \frac{2B^2}{\epsilon}\left(\sqrt {2r\ln(4/\delta)} + 2\ln(4/\delta)\right)$, publishing $RA$ is $(\epsilon,\delta)$-differentially private for a $r$-row matrix $R$ whose entries sampled are i.i.d normal Gaussians. Since we have that all of the singular values of $A'$ are greater than $w$ (as specified in Algorithm~\ref{alg:private_JL}),  outputting $RA'$ is $(\epsilon/2,\delta/2)$-differentially private. The rest of the proof boils down to showing that (i) the if-else-condition is $(\epsilon/2,0)$-differentially private and that (ii) w.p. $\leq \delta/2$ any matrix $A$ whose smallest singular value is smaller than $w$ passes the if-condition (step 3). If both these facts hold, then knowing whether we pass the if-condition or not is $(\epsilon/2$)-differentially private and the output of the algorithm is $(\epsilon/2,\delta)$-differentially private, hence basic composition gives the overall bound of $(\epsilon,\delta)$-differential privacy.

To prove (i) we have that for any pair of neighboring matrices $A$ and $B$ that differ only on the $i$-th row, denoted $\vec a_i$ and $\vec b_i$ resp., we have $B^\T B - \vec b_i \vec b_i^\T = A^\T A - \vec a_i \vec a_i^\T$. Applying Weyl's inequality we have
\[ \svmin(B^\T B) \leq  \svmin(B^\T B-\vec b_i \vec b_i^\T) + \svmax(\vec b_i \vec b_i^\T) \leq \svmin(A^\T A) + \svmax(\vec a_i \vec a_i^\T)  + \svmax(\vec b_i \vec b_i^\T) \leq \svmin(A^\T A) +2B^2\]
hence $|\svmin(A)^2 - \svmin(B)^2| \leq 2 B^2$, so adding $Lap(\frac{4B^2}{\epsilon})$ is $(\epsilon/2)$-differentially private.

To prove (ii), note that by standard tail-bounds on the Laplace distribution we have that $\Pr[Z~<~-\frac {4B^2\ln(1/\delta)}\epsilon] \leq \tfrac \delta 2$. Therefore, w.p. $1-\delta/2$ it holds that any matrix $A$ that passes the if-test of the algorithm must have $\svmin(A)^2 > w^2$.
Also note that a similar argument shows that for any $0 < \beta < 1$, any matrix $A$ s.t. $\svmin(A)^2 > w^2 + \frac {4B^2\ln(1/\beta)}\epsilon$  passes the if-condition of the algorithm w.p. $1-\beta$.}
}
\end{proof}
}

\ifdefined\confversion
\appendix
\ 
\newpage 
\section{Extended Introductory Discussion}
\label{apx:extended_intro}

Due to space constraint, a few details from the introductory parts (Sections~\ref{sec:intro},\ref{sec:preliminaries}) were omitted. We bring them in this appendix. We especially recommend the uninformed reader to go over the extended OLS background we provide in Appendix~\ref{apx_subsec:OLS_background}.

\subsection{Proof Of Privacy of Algorithm~\ref{alg:private_JL}}
\label{apx_subsec:thmprivacy}
\begin{theorem}\thmalgprivacystatement
\end{theorem}
\begin{proof}
The proof of the theorem is based on the fact the Algorithm~\ref{alg:private_JL} is the result of composing the differentially private Propose-Test-Release algorithm of~\mycite{DworkL09} with the differentially private analysis of the Johnson-Lindenstrauss transform of~\mycite{Sheffet15}.
	
More specifically, we use Theorem B.1 from~\mycite{Sheffet15} that states that given a matrix $A$ whose all of its singular values at greater than $T(\epsilon,\delta)$ where $T(\epsilon,\delta)^2 = \frac{2B^2}{\epsilon}\left(\sqrt {2r\ln(4/\delta)} + 2\ln(4/\delta)\right)$, publishing $RA$ is $(\epsilon,\delta)$-differentially private for a $r$-row matrix $R$ whose entries sampled are i.i.d normal Gaussians. Since we have that all of the singular values of $A'$ are greater than $w$ (as specified in Algorithm~\ref{alg:private_JL}),  outputting $RA'$ is $(\epsilon/2,\delta/2)$-differentially private. The rest of the proof boils down to showing that (i) the if-else-condition is $(\epsilon/2,0)$-differentially private and that (ii) w.p. $\leq \delta/2$ any matrix $A$ whose smallest singular value is smaller than $w$ passes the if-condition (step 3). If both these facts hold, then knowing whether we pass the if-condition or not is $(\epsilon/2$)-differentially private and the output of the algorithm is $(\epsilon/2,\delta)$-differentially private, hence basic composition gives the overall bound of $(\epsilon,\delta)$-differential privacy.
		
		To prove (i) we have that for any pair of neighboring matrices $A$ and $B$ that differ only on the $i$-th row, denoted $\vec a_i$ and $\vec b_i$ resp., we have $B^\T B - \vec b_i \vec b_i^\T = A^\T A - \vec a_i \vec a_i^\T$. Applying Weyl's inequality we have
		\ifx\twocols\undefined
		\[ \svmin(B^\T B) \leq  \svmin(B^\T B-\vec b_i \vec b_i^\T) + \svmax(\vec b_i \vec b_i^\T) \leq \svmin(A^\T A) + \svmax(\vec a_i \vec a_i^\T)  + \svmax(\vec b_i \vec b_i^\T) \leq \svmin(A^\T A) +2B^2\]
		\else
		\begin{align*}
		\svmin(B^\T B) &\leq  \svmin(B^\T B-\vec b_i \vec b_i^\T) + \svmax(\vec b_i \vec b_i^\T) \\ &\leq \svmin(A^\T A) + \svmax(\vec a_i \vec a_i^\T)  + \svmax(\vec b_i \vec b_i^\T) \\& \leq \svmin(A^\T A) +2B^2
		\end{align*}
		hence $|\svmin(A)^2 - \svmin(B)^2| \leq 2 B^2$, so adding $Lap(\frac{4B^2}{\epsilon})$ is $(\epsilon/2)$-differentially private.

		To prove (ii), note that by standard tail-bounds on the Laplace distribution we have that $\Pr[Z~<~-\frac {4B^2\ln(1/\delta)}\epsilon] \leq \tfrac \delta 2$. Therefore, w.p. $1-\delta/2$ it holds that any matrix $A$ that passes the if-test of the algorithm must have $\svmin(A)^2 > w^2$.
		Also note that a similar argument shows that for any $0 < \beta < 1$, any matrix $A$ s.t. $\svmin(A)^2 > w^2 + \frac {4B^2\ln(1/\beta)}\epsilon$  passes the if-condition of the algorithm w.p. $1-\beta$.
\end{proof}

\subsection{Omitted Preliminary Details}
\label{apx_subsec:extended_preliminary}

\linearAlgebraPreliminaries

\myparagraph{The Gaussian Distribution.} \gaussianfull

\cut{As mentioned, we require the following proposition.
\begin{proposition}
	[Proposition~\ref{pro:ratio_close_gaussians} restated]\propositionstatement
\end{proposition}}
\begin{proof}
The proof is mere calculation.
\ifx\twocols\undefined 
			\begin{align*}
			& \frac {\PDF_X(x)}{\PDF_{cY}(x)} = \sqrt{ \frac{c^2\lambda^2}{\sigma^2}}\cdot \frac{\exp(-\tfrac {x^2}{2\sigma^2})}{\exp(-\tfrac {x^2}{2c^2\lambda^2})} \leq c\cdot \exp(\frac {x^2}2 (\frac {1} {c^2\lambda^2} - \frac 1 {\sigma^2})) \leq c\cdot \exp(0)=c\cr
			& \frac {\PDF_X(x)}{\PDF_{Y}(x)} = \sqrt{ \frac{\lambda^2}{\sigma^2}}\cdot \frac{\exp(-\tfrac {x^2}{2\sigma^2})}{\exp(-\tfrac {x^2}{2\lambda^2})} \geq c^{-1} \exp( \tfrac {x^2}2(\tfrac 1 {\lambda^2}-\tfrac 1 {\sigma^2}) ) \geq c^{-1} \exp(0) = c^{-1}
			\end{align*}
\else
			\begin{align*}
			 \frac {\PDF_X(x)}{\PDF_{cY}(x)} & = \sqrt{ \frac{c^2\lambda^2}{\sigma^2}}\cdot \frac{\exp(-\tfrac {x^2}{2\sigma^2})}{\exp(-\tfrac {x^2}{2c^2\lambda^2})} \\
			& \leq c\cdot \exp(\frac {x^2}2 (\frac {1} {c^2\lambda^2} - \frac 1 {\sigma^2})) \leq c\cdot \exp(0)=c\\
			 \frac {\PDF_X(x)}{\PDF_{Y}(x)} & = \sqrt{ \frac{\lambda^2}{\sigma^2}}\cdot \frac{\exp(-\tfrac {x^2}{2\sigma^2})}{\exp(-\tfrac {x^2}{2\lambda^2})}\\& \geq c^{-1} \exp( \tfrac {x^2}2(\tfrac 1 {\lambda^2}-\tfrac 1 {\sigma^2}) ) \geq \tfrac{ \exp(0)}c = c^{-1}
			\end{align*}
\fi 
\end{proof}
\myparagraph{The $T_k$-Distribution.} \tdistributionfull

\myparagraph{Differential Privacy facts.} \DPfull

\subsection{Detailed Background on Ordinary Least Squares}
\label{apx_subsec:OLS_background}
For the unfamiliar reader, we give a short description of the model under which OLS operates as well as the confidence bounds one derives using OLS. This is by no means an exhaustive account of OLS and we refer the interested reader to~\mycite{Rao73, MullerS06}.

Given $n$ observations $\{(\vec x_i, y_i)\}_{i=1}^n$ where for all $i$ we have $\vec x_i\in\mathbb{R}^p$ and $y_i \in \mathbb{R}$, we assume the existence of a $p$-dimensional vector $\vec\beta \in \mathbb{R}^p$ s.t. the label $y_i$ was derived by $y_i = \vec\beta^\T\vec x_i + e_i$ where $e_i \sim \gaussian(0,\sigma^2)$ independently (also known as the  homoscedastic Gaussian model). We use the matrix notation where $X$ denotes the $(n\times p)$-matrix whose rows are $\vec x_i$, and use $\vec y,\vec e \in \mathbb{R}^n$ to denote the vectors whose $i$-th entry is $y_i$ and $e_i$ resp. To simplify the discussion, we assume $X$ has full rank.

The parameters of the model are therefore $\vec \beta$ and $\sigma^2$, which we set to discover. To that end, we minimize $\min_{\vec z} \| \vec y - X\vec z\|^2$ and solve
\begin{equation*}
\hvec\beta  = (X^\T X)^{-1} X^\T \vec y = (X^\T X)^{-1} X^\T (X\vec\beta + \vec e) = \vec\beta + X^\mpi\vec e\end{equation*}
As $\vec e\sim \gaussian(\vec 0_n, \sigma^2 I_{n\times n})$, it holds that $\hat{\vec\beta} \sim \gaussian(\vec\beta, \sigma^2 (X^\T X)^{-1})$, or alternatively, that for every coordinate $j$ it holds that $\hat\beta_j = \vec e_j^\T \hat{\vec \beta} \sim \gaussian(\beta_j, \sigma^2 (X^\T X)^{-1}_{j,j})$. Hence we get $ \frac{\hat\beta_j - \beta_j}{\sigma \sqrt{(X^\T X)^{-1}_{j,j}}} \sim \gaussian(0,1)$.
In addition, we denote the vector 
\begin{equation*}\vec\zeta = y-X\hat{\vec\beta} = (X\vec\beta + \vec e) - X(\vec\beta + X^\mpi\vec e) = (I - XX^\mpi)\vec e \end{equation*} and since $XX^\mpi$ is a rank-$p$ (symmetric) projection matrix, we have $\vec\zeta \sim\gaussian(0, \sigma^2(I-XX^\mpi))$. Therefore, $\|\vec\zeta\|^2$ is equivalent to summing the squares of $(n-p)$ i.i.d samples from $\gaussian(0,\sigma^2)$. In other words, the quantity $\|\vec\zeta\|^2/\sigma^2$ is sampled from a $\chi^2$-distribution with $(n-p)$ degrees of freedom.

We sidetrack from the OLS discussion to give the following bounds on the $l_2$-distance between $\vec\beta$ and $\hvec\beta$, as the next claim shows.
\begin{claim}
	\label{clm:sizes_hatbeta_zeta}
	For any $0<\nu<1/2$, the following holds w.p. $\geq 1- \nu$ over the \emph{randomness of the model} (the randomness over  $\vec e$)
	\begin{align}
	{ \|\vec\beta-\hvec\beta\|^2 } &= \|X^\mpi \vec e\|^2 \cr &= O\left(  \sigma^2  \log(p/\nu) \cdot \|X^\mpi\|_F^2 \right) \label{eq:diff_beta_hatbeta}\\
	\|\hvec\beta\|^2 &= \|\vec\beta + X^\mpi \vec e\|^2 \cr & = O(\left( \|\vec\beta\|+ \sigma\cdot \|X^\mpi\|_F \cdot\sqrt{\log(p/\nu)} \right)^2 ) \cr
	\left|\tfrac 1 {n-p}\|\vec\zeta\|^2 - \sigma^2\right| &= O(\sqrt{\tfrac {\ln(1/\nu)}{n-p}}) \notag
	\end{align}
\end{claim}
\begin{proof}
		Since $\vec e \sim\gaussian(\vec 0_n, \sigma^2 I_{n\times n})$ then $X^\mpi \vec e \sim \gaussian(\vec 0_n, \sigma^2 (X^\T X)^{-1})$. Denoting the SVD decomposition $(X^\T X)^{-1} = V S V^\T$ with $S$ denoting the diagonal matrix whose entries are $\svmax^{-2}(X),\ldots, \svmin^{-2}(X)$, we have that $V^\T X^\mpi\vec e \sim\gaussian(\vec 0_n, \sigma^2 S)$. And so, each coordinate of  $V^\T X^\mpi \vec e$ is distributed like an i.i.d Gaussian. So w.p. $\geq 1- \nu/2$ non of these Gaussians is a factor of $O(\sigma\sqrt{\ln(p/\nu)})$ greater than its standard deviation. And so w.p. $\geq 1-\nu/2$ it holds that $\|X^\mpi \vec e\|^2 = \|V^\T X^\mpi\vec e\|^2 \leq O(\sigma^2 \log(p/\nu) \left( \sum_{i}\sigma_i^{-2}(X) \right) )$. Since $\sum_{i}\sigma_i^{-2}(X) = \trace((X^\T X)^{-1}) = \trace( X^\mpi (X^\mpi)^\T) = \|X^\mpi\|_F^2$, the bound of~\eqref{eq:diff_beta_hatbeta} is proven.
		
		The bound on $\|\hvec\beta\|^2$ is an immediate corollary of~\eqref{eq:diff_beta_hatbeta} using the triangle inequality.\footnote{Observe, though $\vec e$ is spherically symmetric, and is likely to be approximately-orthogonal to $\vec\beta$, this does not necessarily hold for $X^\mpi\vec e$ which isn't spherically symmetric. Therefore, we result to bounding the $l_2$-norm of $\hvec\beta$ using the triangle bound.} The bound on $\|\vec\zeta\|^2$ follows from tail bounds on the $\chi^2_{n-p}$ distribution, as detailed in Section~\ref{sec:preliminaries}.
\end{proof}

Returning to OLS, it is important to note that $\hvec\beta$ and $\vec\zeta$ are independent of one another. (Note, $\hvec\beta$ depends solely on $X^\mpi \vec e = (X^\mpi X)X^\mpi\vec e = X^\mpi P_U\vec e$, whereas $\vec\zeta$ depends on $(I-XX^\mpi)\vec e = P_{U^\perp}\vec e$. As $\vec e$ is spherically symmetric, the two projections are independent of one another and so $\hvec\beta$ is independent of $\vec\zeta$.) As a result of the above two calculations, we have that the quantity
\[ \resizebox*{0.49\textwidth}{!}{$t_{\hat\beta_j}(\beta_j) \stackrel{\rm def}{=} \frac {\hat\beta_j -\beta_j} {\sqrt{(X^\T X)^{-1}_{j,j}}\cdot  \frac{\|\vec\zeta\|}{\sqrt{n-p}}} = \frac {\hat\beta_j -\beta_j} {\sigma\sqrt{(X^\T X)^{-1}_{j,j}}} \Big/  \frac{\|\vec\zeta\|}{\sigma\sqrt{n-p}}
$}\]
is distributed like a $T$-distribution with $(n-p)$ degrees of freedom.  Therefore, we can compute an exact probability estimation for this quantity. That is, for any measurable $S\subset \mathbb{R}$ we have
\ifx\twocols\undefined
\[ \Pr\left[\hvec\beta\textrm{ and }\vec\zeta\textrm{ satisfying } \frac {\hat\beta_j -\beta_j} {\sqrt{(X^\T X)^{-1}_{j,j}}\cdot  \frac{\|\vec\zeta\|}{\sqrt{n-p}}} \in S \right] = \int_S \PDF_{T_{n-p}}(x)dx\]
\else 
\[ \Pr\left[\hvec\beta\textrm{ and }\vec\zeta\textrm{ satisfying } t_{\hat\beta_j}(\beta_j) \in S \right] = \int_S \PDF_{T_{n-p}}(x)dx\]
\fi
The importance of the \emph{$t$-value} $t(\beta_j)$ lies in the fact that it can be fully estimated from the observed data $X$ and $y$ (for any value of $\beta_j$), which makes it a \emph{pivotal quantity}. Therefore, given $X$ and $\vec y$, we can use $t(\beta_j)$ to describe the likelihood of any $\beta_j$~--- for any $z\in\mathbb{R}$ we can now give an estimation of how likely it is to have $\beta_j=z$ (which is $\PDF_{T_{n-p}}(t(z))$). 
The $t$-values enable us to perform multitude of statistical inferences. For example, we can say which of two hypotheses is more likely and by how much (e.g., we are $5$-times more likely that the hypothesis $\beta_j =3$ is true than the hypothesis $\beta_j =14$ is true); we can compare between two coordinates $j$ and $j'$ and report we are more confident that $\beta_j >0$ than $\beta_{j'}>0$; or even compare among the $t$-values we get across multiple datasets (such as the datasets we get from subsampling rows from a single dataset).

In particular, we can use $t(\beta_j)$ to $\alpha$-reject unlikely values of $\beta_j$. Given $0<\alpha < 1$, we denote $c_\alpha$ as the number for which the interval $(- c_\alpha, c_\alpha)$ contains a probability mass of $1-\alpha$ from the $T_{n-p}$-distribution. And so we derive a corresponding \emph{confidence interval} $I_\alpha$ centered at $\hat\beta_j$ where $\beta_j \in I_\alpha$ with confidence of level of $1-\alpha$. \omitfromversion{Using tail bounds on the $T_{n-p}$-distribution~\mycite{Soms76}, we have that the length of the interval is $|I_\alpha| = O\left(\sqrt{(X^\T X)^{-1}_{j,j} \cdot \tfrac {\|\vec\zeta\|^2}{n-p}} \cdot \sqrt{(n-p)((\tfrac 1 \alpha)^{\tfrac 2 {n-p-1}}-1)} \right)$. Furthermore, since it is known that as the number of degrees of freedom of a $T$-distribution tends to infinity then the $T$-distribution becomes close to a normal Gaussian, it is common to use the $\PDF$ of a normal Gaussian instead. I.e., denote $\tau_\alpha$ as the number of which $\int_{\tau_\alpha}^\infty \PDF_{\gaussian(0,1)}(x)dx = \tfrac \alpha 2$, then $I_\alpha = \beta_j \pm \tau_\alpha \sqrt{(X^\T X)^{-1}_{j,j} \cdot \tfrac {\|\vec\zeta\|^2}{n-p}}$.}

We comment as to the actual meaning of this confidence interval. Our analysis thus far applied w.h.p to a vector $\vec y$ derived according to this model. Such $X$ and $\vec y$ will result in the quantity $t_{\hat\beta_j}(\beta_j)$ being distributed like a $T_{n-p}$-distribution~--- where $\beta_j$ is given as the model parameters and $\hat\beta_j$ is the random variable. We therefore have that guarantee that for $X$ and $\vec y$ derived according to this model, the event $E_\alpha \stackrel{\rm def}= \hat\beta_j \in  \left(\beta_j \pm c_\alpha\cdot \sqrt{(X^\T X)^{-1}_{j,j} \cdot \tfrac {\|\vec\zeta\|^2}{n-p}} \right)$ happens w.p. $1-\alpha$.  However, the analysis done over a \emph{given} dataset $X$ and $\vec y$ (once $\vec y$ has been drawn) views the quantity $t_{\hat\beta_j}(\beta_j)$ with $\hat\beta_j$ given and $\beta_j$ unknown. Therefore the event $E_\alpha$ either holds or does not hold. That is why the alternative terms of \emph{likelihood} or \emph{confidence} are used, instead of probability. We have a confidence level of $1-\alpha$ that indeed $\beta_j \in \hat{\beta_j}\pm c_\alpha\cdot \sqrt{(X^\T X)^{-1}_{j,j} \cdot \tfrac {\|\vec\zeta\|^2}{n-p}}$, because this event does happen in $1-\alpha$ fraction of all datasets generated according to our model.

\myparagraph{Rejecting the Null Hypothesis.} One important implication of the quantity $t(\beta_j)$ is that we can refer specifically to the hypothesis that $\beta_j=0$, called the \emph{null hypothesis}. This quantity, $t_0 \stackrel{\rm def} = t_{\hat\beta_j}(0) = \frac {\hat\beta_j\sqrt{n-p}} {\|\vec\zeta\|\sqrt{(X^\T X)^{-1}_{j,j}}}$, represents how large is $\hat\beta_j$ relatively to the empirical estimation of standard deviation $\sigma$. Since it is known that as the number of degrees of freedom of a $T$-distribution tends to infinity then the $T$-distribution becomes a normal Gaussian, it is common to think of $t_0$ as a sample from a normal Gaussian $\gaussian(0,1)$. This allows us to associate $t_0$ with a $p$-value, estimating the event ``$\beta_j$ and $\hat\beta_j$ have different signs.''
Formally, we define $p_0 = \int_{|t_0|}^{\infty} \frac 1 {\sqrt{2\pi}} e^{-x^2/2}dx$. It is common to reject the null hypothesis when $p_0$ is sufficiently small (typically, below $0.05$).\footnote{Indeed, it is more accurate to associate with $t_0$ the value $\int_{|t_0|}^\infty \PDF_{T_{n-p}}(x)dx$ and check that this value is $<\alpha$. However, as most uses take $\alpha$ to be a constant (often $\alpha = 0.05$), asymptotically the threshold we get for rejecting the null hypothesis are the same.}

Specifically, given $\alpha \in (0,1/2)$, we say we \emph{$\alpha$-reject the null hypothesis} if $p_0 <\alpha$. Let $\tau_\alpha$ be the number s.t. $\Phi(\tau_\alpha) = \int_{\tau_\alpha}^{\infty} \tfrac 1 {\sqrt{2\pi}} e^{-x^2/2}dx = \alpha$. (Standard bounds give that $\tau_\alpha < 2\sqrt{\ln(1/\alpha)}$.) This means we $\alpha$-reject the null hypothesis if  $t_0 > \tau_\alpha$ or $t_0 < -\tau_\alpha$, meaning if $|\hat\beta_j| > \tau_\alpha \sqrt{(X^\T X)^{-1}_{j,j}}\tfrac {\|\vec\zeta\|}{\sqrt{n-p}}$.

We can now lower bound the number of i.i.d sample points needed in order to $\alpha$-reject the null hypothesis. This bound will be our basis for comparison~--- between standard OLS and the differentially private version.\footnote{This theorem is far from being new (except for maybe focusing on the setting where every row in $X$ is sampled from an i.i.d multivariate Gaussians), it is just stated in a non-standard way, discussing solely the power of the $t$-test in OLS. For further discussions on sample size calculations see~\mycite{MullerS06}.}

\begin{theorem}
	[Theorem~\ref{thm:baseline_rejecting_nh} restated.]
	Fix any positive definite matrix $\Sigma\in\mathbb{R}^{p\times p}$ and any $\nu \in (0,\tfrac 1 2)$. Fix parameters $\vec\beta \in \mathbb{R}^p$ and $\sigma^2$ and a coordinate $j$ s.t. $\beta_j \neq 0$. Let $X$ be a matrix whose $n$ rows are i.i.d samples from $\gaussian(\vec 0, \Sigma)$, and $\vec y$ be a vector where $y_i - (X\vec\beta)_i$ is sampled i.i.d from $\gaussian(0,\sigma^2)$. Fix $\alpha \in (0,1)$. 
	Then w.p. $\geq 1-\nu$ we have that the $(1-\alpha)$-confidence interval is of length $O(c_\alpha\sqrt{\sigma^2 / (n\svmin(\Sigma)}))$ provided $n \geq C_1 (p+\ln(1/\nu))$ for some sufficiently large constant $C_1$. Furthermore, there exists a constant $C_2$ such that w.p. $\geq 1- \alpha-\nu$ we (correctly) reject the null hypothesis provided
	\[ n \geq \max\left\{C_1 (p+\ln(1/\nu)), ~~ C_2 \frac {\sigma^2}{\beta_j^2} \cdot \frac {c_\alpha^2+\tau_\alpha^2} {\svmin(\Sigma)} \right\}\] 
	Here $c_\alpha$ denotes the number for which $\int_{-c_\alpha}^{c_\alpha} \PDF_{T_{n-p}}(x)dx = 1-\alpha$. (If we are content with approximating $T_{n-p}$ with a normal Gaussian than one can set $c_\alpha \approx \tau_\alpha  < 2\sqrt{\ln(1/\alpha)}$.) 
\end{theorem}
\begin{proof}
		The discussion above shows that w.p. $\geq 1- \alpha$ we have $|\beta_j- \hat\beta_j| \leq c_\alpha \sqrt{(X^\T X)^{-1}_{j,j}\tfrac {\|\vec\zeta\|^2}{n-p}}$; and in order to $\alpha$-reject the null hypothesis we must have $|\hat\beta_j| > \tau_\alpha \sqrt{(X^\T X)^{-1}_{j,j}\tfrac {\|\vec\zeta\|^2}{n-p}}$. Therefore, a sufficient condition for OLS to $\alpha$-reject the null-hypothesis is to have $n$ large enough s.t. $|\beta_j | > (c_\alpha+\tau_\alpha) \sqrt{(X^\T X)^{-1}_{j,j}\tfrac {\|\vec\zeta\|^2}{n-p}}$. We therefore argue that w.p.$\geq 1-\nu$ this inequality indeed holds.
		
		We assume each row of $X$ i.i.d vector $\vec x_i \sim \gaussian(\vec 0_p, \Sigma)$, and recall that according to the model $\|\vec\zeta\|^2 \sim \sigma^2 \chi^2(n-p)$. Straightforward concentration bounds on Gaussians and on the $\chi^2$-distribution give:\\
		\indent (i) W.p. $\leq \alpha$ it holds that $\|\vec\zeta\| > \sigma\left( \sqrt{n-p} + 2\ln(2/\alpha))\right)$. (This is part of the standard OLS analysis.)\\
		\indent (ii) W.p. $\leq \nu$ it holds that $\svmin(X^\T X) \leq \svmin(\Sigma) ( \sqrt n - (\sqrt p + \sqrt{2\ln(2/\nu)}))^2$.~\mycite{RudelsonV09}\\
		Therefore, due to the lower bound $n = \Omega(p + \ln(1/\nu))$, w.p.$\geq 1- \nu-\alpha$ we have that none of these events hold. In such a case we have $\sqrt{(X^\T X)^{-1}_{j,j}} \leq \sqrt{\svmax((X^\T X)^{-1})} = O(\tfrac 1 {\sqrt{n \svmin(\Sigma)}})$ and $\|\vec\zeta\| = O(\sigma\sqrt{n-p})$. This implies that the confidence interval of level $1-\alpha$ has length of $c_\alpha\sqrt{(X^\T X)^{-1}_{j,j} \cdot \tfrac {\|\vec\zeta\|^2}{n-p}} = O\left( c_\alpha\sqrt{ \tfrac {\sigma^2} {n\svmin(\Sigma)}} \right)$; and that in order to $\alpha$-reject that null-hypothesis it suffices to have $|\beta_j | = \Omega\left((c_\alpha+\tau_\alpha) \sqrt{\tfrac {\sigma^2} {n\svmin(\Sigma)}}\right)$. Plugging in the lower bound on $n$, we see that this inequality holds.
		
		We comment that for sufficiently large constants $C_1,C_2$, it holds that all the constants hidden in the $O$- and $\Omega$-notations of the proof are close to $1$. I.e., they are all within the interval $(1\pm \eta)$ for some small $\eta>0$ given $C_1, C_2 \in \Omega(\eta^{-2})$. 
\end{proof}

\section{Projecting the Data using Gaussian Johnson-Lindenstrauss Transform}
\label{apx_sec:proof_main_thm}

\subsection{Main Theorem Restated and Further Discussion}
\label{apx_subsec:main_thm_discussion}

\begin{theorem}[Theorem~\ref{thm:main_theorem_matrix_unaltered} restated.]
	Let $X$ be a $n\times p$ matrix, and parameters $\vec\beta \in \mathbb{R}^p$ and $\sigma^2$ are such that we generate the vector $\vec y= X\vec\beta + \vec e$ with each coordinate of $\vec e$ sampled independently from $\gaussian (0,\sigma^2)$. 
	Assume $\svmin(X) \geq C\cdot w$ and that $n$ is sufficiently large s.t. all of the singular values of the matrix $[X;\vec y]$ are greater than $C \cdot w$ for some large constant $C$, and so Algorithm~\ref{alg:private_JL} projects the matrix $A=[X;\vec y]$ without altering it, and publishes $[RX; R\vec y]$.\\
	Fix $\nu\in(0,1/2)$ and $r = p+\Omega(\ln(1/\nu))$. Fix coordinate $j$. Then w.p. $\geq 1-\nu$ we have that deriving $\tvec\beta$, $\tvec{\zeta}$ and $\tilde\sigma^2$ as follows
	\begin{eqnarray*}
	&{\tvec\beta} &= (X^\T R^\T R X)^{-1} (RX)^\T (R\vec y) = \vec\beta + (RX)^\mpi R\vec e \\
	&{\tvec\zeta} &= \tfrac 1 {\sqrt r} R\vec y - \tfrac 1 {\sqrt r} (RX)\tvec\beta \cr && = \tfrac 1 {\sqrt r} \left(I - (RX)(X^\T R^\T RX)^{-1}(RX)^\T)\right)R\vec e \\
	& \tilde\sigma^2 &=\frac r{r-p}  \|\tvec\zeta\|^2 
	\end{eqnarray*}
	then the pivot quantity
	\[ \tilde t(\beta_j) = \frac {\tilde\beta_j-\beta_j} {\tilde\sigma \sqrt{(X^\T R^\T R X)^{-1}_{j,j}}}\]
	has a distribution $\mathcal{D}$ satisfying $e^{-a}\PDF_{T_{r-p}}(x) \leq \PDF_{\mathcal{D}}(x) \leq e^{a} \PDF_{T_{r-p}}(e^{-a}x)$ for any $x\in \mathbb{R}$, where we denote $a= \tfrac {r-p}{n-p}$.\end{theorem}

\JLDiscussion

\ifdefined\showexperiment
\myparagraph{Experiment showing Theorem~\ref{thm:main_theorem_matrix_unaltered} is stronger than standard JLT-approximation guarantees for linear regression.} We note that Sarlos' work give guarantees that hold for all versions of the JLT (such as sparse versions). In contrast, Theorem~\ref{thm:main_theorem_matrix_unaltered} holds solely for the Gaussian version of the JLT. We thus set to show, empirically, that the guarantees of  Theorem~\ref{thm:main_theorem_matrix_unaltered} are indeed stronger than those one gets via plugging off-the-shelf results regarding JLT and linear regression. To that end, we have conducted the following experiment. We sampled (once) $n=512$ observations in a $p=4$-dimensional space, from a Gaussian whose mean is $\vec 0_4$ and variance is $\left(\begin{array}{c c c c}
1  & 0.7 & 0 & 0.1 \cr 0.7 &1 & 0.7 & 0.1 \cr 0 & 0.7 & 1 & -0.1 \cr0.1 & 0.1 & -0.1 &1\end{array} \right)$ (we chose this matrix as it is positive definite, but its least singular value is fairly small). We also fixed $\vec{\beta} = (0.5, -0.5, 0,0)$ and $r=10$. We then repeated the following experiment $T=10^6$ times: we sampled $\vec y$ from the homoscedastic model (with $\sigma^2 = 0.25$) and denoted $A=[X;\vec y]$. We then 
\begin{enumerate}
\item project $A$ using a JLT onto $r$ dimensions
\item project $A$ using the Ailon-Chazelle JLT onto $r$ dimensions 
\item sample u.a.r a subset of $r$ rows from $A$
\end{enumerate} We ran OLS over each output and computed the appropriate $t$-values. By definition, the $t$-values of technique (3) are $T$ i.i.d samples from a $T_{r-p}=T_5$ distribution. From Theorem~\ref{thm:main_theorem_matrix_unaltered} we have that the $t$-values of technique (1) are $T$ i.i.d samples from a distribution very close to $T_5$. Thus, our goal in this experiment is to show that the samples taken from technique (2) come from a distribution which can \emph{not} be a good approximation for $T_5$. Indeed, empirically we do find that the samples from technique (2) are more concentrates around $0$ then they ought to be.

We sort the $T$ samples by size, and for every $\tau >0$ we empirically estimate $\Pr[t\textrm{-values from technique (i)} > \tau]$. In Figure~\ref{fig:???} we plot ($4$ times, one plot for each coordinate) for each $\tau \in [0,10]$ our empirical estimation of
\begin{itemize}
\item $\Pr[t\textrm{-values from technique (3)} > \tau]-\Pr[t\textrm{-values from technique (1)} > \tau]$ (in black)  
\item $\Pr[t\textrm{-values from technique (3)} > \tau]-\Pr[t\textrm{-values from technique (2)} > \tau]$ (in red)
\end{itemize}
\fi 
\subsection{Proof of Theorem~\ref{thm:main_theorem_matrix_unaltered}}
\label{apx_subsec:proof_main_thm}

\proofmainthm

\subsection{Proof of Theorem~\ref{thm:JL_simple_case_rejecting_null_hypothesis}}
\label{apx_subsec:comparison_bound}

\begin{theorem}[Theorem~\ref{thm:JL_simple_case_rejecting_null_hypothesis} restated.]
Fix a positive definite matrix $\Sigma \in \mathbb{R}^{p\times p}$. Fix parameters $\vec\beta\in\mathbb{R}^p$ and $\sigma^2 > 0$ and a coordinate $j$ s.t. $\beta_j\neq 0$. Let $X$ be a matrix whose $n$ rows are sampled i.i.d from $\gaussian(\vec 0_p, \Sigma)$. Let $\vec y$ be a vector s.t. $y_i-(X\vec\beta)_i$ is sampled i.i.d from $\gaussian(0,\sigma^2)$. Fix $\nu \in (0,1/2)$ and $\alpha \in (0,1/2)$. Then there exist constants $C_1$, $C_2$, $C_3$ and $C_4$ such that when we run Algorithm~\ref{alg:private_JL} over $[X;\vec y]$ with parameter $r$  w.p. $\geq 1-\nu$ we correctly $\alpha$-reject the null hypothesis using $\tilde p_0$ (i.e., w.p. $\geq 1-\nu$ Algorithm~\ref{alg:private_JL} returns $\mathtt{matrix~unaltered}$ and we can estimate $\tilde t_0$ and verify that indeed $\tilde p_0 < \alpha\cdot e^{-\tfrac {r-p}{n-p}}$) provided
\[  r \geq p + \max\left\{C_1\frac{\sigma^2 (\tilde c_\alpha^2+\tilde\tau_\alpha^2)}{\beta_j^2\svmin(\Sigma)},~C_2\ln(1/\nu)\right\}
\ifx\twocols\undefined \textrm{,~and~} \else \] and \[\fi  n \geq \max\left\{r, ~ C_3  \frac {w^2} {\min\{\svmin(\Sigma),\sigma^2\}}, ~ C_4 (p+ \ln(1/\nu))\right\} \]
where $\tilde c_\alpha$, $\tilde\tau_\alpha$ denote the numbers s.t. $\int\limits_{\tilde c_\alpha / e^{\tfrac {r-p}{n-p}}}^\infty\PDF_{T_{r-p}}(x)dx = \frac \alpha 2 e^{-\tfrac {r-p}{n-p}}$  and $ \int\limits_{\tilde \tau_\alpha / e^{\tfrac {r-p}{n-p}}}^\infty\PDF_{\gaussian(0,1)}(x)dx = \tfrac \alpha 2 e^{-\tfrac {r-p}{n-p}}$ resp.
\end{theorem}
\begin{flushleft}
	
\end{flushleft}\comparisonboundproof

\section{Projected Ridge Regression}
\label{apx_sec:ridge_regression}
In this section we deal with the case that our matrix does not pass the if-condition of Algorithm~\ref{alg:private_JL}. In this case, the matrix is appended with a $d\times d$-matrix which is $wI_{d\times d}$. Denoting $A' = \left[ \begin{array}{c} A \cr w\cdot I_{d\times d}\end{array}\right]$ we have that the algorithm's output is $RA'$. 

Similarly to before, we are going to denote $d = p+1$ and decompose  $A = [X; \vec y]$ with $X\in \mathbb{R}^{n\times p}$ and $\vec y \in \mathbb{R}^n$, with the standard assumption of $\vec y = X\vec\beta + \vec e$ and $e_i$ sampled i.i.d from $\gaussian(0,\sigma^2)$.\footnote{Just as before, it is possible to denote any single column as $\vec y$ and any subset of the remaining columns as $X$.} We now need to introduce some additional notation. We denote the appended matrix and vectors $X'$ and $\vec y'$ s.t. $A' = [X';\vec y']$. 
Meaning:
\[
X' = \left[ \begin{array}{c}
		X \cr w I_{p\times p} \cr \vec 0_p^\T
	\end{array} \right]  
\ifx\twocols\undefined \textrm{~, and } ~~ \else \] and \[ \fi 
\vec y' = \left[\begin{array}{c}
	\vec y \cr \vec 0_p \cr w
\end{array} \right] 
= X'\vec\beta + \left[\begin{array}{c}
\vec e \cr -w\vec \beta \cr w
\end{array} \right] \stackrel{\rm def}= X'\vec\beta + \vec e'
\]
And so we respectively denote $R = \left[ R_1 ; R_2 ; R_3 \right]$ with $R_1\in \mathbb{R}^{r\times n}$, $R_2\in \mathbb{R}^{r\times p}$ and $R_3\in \mathbb{R}^{r\times 1}$ (so $R_3$ is a vector denoted as a matrix). Hence: 
\[ M' = RX' = R_1 X + w R_2 
\ifx\twocols\undefined \textrm{~, and } ~~ \else \] and \[ \fi 
R\vec y' = RX'\vec\beta + R\vec e' = R_1\vec y + w R_3 = R_1X\vec \beta + R_1 \vec e + wR_3
\] 

And so, using the output $RA'$ of Algorithm~\ref{alg:private_JL}, we solve the linear regression problem derived from $\tfrac 1 {\sqrt r}RX'$ and $\tfrac 1 {\sqrt r}R\vec y'$. I.e., we set
\begin{eqnarray*} \vec\beta' && = \arg\min_{\vec z} \tfrac 1 r \|R\vec y' - RX'\vec z \|^2 
\ifx\twocols\undefined = \else \cr&&= \fi 
(X'^\T R^\T R X')^{-1}(RX')^\T (R\vec y') 
\label{eq:jl_ridge_regression_problem}\end{eqnarray*}

Sarlos' results~\mycitet{Sarlos06} regarding the Johnson Lindenstrauss transform give that, when $R$ has sufficiently many rows, solving the latter optimization problem gives a good approximation for the solution of the optimization problem
\[ \resizebox*{0.49\textwidth}{!}{$\vec\beta^R = \arg\min_{\vec z} \|\vec y' - X'\vec z\|^2 = \arg\min_{\vec z} \left(\| \vec y-X\vec z \|^2 + w^2\|\vec z\|^2 \right)  $}  \] 
The latter problem is known as the Ridge Regression problem. Invented in the 60s~\mycite{Tikhonov63, HoerlK70},
Ridge Regression is often motivated from the perspective of penalizing linear vectors whose coefficients are too large. It is also often applied in the case where $X$ doesn't have full rank or is close to not having full-rank. That is because the Ridge Regression problem is always solvable. One can show that the minimizer $\vec\beta^R = (X^\T X + w^2 I_{p\times p})^{-1}X^\T \vec y$ is the unique solution of the Ridge Regression problem and that the RHS is always defined (even when $X$ is singular). 

The original focus of Ridge Regression is on penalizing $\vec\beta^R$ for having large coefficients. Therefore, Ridge Regression actually poses a family of linear regression problems: $\min_{\vec z} \|y-X\vec z\| + \lambda \|\vec z\|^2$, where one may set $\lambda$ to be any non-negative scalar. And so, much of the literature on Ridge Regression is devoted to the art of fine-tuning this penalty term~--- either empirically or based on the $\lambda$ that yields the best risk: $ \|\E[\vec\beta^R]-\vec\beta\|^2+ \textrm{Var}(\vec\beta^R)$.\footnote{Ridge Regression, as opposed to OLS, does not yield an unbiased estimator. I.e., $\E[\vec\beta^R] \neq \vec\beta$.}
Here we propose a fundamentally different approach for the choice of the normalization factor~--- we set it so that solution of the regression problem would satisfy $(\epsilon,\delta)$-differential privacy (by projecting the problem onto a lower dimension).

While the solution of the Ridge Regression problem might have smaller risk than the OLS solution, it is not known how to derive $t$-values and/or reject the null hypothesis under Ridge Regression (except for using $X$ to manipulate $\vec\beta^R$ back into $\hvec\beta = (X^\T X)^{-1} X^\T \vec y$ and relying on OLS). In fact, prior to our work there was no need for such analysis! For confidence intervals one could just use the standard OLS, because access to $X$ and $\vec y$ was given.  

Therefore, much for the same reason, we are unable to derive $t$-values under projected Ridge Regression.\footnote{Note: The na\"ive approach of using $RX'$ and $R\vec y'$ to interpolate $RX$ and $R\vec y$ and then apply Theorem~\ref{thm:main_theorem_matrix_unaltered} using these estimations of $RX$ and $R\vec y$ ignores the noise added from appending the matrix $A$ into $A'$, and it is therefore bound to produce inaccurate estimations of the $t$-values.} Clearly, there are situations where such confidence bounds simply cannot be derived. (Consider for example the case where $X = 0_{n\times p}$ and $\vec y$ is just i.i.d draws from $\gaussian(0,\sigma^2)$, so obviously $[X;y]$ gives no information about $\vec\beta$.) Nonetheless, under additional assumptions about the data, our work can give confidence intervals for $\beta_j$, and in the case where the interval doesn't intersect the origin~--- assure us that $sign(\beta_j') = sign(\beta_j)$ w.h.p.

	Clearly, Sarlos' work~\mycitet{Sarlos06} gives an upper bound on the distance $\|\vec\beta' - \vec\beta^R\|$.  However, such distance bound doesn't come with the coordinate by coordinate confidence guarantee we would like to have. In fact, it is not even clear from Sarlos' work that $\E[\vec\beta'] = \vec\beta^R$ (though it is obvious to see that $\E[ (X'^\T R^\T R X')] \vec\beta^R = \E[(RX')^\T R\vec y' ]$). Here, we show that $\E[\vec\beta'] = \hvec\beta$ which, more often than not, does not equal $\vec\beta^R$.
	
	\myparagraph{Comment about notation.} Throughout this section we assume $X$ is of full rank and so $(X^\T X)^{-1}$ is well-defined. If $X$ isn't full-rank, then one can simply replace any occurrence of $(X^\T X)^{-1}$ with $X^\mpi (X^\mpi)^\T$. This  makes all our formulas well-defined in the general case. 
	
	\subsection{Running OLS on the Projected Data}
	\label{subsec:linear_regression_on_projected_altered_data}
	
	In this section, we analyze the projected Ridge Regression, under the assumption (for now) that $\vec e$ is fixed. That is, for now we assume that the only source of randomness comes from picking the matrix $R = [R_1 ; R_2 ; R_3]$. As before, we analyze the distribution over $\vec\beta'$ (see Equation~\eqref{eq:jl_ridge_regression_problem}), and the value of the function we optimize at $\vec\beta'$. Denoting $M'=RX'$, we can formally express the estimators:
	\begin{eqnarray}
	&\vec\beta' &= (M'^\T M')^{-1} M'^\T R\vec y' \label{eq:beta'}\\
	&\vec\zeta' &= \tfrac 1 {\sqrt{r}} (R\vec y' - RX'\vec\beta') \label{eq:zeta'}
	\end{eqnarray}
	\begin{claim}
		\label{clm:dist_beta'_zeta'}
		Given that $\vec y = X\vec \beta + \vec e$ for a fixed $\vec e$, and given $X$ and $M' = RX'= R_1X + wR_2$ we have that
		\begin{align*}
		 \vec\beta' &\sim \gaussian\Big( \vec\beta + X^\mpi\vec e,
		\ifx\twocols\undefined   ~~  \else \cr &~~~~~~~ \fi 
		(w^2(\|\vec\beta+X^\mpi\vec e\|^2+1)  + \|P_{U^\perp}\vec e\|^2 ) (M'^\T M')^{-1} \Big)\cr
		\vec\zeta' &\sim \gaussian\Big( \vec 0_r, 
		\ifx\twocols\undefined   ~~  \else \cr &~~~~~~~ \fi 
		\resizebox*{0.25\textwidth}{!}{$\frac {w^2(\|\vec\beta+X^\mpi\vec e\|^2+1)  + \|P_{U^\perp}\vec e\|^2 }{r}$} (I_{r\times r}-M'M'^\mpi) \Big)
		\end{align*}
		and furthermore, $\vec\beta'$ and $\vec\zeta'$ are independent of one another.
	\end{claim}
	\begin{proof}
		First, we write $\vec\beta'$ and $\vec \zeta'$ explicitly, based on $\vec e$ and projection matrices:
		\begin{eqnarray*}
			\vec\beta' &&= (M'^\T M')^{-1} M'^\T R\vec y' 
			\ifx\twocols\undefined = \else \cr &&= \fi 
			M'^\mpi (R_1 X)\vec\beta + M'^\mpi (R_1\vec e + wR_3) \\
			\vec\zeta' &&= \tfrac 1 {\sqrt{r}} (R\vec y' - RX'\vec\beta') 
			\ifx\twocols\undefined = \else \cr &&= \fi 
			\tfrac 1 {\sqrt{r}}(I_{r\times r}-M'M'^\mpi) R\vec e'
			\ifx\twocols\undefined = \else \cr &&= \fi 
			\tfrac 1 {\sqrt r} P_{U'^\perp}(R_1\vec e - wR_2\vec\beta + wR_3)
		\end{eqnarray*}with $U'$ denoting $colspan(M')$ and $P_{U'^\perp}$ denoting the projection onto the subspace $U'^\perp$.
		
		Again, we break $\vec e$ into an orthogonal composition: $\vec e = P_U\vec e + P_{U^\perp}\vec e$ with $U = colspan(X)$ (hence $P_{U} = XX^\mpi$) and $U^\perp = colspan(X)^\perp$. Therefore,
		\ifx\twocols\undefined 
		\begin{eqnarray*}
		&\vec\beta' &= M'^\mpi (R_1 X)\vec\beta + M'^\mpi (R_1X X^\mpi\vec e + R_1P_{U^\perp}\vec e + wR_3) \cr
		&&= M'^\mpi(R_1 X)(\vec\beta + X^\mpi \vec e) +M'^\mpi (R_1P_{U^\perp}\vec e + wR_3)  \\
		&\vec\zeta' &= \tfrac 1 {\sqrt r} (I_{r\times r}- M'M'^\mpi)(R_1X X^\mpi\vec e + R_1 P_{U^\perp}\vec e- wR_2\vec\beta + wR_3)\cr
		&& \stackrel{(\ast)}=\tfrac 1 {\sqrt r} (I_{r\times r}- M'M'^\mpi)(R_1X X^\mpi\vec e + R_1 P_{U^\perp}\vec e + (M'- wR_2)\vec\beta + wR_3) \cr
		&& =\tfrac 1 {\sqrt r} (I_{r\times r}- M'M'^\mpi)( R_1X (\vec \beta + X^\mpi\vec e) + R_1 P_{U^\perp}\vec e +  wR_3) 
		\end{eqnarray*}
		\else
		\begin{eqnarray*}
		&\vec\beta' &= M'^\mpi (R_1 X)\vec\beta + M'^\mpi (R_1X X^\mpi\vec e + R_1P_{U^\perp}\vec e + wR_3) \cr
		&&= M'^\mpi(R_1 X)(\vec\beta + X^\mpi \vec e) +M'^\mpi (R_1P_{U^\perp}\vec e + wR_3)  \end{eqnarray*} whereas $\vec\zeta'$ is essentially 
		\begin{align*}
		&\tfrac 1 {\sqrt r} (I_{r\times r}- M'M'^\mpi)(\resizebox*{0.3\textwidth}{!}{$R_1X X^\mpi\vec e + R_1 P_{U^\perp}\vec e- wR_2\vec\beta + wR_3$})\cr
		&\stackrel{(\ast)}=\tfrac 1 {\sqrt r} (I_{r\times r}- M'M'^\mpi)\cdot 
		\cr & ~~~~~~~~~~~~~~(R_1X X^\mpi\vec e + R_1 P_{U^\perp}\vec e + (M'- wR_2)\vec\beta + wR_3) \cr
		&=\tfrac 1 {\sqrt r} (I_{r\times r}- M'M'^\mpi)\cdot \cr & ~~~~~~~~~~~~~~( R_1X (\vec \beta + X^\mpi\vec e) + R_1 P_{U^\perp}\vec e +  wR_3) 
		\end{align*}
		\fi 
		where equality $(\ast)$ holds because $(I- M'M'^\mpi)M'\vec v = \vec 0$ for any  $\vec v$.
		
		We now aim to describe the distribution of $R$ given that we know $X'$ and $M' = RX'$. Since 
		\begin{align*}
		&M' = R_1 X + wR_2 + 0\cdot R_3 = R_1 X (X^\mpi X) +w R_2 
		\ifx\twocols\undefined = \else \cr &~~~~=\fi 
		(R_1 P_U) X + w R_2
		\end{align*} then $M'$ is independent of $R_3$ and independent of $R_1P_{U^\perp}$. Therefore, given $X$ and $M'$ the induced distribution over $R_3$ remains $R_3\sim\gaussian(\vec 0_r, I_{r\times r})$, and similarly, given $X$ and $M'$ we have $R_1P_{U^\perp} \sim \gaussian(0_{r\times n}, I_{r\times r}, P_{U^\perp})$ (rows remain independent from one another, and each row is distributed like a spherical Gaussian in $colspan(X)^\perp$). And so, we have that $R_1X = R_1 P_U X = M' - wR_2$, which in turn implies:
		\ifx\twocols\undefined 
		\begin{align*}
		&& R_1X &\sim \gaussian\left( M', ~~ I_{r\times r}, ~~w^2\cdot I_{p\times p}\right) \cr
		&&\Rightarrow R_1X(\vec\beta + X^\mpi\vec e) &\sim \gaussian\left(M'\vec\beta + M'X^\mpi\vec e, ~~ w^2\|\vec\beta + X^{\mpi}\vec e\|^2 I_{r\times r}\right)\cr
		&&\Rightarrow M'^\mpi R_1X(\vec\beta + X^{\mpi}\vec e) &\sim \gaussian\left(\vec\beta+ X^\mpi\vec e, ~~ w^2\|\vec\beta+X^\mpi\vec e\|^2  (M'^\T M)^{-1}\right)\cr
		&& & = \|\vec\beta + X^\mpi \vec e\|\cdot \gaussian( \vec u, w^2 (M'^\T M)^{-1})
		\end{align*}
		where $\vec u$ denotes a unit-length vector in the direction of $\vec\beta + X^\mpi \vec e$.
		\else 
				\[R_1X \sim \gaussian\left( M', ~~ I_{r\times r}, ~~w^2\cdot I_{p\times p}\right)\] multiplying this random matrix with a vector, we get
				\[ R_1X(\vec\beta + X^\mpi\vec e)\sim \gaussian\left(\resizebox*{0.3\textwidth}{!}{$M'\vec\beta + M'X^\mpi\vec e, ~~ w^2\|\vec\beta + X^{\mpi}\vec e\|^2 I_{r\times r}$}\right)\] and multiplying this random vector with a matrix we get
				\[M'^\mpi R_1X(\vec\beta + X^{\mpi}\vec e) \sim \gaussian\left(\resizebox*{0.26\textwidth}{!}{$\vec\beta+ X^\mpi\vec e, ~~ w^2\|\vec\beta+X^\mpi\vec e\|^2  (M'^\T M)^{-1}$}\right)\] I.e., 
				\[M'^\mpi R_1X(\vec\beta + X^{\mpi}\vec e) \sim \|\vec\beta + X^\mpi \vec e\|\cdot \gaussian( \vec u, w^2 (M'^\T M)^{-1})\]
				where $\vec u$ denotes a unit-length vector in the direction of $\vec\beta + X^\mpi \vec e$.
		\fi  
		
		Similar to before we have
		\begin{align*}
		& RP_{U^\perp} \sim \gaussian(0_{r\times n}, I_{r\times r}, P_{U^\perp}) 
		\cr & ~~~\Rightarrow M'^\mpi (R P_{U^\perp} \vec e) \sim \gaussian (\vec 0_d, ~ \|P_{U^\perp} e\|^2 (M'^\T M')^{-1}) 
		\cr & wR_3 \sim \gaussian(\vec 0_r, w^2 I_{r\times r}) \cr &~~~\Rightarrow M'^\mpi(wR_3) \sim \gaussian( \vec 0_d, w^2 (M'^\mpi M')^{-1})
		\end{align*}
		Therefore, the distribution of  $\vec\beta'$, which is the sum of the $3$ independent Gaussians, is as required.
		
		Also, $\vec\zeta' = \tfrac 1 {\sqrt r} P_{U'^\perp} \left( R_1 X(\vec\beta + X^\mpi\vec e) + R_1 P_{U^\perp}\vec e + wR_3 \right)$ is the sum of $3$ independent Gaussians, which implies its distribution is
		\ifx\twocols\undefined 
		\[ \gaussian\left(\tfrac 1 {\sqrt r}P_{U'^\perp}M'(\vec\beta+X^\mpi\vec e), ~~ \tfrac 1 r (w^2(\|\vec\beta+X^\mpi\vec e\|^2 +1) + \|P_{U^\perp}\vec e\|^2)P_{U'^\perp} \right) \] 
		\else 
		\begin{align*}
		&\gaussian\Big(\tfrac 1 {\sqrt r}P_{U'^\perp}M'(\vec\beta+X^\mpi\vec e), 
		\cr & ~~~~~~~~~~~~~ \tfrac 1 r (w^2(\|\vec\beta+X^\mpi\vec e\|^2 +1) + \|P_{U^\perp}\vec e\|^2)P_{U'^\perp} \Big)
		\end{align*}
		\fi I.e., $\gaussian\left(\vec 0_r, ~~ \tfrac 1 r (w^2(\|\vec\beta+X^\mpi\vec e\|^2 +1) + \|P_{U^\perp}\vec e\|^2)P_{U'^\perp} \right)$ as $P_{U'^\perp} M'= 0_{r\times r}$.
		
		Finally, observe that $\vec\beta'$ and $\vec \zeta'$ are independent as the former depends on the projection of the spherical Gaussian $R_1 X(\beta + X^\mpi\vec e) + R_1 P_{U^\perp}\vec e + wR_3$ on $U'$, and the latter depends on the projection of the same multivariate Gaussian on $U'^\perp$.
	\end{proof}
	
	Observe that Claim~\ref{clm:dist_beta'_zeta'} assumes $\vec e$ is given. This may seem somewhat strange, since without assuming anything about $\vec e$ there can be many combinations of $\vec\beta$ and $\vec e$ for which $\vec y = X\vec\beta + \vec e$. However, we always have that $\vec\beta + X^\mpi \vec e = X^\mpi\vec y = \hvec\beta$. Similarly, it is always the case the $P_{U^\perp}\vec e = (I-XX^\mpi) \vec y = \vec\zeta$. (Recall OLS definitions of $\hvec\beta$ and $\vec\zeta$ in Equation~\eqref{eq:hat_beta} and~\eqref{eq:zeta}.) Therefore, the distribution of $\vec\beta'$ and $\vec\zeta'$ is unique (once $\vec y$ is set):
	\begin{align*}
	&& \vec\beta' &\sim \gaussian\left( \hvec\beta ,    (w^2(\|\hvec\beta\|^2+1)  + \|\vec\zeta\|^2  ) (M'^\T M')^{-1} \right)\cr
	&&\vec\zeta' &\sim \gaussian\left( \vec 0_r, \frac {w^2(\|\hvec\beta\|^2+1)  + \|\vec\zeta\|^2 }r (I_{r\times r}-M'M'^\mpi) \right)
	\end{align*}
	And so for a given dataset $[X;\vec y]$ we have that $\vec\beta'$ serves as an approximation for $\hvec\beta$.
	
	An immediate corollary of Claim~\ref{clm:dist_beta'_zeta'} is that for any fixed $\vec e$ it holds that the quantity
	$t'(\beta_j) = \frac {\beta'_j - (\beta_j + (X^\mpi\vec e)_j)} {   \|\vec \zeta'\|\sqrt{\tfrac r {r-p}\cdot (M'^\T M')^{-1}_{j,j}}  }=\frac {\beta'_j - \hat\beta_j} {   \|\vec \zeta'\|\sqrt{\tfrac r {r-p}\cdot (M'^\T M')^{-1}_{j,j}}  }$ is distributed like a $T_{r-p}$-distribution. Therefore, the following theorem  follows immediately.
	\begin{theorem}
		\label{thm:approximating_hatbeta}
		Fix $X\in \mathbb{R}^{n\times p}$ and $\vec y\in \mathbb{R}$. Define $\hvec\beta = X^\mpi \vec y$ and $\zeta = (I- XX^\mpi)\vec y$. Let $RX'$ and $R\vec y'$ denote the result of applying Algorithm~\ref{alg:private_JL} to the matrix $A = [X;\vec y]$ when the algorithm appends the data with a $w\cdot I$ matrix. Fix a coordinate $j$ and any $\alpha\in (0,1/2)$. When computing $\vec\beta'$ and $\vec\zeta'$ as in Equations~\eqref{eq:beta'} it and~\eqref{eq:zeta'}, we have that w.p. $\geq 1-\alpha$ it holds that
		\[ \hat\beta_j \in \left( \beta'_j \pm c'_\alpha \|\vec \zeta'\|\sqrt{\tfrac r {r-p}\cdot (M'^\T M')^{-1}_{j,j}}  \right)\] where $c'_\alpha$ denotes the number such that $(-c'_\alpha,c'_\alpha)$ contains $1-\alpha$ mass of the $T_{r-p}$-distribution.
	\end{theorem}
	
	Note that Theorem~\ref{thm:approximating_hatbeta}, much like the rest of the discussion in this Section, builds on $\vec y$ being fixed, which means $\beta_j'$ serves as an approximation for $\hat\beta_j$. Yet our goal is to argue about similarity (or proximity) between $\beta_j'$ and $\beta_j$. To that end, we combine the standard OLS confidence interval~--- which says that w.p. $\geq 1-\alpha$ over the randomness of picking $\vec e$ in the homoscedastic model we have $|\beta_j-\hat\beta_j| \leq c_\alpha \|\vec\zeta\|\sqrt{\tfrac{(X^\T X)^{-1}_{j,j}}{n-p}}$~--- with the confidence interval of Theorem~\ref{thm:approximating_hatbeta} above, and deduce that
	\ifx\twocols\undefined 
	\begin{equation}
	\Pr\left[ |\beta'_j - \beta_j| = O\left( c_\alpha \frac{\|\vec\zeta\|}{\sqrt{n-p}} \sqrt{(X^\T X)^{-1}_{j,j}} +  c'_\alpha \frac {\|\vec\zeta'\|} {\sqrt{r-p}} \sqrt{r(M'^\T M')^{-1}_{j,j}}\right)\right] \geq 1-\alpha \label{eq:dist_beta'j_betaj}
	\end{equation}
	\else 
	w.p. $\geq 1-\alpha$ we have that $|\beta'_j - \beta_j|$ is at most
	\begin{equation}
	O\left( c_\alpha \frac{\|\vec\zeta\|\sqrt{(X^\T X)^{-1}_{j,j}}}{\sqrt{n-p}}  +  c'_\alpha \frac {\|\vec\zeta'\|\sqrt{r(M'^\T M')^{-1}_{j,j}}} {\sqrt{r-p}} \right)\label{eq:dist_beta'j_betaj}
	\end{equation}
	\fi 
	\footnote{Observe that w.p. $\geq 1- \alpha$ over the randomness of $\vec e$ we have that $|\beta_j-\hat\beta_j| \leq c_\alpha \|\vec\zeta\|\sqrt{\tfrac{(X^\T X)^{-1}_{j,j}}{n-p}}$, and w.p. $\geq1-\alpha$ over the randomness of $R$ we have that $|\beta_j'-\hat\beta_j|\leq c'_\alpha \|\vec \zeta'\|\sqrt{\tfrac r {r-p}\cdot (M'^\T M')^{-1}_{j,j}}$. So technically, to give a $(1-\alpha)$-confidence interval around $\beta_j'$ that contains $\beta_j$ w.p. $\geq 1-\alpha$, we need to use $c_{\alpha/2}$ and $c'_{\alpha/2}$ instead of $c_\alpha$ and $c'_\alpha$ resp. To avoid overburdening the reader with what we already see as too many parameters, we switch to asymptotic notation.}And so, in the next section, our goal is to give conditions under which the interval of Equation~\eqref{eq:dist_beta'j_betaj} isn't much larger in comparison to the interval length of $c'_\alpha \tfrac {\|\vec\zeta'\|} {\sqrt{r-p}} \sqrt{r(M'^\T M')^{-1}_{j,j}}$ we get from Theorem~\ref{thm:approximating_hatbeta}; and more importantly~--- conditions that make the interval of Theorem~\ref{thm:approximating_hatbeta} useful and not too large. (Note, in expectation $\tfrac {\|\vec\zeta'\|} {\sqrt{r-p}}$ is about $\sqrt{(w^2 + w^2\|\hvec\beta\|^2 + \|\vec\zeta\|^2)/r}$. So, for example, in situations where $\|\hvec\beta\|$ is very large, this interval isn't likely to inform us as to the sign of $\beta_j$.) 
	
	\myparagraph{Motivating Example.} A good motivating example for the discussion in the following section is when $[X;\vec y]$ is a strict submatrix of the dataset $A$. That is, our data contains many variables for each entry (i.e., the dimensionality $d$ of each entry is large), yet our regression is made only over a modest subset of variables out of the $d$. In this case, the least singular value of $A$ might be too small, causing the algorithm to alter $A$; however, $\svmin(X^\T X)$ could be sufficiently large so that had we run Algorithm~\ref{alg:private_JL} only on $[X;\vec y]$ we would not alter the input. (Indeed, a differentially private way for finding a subset of the variables that induce a submatrix with high $\svmin$ is an interesting open question, partially answered~--- for a single regression~--- in the work of Thakurta and Smith~\mycite{ThakurtaS13}.) Indeed, the conditions we specify in the following section depend on $\svmin(\tfrac 1 n X^\T X)$, which, for a zero-mean data, the minimal variance of the data in any direction. For this motivating example, indeed such variance isn't necessarily small.

	\subsection{Conditions for Deriving a Confidence Interval for Ridge Regression}
	\label{subsec:analyzing_beta'}
	
	Looking at the interval specified in Equation~\eqref{eq:dist_beta'j_betaj}, we now give an upper bound on the the random quantities in this interval: $\|\vec\zeta\|, \|\vec\zeta'\|$, and $(M'^\T M')^{-1}_{j,j}$. First, we give bound that are dependent on the randomness in $R$ (i.e., we continue to view $\vec e$ as fixed). 
	\begin{proposition}
		\label{pro:zeta'_M'TM'}
		For any $\nu\in(0,1/2)$, if we have $r = p + \Omega(\ln(1/\nu))$ then with probability $\geq 1-\nu$ over the randomness of $R$ we have $(r-p)(M'^\T M)^{-1}_{j,j} = \Theta\left( (w^2 I_{p\times p} + X^\T X)^{-1}_{j,j}\right)$ and $\tfrac {\|\vec\zeta'\|^2}{r-p} = \Theta( \tfrac{w^2 + w^2\|\hvec\beta\|^2 + \|\vec\zeta\|^2}r)$.
	\end{proposition}
	\begin{proof}
		The former bound follows from known results on the Johnson-Lindenstrauss transform (as were shown in the proof of Claim~\ref{clm:comparability_sigmaXX_lMM}). The latter bound follows from standard concentration bounds of the $\chi^2$-distribution.
	\end{proof}
	
	Plugging in the result of Proposition~\ref{pro:zeta'_M'TM'} to Equation~\eqref{eq:dist_beta'j_betaj} we get that w.p. $\geq1-\nu$ 
	\ifx\twocols\undefined 
	\begin{equation}
	|\beta'_j - \beta_j| = O\left( c_\alpha \frac{\|\vec\zeta\|}{\sqrt{n-p}} \sqrt{(X^\T X)^{-1}_{j,j}} +  c'_\alpha \sqrt{\frac {w^2 + w^2\|\hvec\beta\|^2 + \|\vec\zeta\|^2}{r-p} } \sqrt{(w^2 I_{p\times p} + X^\T X)^{-1}_{j,j}}\right) \label{eq:dist_beta'j_betaj_revised}
	\end{equation} 
	\else 
	the difference $|\beta'_j - \beta_j|$ is at most
	\begin{align}
	&O\Big( c_\alpha \frac{\|\vec\zeta\|}{\sqrt{n-p}} \sqrt{(X^\T X)^{-1}_{j,j}} 
	\cr & ~~~~~~+  c'_\alpha \sqrt{\frac {w^2 + w^2\|\hvec\beta\|^2 + \|\vec\zeta\|^2}{r-p} } \sqrt{(w^2 I_{p\times p} + X^\T X)^{-1}_{j,j}}\Big)\cr \label{eq:dist_beta'j_betaj_revised}
	\end{align} 
	\fi 
	We will also use the following proposition.
	\begin{proposition}
		\label{pro:XTX}
		\[ (X^\T X)^{-1}_{j,j} \leq \left(1 + \frac{w^2}{\svmin(X^\T X)}\right) (w^2 I_{p\times p} + X^\T X)^{-1}_{j,j} \]
	\end{proposition}
	\begin{proof}
		We have that
		\ifx\twocols\undefined 
		\begin{eqnarray*}
			& (X^\T X)^{-1} &= (X^\T X)^{-1}(X^\T X + w^2 I_{p\times p})(X^\T X + w^2 I_{p\times p})^{-1} \cr
			&& = (X^\T X + w^2 I_{p\times p})^{-1} + w^2(X^\T X)^{-1}(X^\T X + w^2 I_{p\times p})^{-1}\cr
			&& = (I_{p\times p} + w^2 (X^\T X)^{-1})(X^\T X + w^2 I_{p\times p})^{-1}\cr
			&& = (X^\T X + w^2 I_{p\times p})^{-1/2} (I_{p\times p} + w^2 (X^\T X)^{-1}) (X^\T X + w^2 I_{p\times p})^{-1/2}
		\end{eqnarray*}
		\else 
		\begin{align*}
			& (X^\T X)^{-1} 
			\cr &= (X^\T X)^{-1}(X^\T X + w^2 I_{p\times p})(X^\T X + w^2 I_{p\times p})^{-1} \cr
			& = (X^\T X + w^2 I_{p\times p})^{-1} + w^2(X^\T X)^{-1}(X^\T X + w^2 I_{p\times p})^{-1}\cr
			& = (I_{p\times p} + w^2 (X^\T X)^{-1})(X^\T X + w^2 I_{p\times p})^{-1}\cr
			& = (X^\T X + w^2 I_{p\times p})^{-1/2}\cdot \cr &~~~~~~~~~~~~~~ (I_{p\times p} + w^2 (X^\T X)^{-1}) \cdot \cr &~~~~~~~~~~~~~~~~~~~~~~~~~~~(X^\T X + w^2 I_{p\times p})^{-1/2}
		\end{align*}
		\fi 
		where the latter holds because $(I_{p\times p} + w^2 (X^\T X)^{-1})$ and $(X^\T X + w^2 I_{p\times p})^{-1}$ are diagonalizable by the same matrix $V$ (the same matrix for which $(X^\T X) = V S^{-1} V^\T$).
		Since we have $\|I_{p\times p} + w^2 (X^\T X)^{-1} \| = 1 + \tfrac{w^2}{\svmin^{2}(X)}$, it is clear that $(I_{p\times p} + w^2 (X^\T X)^{-1}) \preceq (1 + \tfrac{w^2}{\svmin^{2}(X)}) I_{p\times p}$. 
		We deduce that $(X^\T X)^{-1}_{j,j} = \vec e_j^\T (X^\T X)^{-1} \vec e_j \leq (1 + \tfrac{w^2}{\svmin^{2}(X)}) (X^\T X+w^2I_{p\times p})^{-1}_{j,j}$.
	\end{proof}
	
	Based on Proposition~\ref{pro:XTX} we get from Equation~\eqref{eq:dist_beta'j_betaj_revised} that
	\ifx\twocols\undefined 
	\begin{equation}
	|\beta'_j - \beta_j| = O(\left( c_\alpha \sqrt{\frac{\|\vec\zeta\|^2(1+\tfrac {w^2}{\svmin(X^\T X)})}{n-p}} +  c'_\alpha \sqrt{\frac {w^2 + w^2\|\hvec\beta\|^2 + \|\vec\zeta\|^2}{r-p} } \right)\sqrt{(w^2 I_{p\times p} + X^\T X)^{-1}_{j,j}}) \label{eq:dist_beta'j_betaj_revised2}
	\end{equation}
	\else
		$|\beta'_j - \beta_j|$ is at most \begin{align}
		&O(\Big( c_\alpha \sqrt{\frac{\|\vec\zeta\|^2(1+\tfrac {w^2}{\svmin(X^\T X)})}{n-p}} + 
		\cr &~~~ c'_\alpha \sqrt{\frac {w^2 + w^2\|\hvec\beta\|^2 + \|\vec\zeta\|^2}{r-p} } \Big)\sqrt{(w^2 I_{p\times p} + X^\T X)^{-1}_{j,j}})\cr  \label{eq:dist_beta'j_betaj_revised2}
		\end{align}
		\fi 
	And so, if it happens to be the case that exists some small $\eta>0$ for which $\hvec\beta, \vec\zeta$ and $w^2$ satisfy
	\begin{equation}
	\frac{\|\vec\zeta\|^2(1+\tfrac {w^2}{\svmin(X^\T X)})}{n-p} \leq \eta^2\left( \frac {w^2 + w^2\|\hvec\beta\|^2 + \|\vec\zeta\|^2}{r-p} \right)
	\label{eq:constraints_on_zeta_betahat}
	\end{equation} then we have that 
	$\Pr[\beta_j \in \left(\beta_j' \pm O((1+\eta)\cdot c'_\alpha \|\vec \zeta'\|\sqrt{\tfrac r {r-p}\cdot (M'^\T M')^{-1}_{j,j}})\right)] \geq 1-\alpha$.\footnote{We assume $n\geq r$ so $c_\alpha<c'_\alpha$ as the $T_{n-p}$-distribution is closer to a normal Gaussian than the $T_{r-p}$-distribution.} Moreover, if in this case $|\beta_j| > c'_\alpha(1+\eta)\sqrt{\frac {w^2 + w^2\|\hvec\beta\|^2 + \|\vec\zeta\|^2}{r-p} }\sqrt{(w^2 I_{p\times p} + X^\T X)^{-1}_{j,j}}$ then $\Pr[sign(\beta_j') = sign(\beta_j)]\geq 1-\alpha$. This is precisely what Claims~\ref{clm:sufficient_condition_for_interval_on_beta'} and~\ref{clm:sufficient_conditions_for_finding_right_sign_beta'} below do.
	
	\begin{claim}
		\label{clm:sufficient_condition_for_interval_on_beta'}
		If there exists $\eta>0$ s.t. $n-p \geq \tfrac{2}{\eta^2}(r-p)$ and $n^2 = \Omega\left(r^{3/2}\cdot\frac{B^2\ln(1/\delta)}\epsilon \cdot  \frac 1 {\eta^2\svmin(\tfrac 1 n X^\T X)}\right)$, then $\Pr[\beta_j \in \left(\beta_j' \pm O((1+\eta)\cdot c'_\alpha \|\vec \zeta'\|\sqrt{\tfrac r {r-p}\cdot (M'^\T M')^{-1}_{j,j}})\right)] \geq 1-\alpha$.
	\end{claim}
	\begin{proof}
		Based on the above discussion, it is enough to argue that under the conditions of the claim, the constraint of Equation~\eqref{eq:constraints_on_zeta_betahat} holds. Since we require $\tfrac{\eta^2}2 \geq \tfrac{r-p}{n-p}$ then it is evident that $\tfrac{\|\vec\zeta\|^2}{n-p} \leq \tfrac{\eta^2\|\vec\zeta\|^2}{2(r-p)}$. So we now show  that $\tfrac{\|\vec\zeta\|^2}{n-p} \cdot \tfrac {w^2}{\svmin(X^\T X)}\leq \tfrac{\eta^2\|\vec\zeta\|^2}{2(r-p)}$ under the conditions of the claim, and this will show the required.
		All that is left is some algebraic manipulations. It suffices to have:
		\begin{align*}
		\tfrac{\eta^2}{2} \cdot \tfrac{n-p}{r-p} \svmin(X^\T X) & \geq \tfrac{\eta^2}{2} \cdot \tfrac{n^2}{r} \svmin(\tfrac 1 n X^\T X) 
		\ifx\twocols\undefined \geq \else \cr &\geq \fi 
		\frac {32B^2\sqrt{r}\ln(8/\delta)}\epsilon \geq w^2 
		\end{align*} which holds for $n^2 \geq r^{3/2}\cdot\frac {64B^2\ln(1/\delta)} {\epsilon \eta^2}\svmin(\tfrac 1 n X^\T X)^{-1}$, as we assume to hold.
	\end{proof}
	\cut{
		Therefore, all that is left to show is that w.p. $\geq 1-\nu$ it holds that $\tfrac{\|\vec\zeta\|^2}{n-p} \cdot \tfrac 1{\svmin(X^\T X)}\leq \tfrac{\eta^2 (1+\|\hvec\beta\|^2)}{r-p}$.
		
		To see that, we now invoke Claim~\ref{clm:sizes_hatbeta_zeta}, which states that w.p. $\geq 1-\nu$ we have $\|\vec\beta-\hvec\beta\|^2 = O(\sigma^2\|X^\mpi\|_F^2\ln(\tfrac p \nu))$ and $\|\vec\zeta\|^2= \Theta( \sigma^2 (n-p) )$. Therefore, for $n = p + \Omega(\ln(1/\nu))$ we need to have $\tfrac{2\sigma^2}{\svmin(X^\T X)} \leq \tfrac {\eta^2 (1+\|\vec\beta\|^2)}{2(r-p)}$. Due to the lower bound on $\eta^2$ this indeed holds.

	}
	
	\begin{claim}
		\label{clm:sufficient_conditions_for_finding_right_sign_beta'}
		Fix $\nu \in (0,\tfrac1 2)$. If (i) $n= p + \Omega(\ln(1/\nu))$, (ii) $\|\vec\beta\|^2  = \Omega(\sigma^2 \|X^\mpi\|^2_F \ln(\tfrac p \nu))$ and (iii) $r-p = \Omega\left( \tfrac{(c'_\alpha)^2(1+\eta)^2}{\beta_j^2} \left(1+\|\vec\beta\|^2 +  \frac {\sigma^2}{\svmin(\tfrac 1 n X^\T X)} \right)\right)$, then in the homoscedastic model, with probability $\geq 1-\nu-\alpha$ we have that $sign(\beta_j) = sign(\beta_j')$.
	\end{claim}
	\begin{proof}
		Based on the above discussion, we aim to show that in the homoscedastic model (where each coordinate $e_i\sim\gaussian(0,\sigma^2)$ independently) w.p. $\geq 1- \nu$ it holds that 
		\ifx\twocols\undefined 
		\[|\beta_j| > c'_\alpha(1+\eta)\sqrt{\frac {w^2 + w^2\|\hvec\beta\|^2 + \|\vec\zeta\|^2}{r-p} }\sqrt{(w^2 I_{p\times p} + X^\T X)^{-1}_{j,j}}\] 
		\else
		the magnitude of $\beta_j$ is greater than
		\[ c'_\alpha(1+\eta)\sqrt{\frac {w^2 + w^2\|\hvec\beta\|^2 + \|\vec\zeta\|^2}{r-p} }\sqrt{(w^2 I_{p\times p} + X^\T X)^{-1}_{j,j}} \] 
		\fi
		To show this, we invoke Claim~\ref{clm:sizes_hatbeta_zeta} to argue that w.p. $\geq 1-\nu$ we have (i) $\|\vec\zeta\|^2 \leq 2\sigma^2(n-p)$ (since $n = p + \Omega(\ln(1/\nu))$), and (ii) $\|\hvec\beta\|^2 \leq 2\|\vec\beta\|^2$ (since $\|\vec\beta-\hvec\beta\|^2 \leq \sigma^2 \|X^\mpi\|^2_F \ln(\tfrac p \nu)$ whereas $\|\vec\beta\|^2  = \Omega(\sigma^2 \|X^\mpi\|^2_F \ln(\tfrac p \nu))$). We also use the fact that $(w^2 I_{p\times p}+ X^\T X)^{-1}_{j,j} \leq (w^2 + \svmin^{-1}(X^\T X))$, and then deduce that 
		\begin{align*}
		&(1+\eta)c'_\alpha\sqrt{\frac {w^2 + w^2\|\hvec\beta\|^2 + \|\vec\zeta\|^2}{r-p} }\sqrt{(w^2 I_{p\times p} + X^\T X)^{-1}_{j,j}} 
		\cr&~~~~\leq \frac{(1+\eta)c'_\alpha}{\sqrt{r-p}} \sqrt{2\frac{w^2(1+\|\vec\beta\|^2) + \sigma^2(n-p) }{w^2 + \svmin(X^\T X)}} \ifx\twocols\undefined \leq \else \cr &~~~~\leq \fi 
		\frac{(1+\eta)c'_\alpha}{\sqrt{r-p}} \sqrt{2(1+\|\vec\beta\|^2)+ \frac{2\sigma^2(n-p)}{\svmin(X^\T X)} } \leq |\beta_j|   
		\end{align*}
		due to our requirement on $r-p$.
	\end{proof}
	
	Observe, out of the $3$ conditions specified in Claim~\ref{clm:sufficient_conditions_for_finding_right_sign_beta'}, condition (i) merely guarantees that the sample is large enough to argue that estimations are close to their expect value; and condition (ii) is there merely to guarantee that $\|\hvec\beta\|\approx \|\vec\beta\|$. It is condition (iii) which is non-trivial to hold, especially together with the conditions of Claim~\ref{clm:sufficient_condition_for_interval_on_beta'} that pose other constraints in regards to $r$, $n$, $\eta$ and the various other parameters in play. It is interesting to compare the requirements on $r$ to the lower bound we get in Theorem~\ref{thm:JL_simple_case_rejecting_null_hypothesis}~--- especially the latter bound. The two bounds are strikingly similar, with the exception that here we also require $r-p$ to be greater than $\frac{1+\|\vec\beta\|^2}{\beta_j^2}$. This is part of the unfortunate effect of altering the matrix $A$: we cannot give confidence bounds only for the coordinates $j$ for which $\beta_j^2$ is very small \emph{relative to $\|\vec\beta\|^2$}.
	
	In summary, we require to have $n = p  + \Omega(\ln(1/\nu))$ and that $X$ contains enough sample points to have $\|\hvec\beta\|$ comparable to $\|\vec\beta\|$, and then set $r$ and $\eta$ such that (it is convenient to think of $\eta$ as a small constant, say, $\eta = 0.1$) 
	\begin{itemize}
		\item $r - p = O(\eta^2(n-p))$ (which implies $r = O(n)$)
		\item $r = O(\left( \eta^2 \frac{\epsilon n^2}{B^2 \ln(1/\delta)}\svmin(\tfrac 1 n X^\T X)\right)^{\tfrac 2 3})$ 
		\item $r - p = \Omega( \frac{1+\|\vec\beta\|^2}{\beta_j^2} + \frac{\sigma^2}{\beta_j^2} \cdot {\svmin^{-1}(\tfrac 1 n X^\T X)})$
	\end{itemize}
	to have that the $(1-\alpha)$-confidence interval around $\beta_j'$ does not intersect the origin. Once again, we comment that these conditions are sufficient but not necessary, and furthermore~--- even with these conditions holding~--- we do not make any claims of optimality of our confidence bound. That is because from Proposition~\ref{pro:XTX} onwards our discussion uses upper bounds that do not have corresponding lower bounds, to the best of our knowledge.

\AGsection

\section{Experiment: Additional Figures}
\label{apx_sec:experiment}

To complete our discussion about the experiments we have conducted, we attach here additional figures, plotting both the $t$-value approximations we get from both algorithms, and the ``high-level decision'' of whether correctly reject or not-reject the null hypothesis (and with what sign). First, we show the distribution of the $t$-value approximation for coordinates that should be rejected, in Figure~\ref{apx_fig:t-values-reject}, and then the decision of whether to reject or not based on this $t$-value --- and whether it was right, conservative (we didn't reject while we needed to) or wrong (we rejected with the wrong sign, or rejected when we shouldn't have rejected) in Figure~\ref{apx_fig:decision-reject}. As one can see, Algorithm~\ref{alg:private_JL} has far lower $t$-values (as expected) and therefore is much more conservative. In fact, it tends to not-reject coordinate $1$ of the real-data even on the largest value of $n$ (Figure~\ref{apx_subfig:decision-real-coord1}).

However, because Algorithm~\ref{alg:private_JL} also has much smaller variance, it also does not reject when it ought to not-reject, whereas Algorithm~\ref{alg:AG} erroneiously rejects the null-hypotheses. This can be seen in Figures~\ref{apx_fig:t-values-not-reject} and~\ref{apx_fig:decision-not-reject}.

\begin{figure}[h]
	\begin{center}
		\begin{subfigure}[h]{0.5\textwidth}
			\includegraphics[scale=\SCALE]{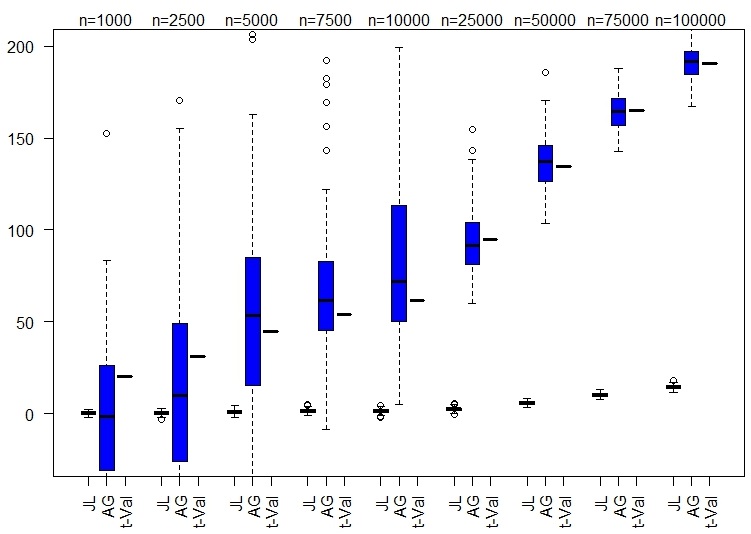}
			\caption{\label{apx_subfig:values-synth-coord1} Synthetic data, coordinate $\beta_1=0.5$}
		\end{subfigure}
		\begin{subfigure}[h]{0.5\textwidth}
			\includegraphics[scale=\SCALE]{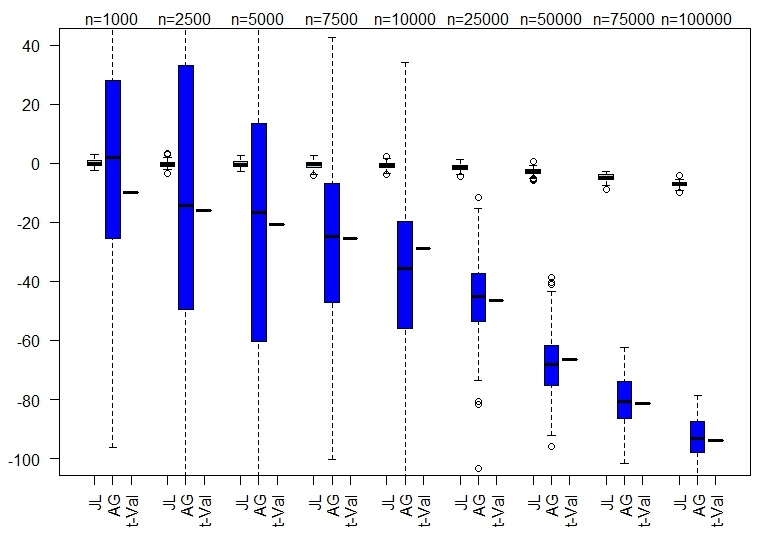}
			\caption{\label{apx_subfig:values-synth-coord2} Synthetic data, coordinate $\beta_2=-0.25$}
		\end{subfigure}
		\begin{subfigure}[h]{0.5\textwidth}
			\includegraphics[scale=\SCALE]{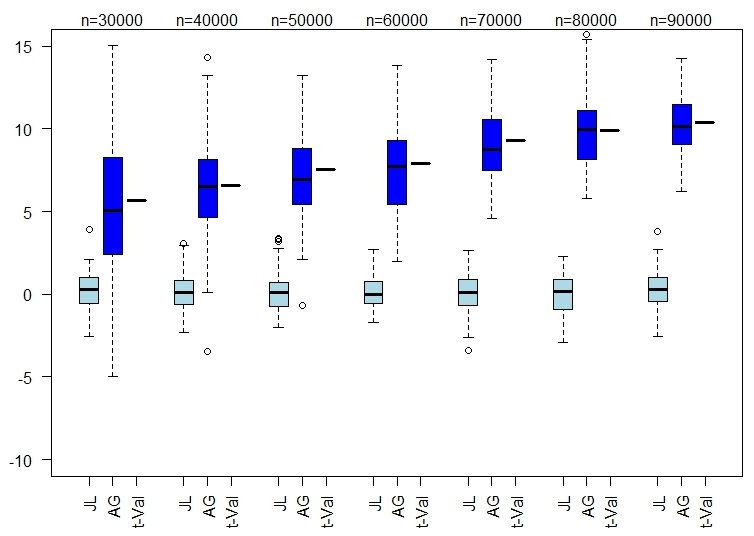}
			\caption{\label{apx_subfig:values-real-coord1} real-life data, coordinate $\beta_1=14.07$}
		\end{subfigure}
	\end{center}
	\caption{\label{apx_fig:t-values-reject} The distribution of the $t$-value approximations from selected experiments on synthetic and real-life data where the null hypothesis should be rejected}
\end{figure}

\begin{figure}[h]
	\begin{center}
		\begin{subfigure}[h]{0.5\textwidth}
			\includegraphics[scale=\SCALE]{decision-synth-coord1.jpeg}
			\caption{\label{apx_subfig:decision-synth-coord1} Synthetic data, coordinate $\beta_1=0.5$}
		\end{subfigure}
		\begin{subfigure}[h]{0.5\textwidth}
			\includegraphics[scale=\SCALE]{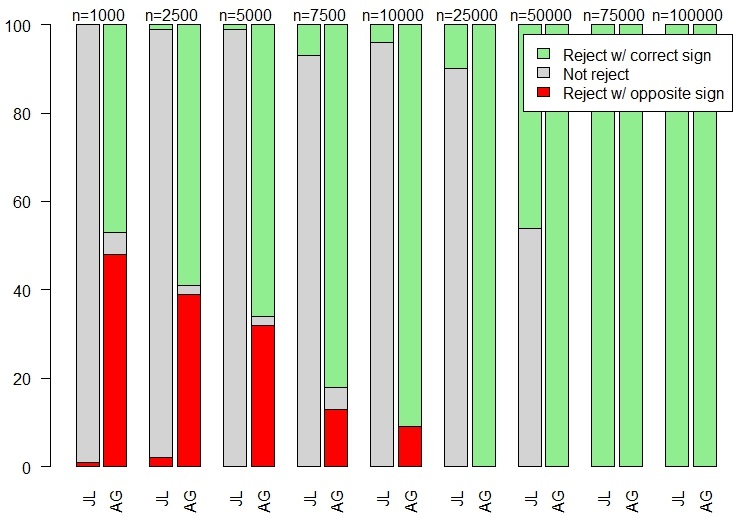}
			\caption{\label{apx_subfig:decision-synth-coord2} Synthetic data, coordinate $\beta_2=-0.25$}
		\end{subfigure}
		\begin{subfigure}[h]{0.5\textwidth}
			\includegraphics[scale=\SCALE]{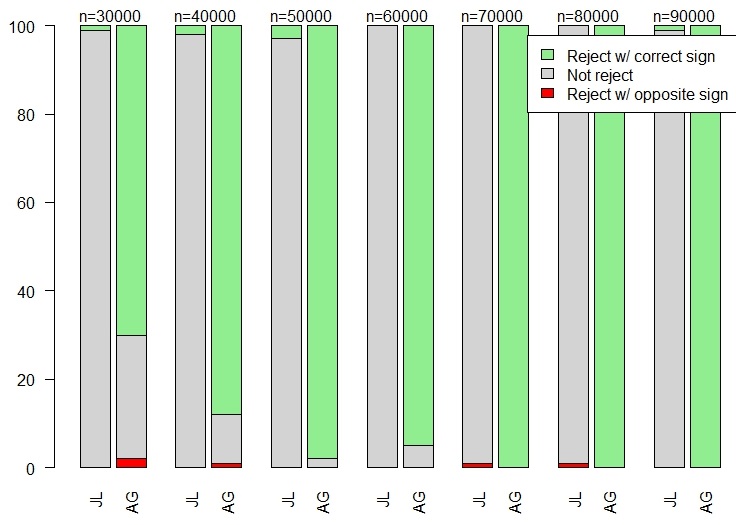}
			\caption{\label{apx_subfig:decision-real-coord1} real-life data, coordinate $\beta_1=14.07$}
		\end{subfigure}
	\end{center}
	\caption{\label{apx_fig:decision-reject} The correctness of our decision to reject the null-hypothesis based on the approximated $t$-value where the null hypothesis should be rejected}
\end{figure}

\begin{figure}[h]
	\begin{center}
		\begin{subfigure}[h]{0.5\textwidth}
			\includegraphics[scale=\SCALE]{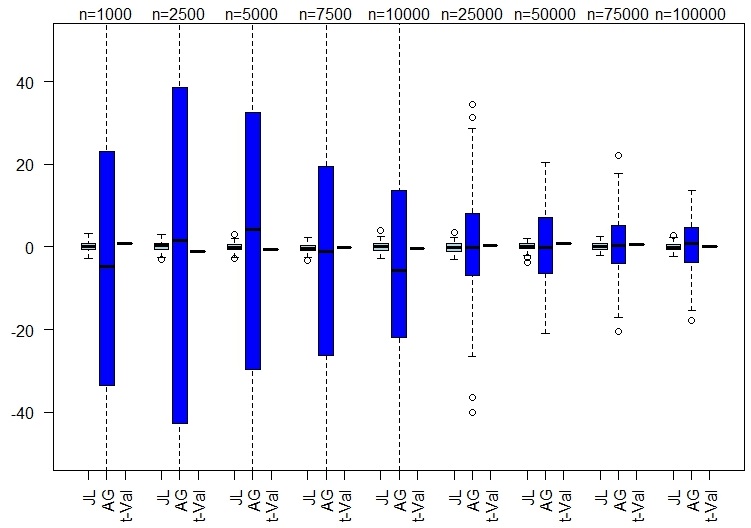}
			\caption{\label{apx_subfig:values-synth-coord3} Synthetic data, coordinate $\beta_3=0$}
		\end{subfigure}
		\begin{subfigure}[h]{0.5\textwidth}
			\includegraphics[scale=\SCALE]{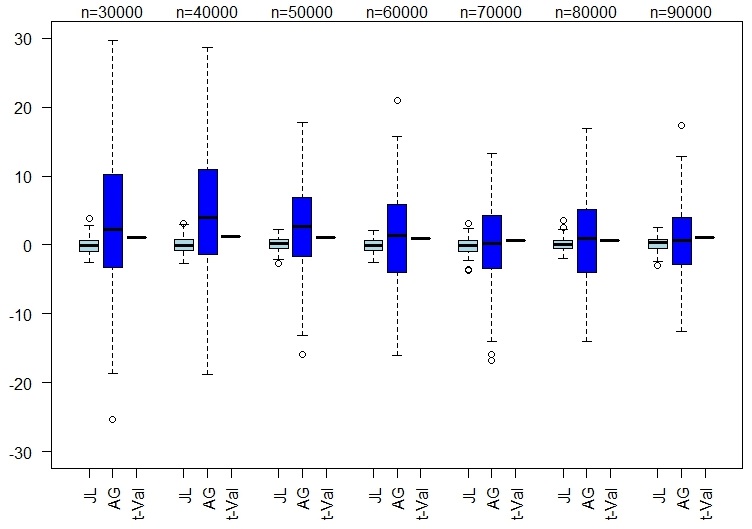}
			\caption{\label{apx_subfig:values-real-coord2} Real-life data, coordinate $\beta_2=0.57$}
		\end{subfigure}
	\end{center}
	\caption{\label{apx_fig:t-values-not-reject} The distribution of the $t$-value approximations from selected experiments on synthetic and real-life data when the null hypothesis is (essentially) true}
\end{figure}

\begin{figure}[h]
	\begin{center}
		\begin{subfigure}[h]{0.5\textwidth}
			\includegraphics[scale=\SCALE]{decision-synth-coord3.jpeg}
			\caption{\label{apx_subfig:decision-synth-coord3} Synthetic data, coordinate $\beta_3=0$}
		\end{subfigure}
		\begin{subfigure}[h]{0.5\textwidth}
			\includegraphics[scale=\SCALE]{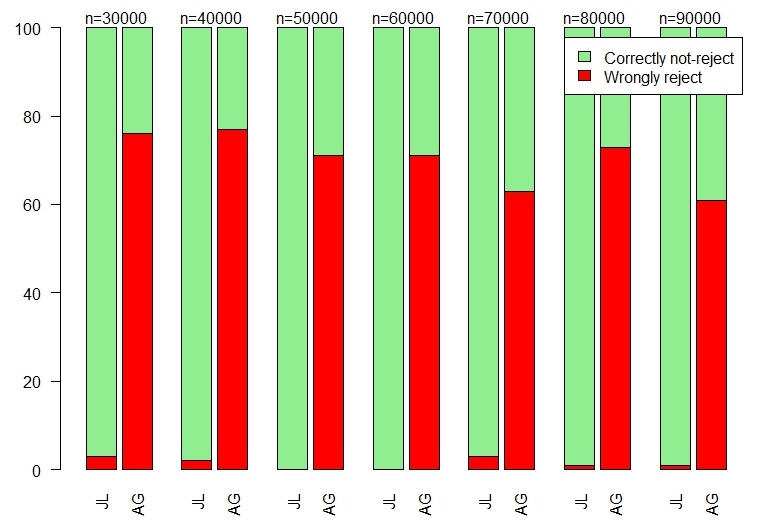}
			\caption{\label{apx_subfig:decision-real-coord2} Real-life data, coordinate $\beta_2=0.57$}
		\end{subfigure}
	\end{center}
	\caption{\label{apx_fig:decision-not-reject} The correctness of our decision to reject the null-hypothesis based on the approximated $t$-value when the null hypothesis is (essentially) true}
\end{figure}

\fi
\end{document}